\documentclass[11pt,a4paper]{article}
\usepackage[a4paper]{geometry}
\usepackage{amssymb,latexsym,amsmath,amsfonts,amsthm,enumerate}
\usepackage{graphicx}
\usepackage{epstopdf}
\usepackage{tikz}
\usepackage{pgflibraryshapes}
\usetikzlibrary{arrows,decorations.markings}
\usepackage{subfigure}
\usepackage{overpic}
\usepackage{fullpage}
\usepackage{comment,verbatim}
\usepackage{hyperref}
\usepackage{color}
\usepackage{mathrsfs}
\usepackage[title,titletoc]{appendix}
\usepackage{ulem}
\usepackage{enumitem}

\DeclareMathOperator{\diag}{diag}
\DeclareMathOperator*{\Res}{Res}

\DeclareMathOperator{\CHF}{CHF}

\newcommand{\R}{\mathbb{R}}
\newcommand{\C}{\mathbb{C}}

\renewcommand{\Re}{\mathrm{Re}\,}
\renewcommand{\Im}{\mathrm{Im}\,}

\renewcommand{\vec}{\mathbf}
\newcommand{\ud}{\,\mathrm{d}}

\newcommand{\Boh}{\mathcal{O}}

\newtheorem{theorem}{Theorem}[section]
\newtheorem{lemma}[theorem]{Lemma}
\newtheorem{proposition}[theorem]{Proposition}

\newtheorem{corollary}[theorem]{Corollary}

\newtheorem{rhp}[theorem]{RH problem}

\theoremstyle{definition}

\theoremstyle{remark}

\newtheorem{remark}[theorem]{Remark}

\numberwithin{equation}{section}

\begin{document}

\title{Gap probability for the hard edge Pearcey process}

\author{Dan Dai\footnotemark[1], ~Shuai-Xia Xu\footnotemark[2] ~and Lun Zhang\footnotemark[3]}

\renewcommand{\thefootnote}{\fnsymbol{footnote}}
\footnotetext[1]{Department of Mathematics, City University of Hong Kong, Tat Chee
Avenue, Kowloon, Hong Kong. E-mail: \texttt{dandai@cityu.edu.hk}}
\footnotetext[2]{Institut Franco-Chinois de l'Energie Nucl\'{e}aire, Sun Yat-sen University,
Guangzhou 510275, China. E-mail: \texttt{xushx3@mail.sysu.edu.cn}}
\footnotetext[3] {School of Mathematical Sciences and Shanghai Key Laboratory for Contemporary Applied Mathematics, Fudan University, Shanghai 200433, China. E-mail: \texttt{lunzhang@fudan.edu.cn }}

\date{\today}

\maketitle

\maketitle

\begin{abstract}
The hard edge Pearcey process is universal in random matrix theory and many other stochastic models. This paper deals with the gap probability for the thinned/unthinned hard edge Pearcey process over the interval $(0,s)$ by working on the relevant Fredholm determinants. We establish an integral representation of the gap probability via a Hamiltonian related a system of coupled differential equations. Together with some remarkable differential identities for the Hamiltonian, we derive the large gap asymptotics for the thinned case, up to and including the constant term. As an application, we also obtain the asymptotic statistical properties of the counting function for the hard edge Pearcey process.
\end{abstract}

\setcounter{tocdepth}{2} \tableofcontents

\section{Introduction}

A fundamental question in random matrix theory is to understand distributions of the eigenvalues for large random matrices in different regimes.  For the local eigenvalue statistics of random Hermitian matrices, it is well-known that they are encoded in the Fredholm determinants of some integral operators \cite{Dyson,Forrester,metha}. Moreover, the associated determinantal point processes \cite{Johansson06,Soshnikov00} are believed to be universal \cite{EY12,Kui11} for a large family of interacting particle systems, which serves as the cornerstone of widespread applications of random matrix theory in different areas ranging from theoretical to practical sciences.

It comes out that many canonical point processes arising from random matrix theory are characterized by integral kernels of the form
\begin{equation}\label{def:Kxy}
K(x,y)=\frac{1}{ (2 \pi i)^2}\int_{\gamma_s}\int_{\gamma_t}\frac{e^{\mathcal{V}(s;x)-\mathcal{V}(t;y)}}{t-s}\ud s \ud t
\end{equation}
for $x,y \in \mathbb{R}$ or $(0,+\infty)$, where $\mathcal{V}(s;x)$ takes a simple form in $s$, the contours $\gamma_s$ and $\gamma_t$ depend on the precise form of $\mathcal{V}$. The most prominent example is the Airy kernel \cite{metha}, which corresponds to $\mathcal{V}(s;x)=s^3/3-xs$. The Airy determinant gives us the celebrated Tracy-Widom distribution \cite{TWAiry}, which describes the limiting distributions of extreme eigenvalues for many random matrix ensembles and share an inherent connection with statistics of other stochastic models beyond the realm of random matrices. One can generalize the Airy kernel by setting $\mathcal{V}$ to be an odd polynomial of the form $$ (-1)^{n+1}\left(\frac{s^{2n+1}}{2n+1}+\sum_{j=1}^{n-1}\rho_j s^{2j+1}\right)-xs, \quad \rho_j\in\mathbb{R}, \quad n=1,2,\ldots. $$ The associated point processes appear in multicritical edge statistics for the momenta of fermions in nonharmonic traps \cite{DMS18} and other statistical physics model \cite{ACV12}. The classical Pearcey kernel corresponds to $\mathcal{V}(s;x)=s^4/4-\rho s^2/2+xs$, $\rho\in \mathbb{R}$, i.e., a quartic polynomial in $s$. Near the points of the spectrum where the density exhibits a cusp-like singularity, the local statistics is usually modeled by the Pearcey process. It originates from Gaussian random matrices with an external source \cite{BK3,BH,BH1} and appears in many other random matrix ensembles \cite{EKS,GZ,HHNb}. It is also possible that the function $\mathcal{V}$ admits some singularities near the origin. By choosing $\mathcal{V}(s;x)=-xs-\alpha \ln s + 1/(4s)$, $\alpha>-1$, we obtain, up to a conjugation term $(y/x)^{\alpha/2}$, the classical Bessel kernel \cite{Forrester93,TWBessel}. The Bessel point process typically describes the smallest eigenvalue distribution near the hard edge 0 for the Wishart-Laguerre ensemble and its unitary invariant generalizations \cite{KV02,Van07}.

Let $\mathcal{K}$ be the integral operator acting on a proper interval $I\in \mathbb{R}$ with kernel $K$. The Fredholm determinant
\begin{equation}
D(I;\gamma):=\det(I-\gamma \mathcal{K}),\quad 0<\gamma\leq 1,
\end{equation}
has attracted great interest over the past few years. On one hand, $D(I;\gamma)$ is essential in understanding the point process determined by $K$. Indeed, $D(I;1)$ can be interpreted as the probability of finding no particles (also known as the gap probability) on the interval $I$ \cite{Johansson06,Soshnikov00}-- a basic object in the theory, while $D(I;\gamma)$, $0<\gamma<1$, is the gap probability for the associated thinned process. Thinning is a classical operation in the studies of point processes \cite{IPSS}, which is obtained from the original one by removing each particle independently with probability $1-\gamma$; we refer to \cite{Boh06,CC17} for more information about thinning in random matrix theory. It is also worth noticing that one can deduce global rigidity upper bounds for determinantal point processes from $D(I;\gamma)$ by using a general result established in \cite{CC20}. On the other hand, the Fredholm determinant $D(I;\gamma)$ alone has rich structures and connections. It is well-known that for the universal sine, Airy and Bessel kernels, the associated determinants admit integral representations in terms of Painlev\'{e} transcendents or the associated Hamiltonians \cite{JM80,TWBessel,TWAiry}. Moreover, although it is impossible to evaluate $D(I;\gamma)$ explicitly, but the relevant asymptotic expansions as $|I|\to \infty$ usually have a simple and elegant form. To derive these large gap asymptotics, however, one needs a deep and strong effort \cite{Kra1}. For the aforementioned various kernels in the form \eqref{def:Kxy}, the asymptotics of $D(I;\gamma)$ with $I$ being a single interval can be found in \cite{BRD08,Bot:Buck2018,DIK2008} for the Airy point process, in \cite{CCG21} for the generalized Airy point process, in \cite{BIP19,Charlier21,dkv,E10} for the Bessel point process, and in \cite{CM,DXZ202,DXZ21} for the Pearcey process.

In this paper, we are concerned with the hard edge Pearcey process characterized by the kernel
(see \cite[Equation (1.19)]{Kui:AMF:Wie2011})
\begin{equation}\label{eq: pearcey kernel}
K_\alpha(x,y;\rho)
:=\frac{1}{ (2 \pi i)^2}\int_{\gamma_s}\int_{\gamma_t}\frac{e^{\rho/s+1/(2s^2)-\rho/t -1/(2t^2)+xs-yt}}{t-s}
\left(\frac{s}{t}\right)^{\alpha}\ud s \ud t, \quad x,y>0,
\end{equation}
that is, $\mathcal{V}(s;x)=xs+\alpha \ln s+ \rho/s+1/(2s^2)$ in \eqref{def:Kxy}, where the parameters $\alpha > -1$, $\rho\in\mathbb{R}$, and the contours $\gamma_s$ and $\gamma_t$ are illustrated in Figure \ref{fig:gammas}. This process was first introduced by
Desrosiers and Forrester for a chiral Gaussian unitary ensemble with a source term \cite{DF08}, where $K_\alpha$ is given in a double integral form involving the Bessel function of the first kind of order $\alpha$. In the context of non-intersecting squared Bessel paths, Kuijlaars, Mart\'{i}nez-Finkelshtein and Wielonsky obtained the one shown in \eqref{eq: pearcey kernel}. More precisely, consider $n$ non-intersecting squared Bessel paths with the same positive starting point and the same ending point taken to be 0. As $n\to\infty$, all paths, after proper scaling, initially stay away from the hard edge 0, but at a certain critical time $t=t^*$, the lowest paths hit the hard edge and are stuck to it from then on. If $t\neq t^*$, the local correlations are described by the classical sine, Airy, and Bessel kernels in random matrix theory \cite{KMW09}, and one obtains the kernel $K_\alpha$ at the hard edge after proper scaling around $t^*$ \cite{Kui:AMF:Wie2011}. The two different formulations of $K_\alpha$ in \cite{DF08} and \cite{Kui:AMF:Wie2011} are later proved to be equivalent by Delvaux and Vet\H{o} in \cite{DV15}. Since the kernel $K_\alpha$ appears near the cusp of the non-intersecting squared Bessel paths model and bears a resemblance to the Pearcey kernel, we call it the hard edge Pearcey kernel, following the terminology initiated in \cite{DV15}. The hard edge Pearcey process is also universal as can be seen from its emergence in various different stochastic models, which include the random surface model \cite{BK10,Cer}, hard edge asymptotics of the Jacobi growth model \cite{CK20}, critical behavior in Angelesco ensembles \cite{DesKui} and a model of non-intersecting Brownian bridges between reflecting or absorbing walls \cite{LWAdv}; see also \cite{DV15,Gir14} for multi-time extensions of the hard edge Pearcey process.

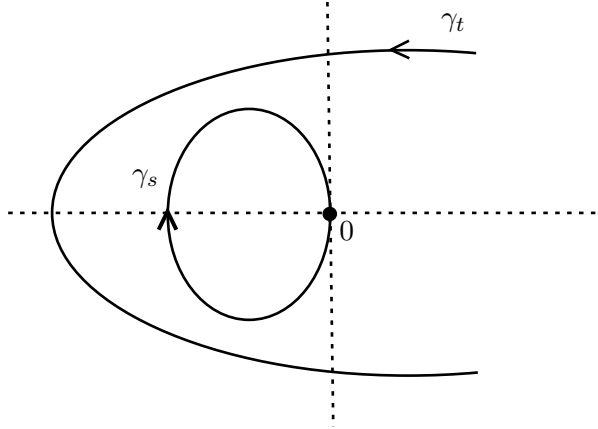
\begin{figure}[h]
\centering
\begin{tikzpicture}[x=0.75pt,y=0.75pt,yscale=-1,xscale=1]

\draw [line width=1]  [dash pattern={on 1.69pt off 2.76pt}]  (171.57,165.96) -- (465.28,165.96) ;
\draw [line width=1]  [dash pattern={on 1.69pt off 2.76pt}]  (330.85,59.86) -- (333.77,273.57) ;
\draw  [draw opacity=0][line width=1]  (405.9,246.23) .. controls (394.63,247.26) and (382.97,247.8) .. (371.03,247.8) .. controls (272.98,247.8) and (193.49,211.16) .. (193.49,165.96) .. controls (193.49,120.75) and (272.98,84.11) .. (371.03,84.11) .. controls (382.67,84.11) and (394.04,84.62) .. (405.05,85.61) -- (371.03,165.96) -- cycle ; \draw  [line width=1]  (405.9,246.23) .. controls (394.63,247.26) and (382.97,247.8) .. (371.03,247.8) .. controls (272.98,247.8) and (193.49,211.16) .. (193.49,165.96) .. controls (193.49,120.75) and (272.98,84.11) .. (371.03,84.11) .. controls (382.67,84.11) and (394.04,84.62) .. (405.05,85.61) ;
\draw  [line width=1]  (251.21,166.71) .. controls (251.21,137.42) and (269.36,113.66) .. (291.76,113.66) .. controls (314.15,113.66) and (332.31,137.42) .. (332.31,166.71) .. controls (332.31,196.01) and (314.15,219.76) .. (291.76,219.76) .. controls (269.36,219.76) and (251.21,196.01) .. (251.21,166.71) -- cycle ;
\draw  [line width=1]  (371.57,79.33) -- (362.5,83.92) -- (371.72,88.17) ;
\draw  [fill={rgb, 255:red, 0; green, 0; blue, 0 }  ,fill opacity=1 ][line width=1]  (332.31,169.67) .. controls (332.31,169.54) and (332.2,169.43) .. (332.06,169.43) .. controls (331.92,169.43) and (331.81,169.54) .. (331.81,169.67) .. controls (331.81,169.81) and (331.92,169.92) .. (332.06,169.92) .. controls (332.2,169.92) and (332.31,169.81) .. (332.31,169.67) -- cycle ;
\draw  [fill={rgb, 255:red, 0; green, 0; blue, 0 }  ,fill opacity=1 ][line width=1]  (329.09,166.48) .. controls (329.09,164.85) and (330.42,163.53) .. (332.06,163.53) .. controls (333.69,163.53) and (335.02,164.85) .. (335.02,166.48) .. controls (335.02,168.11) and (333.69,169.43) .. (332.06,169.43) .. controls (330.42,169.43) and (329.09,168.11) .. (329.09,166.48) -- cycle ;
\draw  [line width=1.5]  (255.49,174.41) -- (251.07,165.26) -- (246.65,174.4) ;

\draw (386.3,64.52) node [anchor=north west][inner sep=0.75pt]   [align=left] {$\gamma_t$};
\draw (232.18,142.75) node [anchor=north west][inner sep=0.75pt]   [align=left] {$\gamma_s$};
\draw (335.68,169.13) node [anchor=north west][inner sep=0.75pt]   [align=left] {0};

\end{tikzpicture}

\caption{The contours $\gamma_s$ and $\gamma_t$ in the definition of $K_\alpha$ \eqref{eq: pearcey kernel}.}
\label{fig:gammas}

\end{figure}

The purpose of this paper to add to the collection of formulas for the very few basic gap probabilities by working on the hard edge Pearcey determinant. Our main results will be stated in the next section.

\section{Main results}
Let $\mathcal{K}_{s, \alpha}$ be the trace class operator acting on $L^2\left(0, s\right)$, $s\geq 0$, with the kernel $K_\alpha$ shown in \eqref{eq: pearcey kernel}. As aforementioned, the Fredholm determinant $\det\left(I-\gamma \mathcal{K}_{s, \alpha} \right)$, $0< \gamma \leq 1$ gives us the gap probability for the thinned/unthinned hard edge Pearcey process over the interval $(0,s)$, which corresponds to $0<\gamma<1$ and $\gamma=1$, respectively. By setting
\begin{equation}\label{def:Fnotation}
F(s;\gamma,\rho):=\ln  \det \left(I-\gamma \mathcal{K}_{s, \alpha}  \right),
\end{equation}
our first result establishes an integral representation of $F$ via a Hamiltonian related a system of 12 coupled differential equations.

The system of differential equations involves 12 unknown functions denoted by
\begin{equation}
p_{i,k}(s),\quad q_{i,k}(s), \quad i=0,1, \quad  k=1,2,3,
\end{equation}
and reads as follows:
\begin{equation}\label{def:sysdiff}
 \left\{
   \begin{array}{ll}
     q_{0,k}'(s)=\frac{1}{s}q_{1,k}(s)\sum_{j=1}^3p_{1,j}(s)q_{0,j}(s),
     \\
     q_{1,k}'(s)=q_{1,k+1}(s)  +\frac{1}{s}q_{0,k}(s)\sum_{j=1}^3p_{0,j}(s)q_{1,j}(s),
     \\
      p_{0,k}'(s)=-\frac{1}{s}p_{1,k}(s)\sum_{j=1}^3p_{0,j}(s)q_{1,j}(s),
     \\
     p_{1,k}'(s)=-p_{1,k-1}(s)  -\frac{1}{s}p_{0,k}(s)\sum_{j=1}^3p_{1,j}(s)q_{0,j}(s),
   \end{array}
 \right.
\end{equation}
where it is assumed that $p_{1,0}=q_{1,4} \equiv 0$. By furthering imposing the conditions
\begin{equation}\label{eq:constraint}
\sum_{k=1}^3q_{0,k}(s)p_{0,k}(s)=-\alpha, \qquad \sum_{k=1}^3q_{1,k}(s)p_{1,k}(s)=0,
\end{equation}
it is readily seen that the function
\begin{equation}\label{def:H}
H(s) := \left(p_{1,1}(s) q_{1,2}(s)+p_{1,2}(s) q_{1,3}(s) \right)
+\frac 1 s \left(\sum_{j=1}^3p_{1,j}(s)q_{0,j}(s)\right)
\left (\sum_{j=1}^3p_{0,j}(s)q_{1,j}(s)\right)
\end{equation}
satisfies the relations
\begin{equation}\label{eq:Heq}
\frac{\partial H}{\partial p_{i,k} }=q'_{i,k}(s), \qquad \frac{\partial H}{\partial q_{i,k} }=-p'_{i,k}(s),
\quad i=0,1,\quad  k=1,2,3.
\end{equation}
Thus, $H$ in \eqref{def:H} is a Hamiltonian for the system \eqref{def:sysdiff} and \eqref{eq:constraint}. The following theorem reveals an elegant relation between $F$ and $H$.

\begin{theorem}\label{thm:integralRep}
With the function $F(s;\gamma,\rho)$ defined in \eqref{def:Fnotation}, we have
\begin{equation}\label{eq: integralRep}
F(s;\gamma,\rho)=\int_0^sH(\tau) \ud \tau, \qquad s\in (0,+\infty),
\end{equation}
where $H(s)$ given in \eqref{def:H} is the Hamiltonian associated with a special solution to the system of equations \eqref{def:sysdiff} and \eqref{eq:constraint}. Moreover, $H(s)$ satisfies the following asymptotics:
as $s \to 0^+$,
\begin{equation}\label{eq:Hasy0}
H(s)=\Boh (s^{\alpha} ), \quad 0<\gamma \leq 1,
\end{equation}
and as $s \to +\infty$,
\begin{equation} \label{thm:H-asy-infty}
H(s)=\frac{\sqrt{3} \beta i}{ s ^{1/3}} - \frac{\sqrt{3}\rho \beta i}{3 s^{2/3}}-\frac{2\beta^2}{3s} -\frac{2\beta i}{ 3 \sqrt{3} s} \cos \left(2\psi(s) \right)+\Boh(s^{-\frac43})
\end{equation}
for $0<\gamma<1$, where
\begin{equation} \label{beta-def}
\beta= \frac{1}{2\pi i} \ln (1-\gamma),
\end{equation}
and
\begin{equation}\label{def:psi}
\psi(s)=\frac{\alpha \pi  }{3} +  \arg \Gamma(1-\beta) -\beta i \left(\frac{2}{3}\ln s+ \ln 9 \right)+ \frac{\sqrt{3}}{2}\left(\rho s^{\frac13}-\frac32 s^{\frac23} \right)
\end{equation}
with $\Gamma$ being Euler's gamma function.
\end{theorem}

The local behavior of $H$ in \eqref{eq:Hasy0} ensures the integral \eqref{eq: integralRep} is well-defined. The existence of a special solution to the system of equations \eqref{def:sysdiff} and \eqref{eq:constraint} follows from the Lax pair for a Riemann-Hilbert (RH) problem. Asymptotic properties of this solution is collected in Proposition \ref{prop:asypq} below, which might be of independent interest. We also note that Girotti found a connection between $F(s;1,\rho)$ and a Lax pair which might be related to a higher order Painlev\'{e} III hierarchy in \cite{Gir14}. It would be interesting to compare \eqref{eq: integralRep} with the result therein and to work out the precise relationship.

The complicated form of $H$ in $p_{i,k}(s)$ and $q_{i,k}(s)$ implies that one cannot evaluate $F$ explicitly for a given value of $s$, it is then  natural to approximate it by establishing the asymptotic formula. We next show the large-$s$ asymptotics of $F(s;\gamma,\rho)$ for $0<\gamma<1$, up to and including the notoriously difficult constant term.
\begin{theorem}\label{thm:FAsy}
With the function $F(s;\gamma,\rho)$ defined in \eqref{def:Fnotation}, we have, as $s\to+\infty$,
\begin{align}
F(s;\gamma,\rho) &=\frac{3\sqrt{3} \beta i}{2} s ^{\frac23} - \sqrt{3} \rho \beta i  s^{\frac13}-\frac{2\beta^2}{3} \ln s \nonumber \\
& ~~~ +\ln\left( G(1+\beta) G(1-\beta) \right) - 2 \beta^2  \ln 3- \frac{2\alpha \beta \pi i}{3}+\Boh(s^{-\frac13}), \qquad 0 \leq \gamma <1,
\label{eq: large gap asy}
\end{align}
uniformly for $\gamma$ and $\rho$ in any compact subset of $[0,1)$ and $\mathbb{R}$, respectively, where $\beta$ is given in \eqref{beta-def} and $G(z)$ is the Barnes G-function.
\end{theorem}

If $\gamma \to 0^+$, we have $\beta \to 0$. Since $G(1)=0$, it is readily seen from \eqref{eq: large gap asy} that $F(s;0,\rho)=\Boh(s^{-1/3})$, which is consistent with the fact that $F(s;0,\rho)=0$. The asymptotics of $F(s;1,\rho)$, however, should exhibit a significantly different asymptotic behavior. Indeed, according to the Forrester-Chen-Eriksen-Tracy conjecture \cite{CET,Forrester93}, it is expected that $F(s;1,\rho)=\Boh(s^{4/3})$ for large $s$. The relevant result will be reported elsewhere.

As an application of Theorem \ref{thm:FAsy}, we are able to extract the asymptotic statistical properties of $N(s)$ -- the counting function of the number of particles in the hard edge Pearcey process falling into the interval $(0,s)$.
\begin{corollary}\label{cor:application}
As $s\to +\infty$, we have
\begin{align}\label{eq:expecvar}
\mathbb{E}(N(s)) &= \mu(s)-\frac{\alpha}{3}+o(1),\\
{\rm Var}(N(s))  &= \sigma(s)^2+\frac{1+\ln9+\gamma_{{\rm E}}}{2\pi^2} + o(1),\label{eq:var}
\end{align}
where $\gamma_{{\rm E}}=-\Gamma'(1)\thickapprox 0.57721$ is Euler's constant,
\begin{equation}\label{def:musigma}
\mu(s)= \frac{3\sqrt{3}}{4\pi}s^{\frac23}-\frac{\sqrt{3}\rho}{4\pi}s^{\frac13},\qquad \sigma(s)^2=\frac{1}{3\pi^2}\ln s.
\end{equation}
Moreover, the random variable $\frac{N(s)-\mu(s)}{\sqrt{\sigma(s)^2}}$ converges in distribution to the normal law $\mathcal{N}(0,1)$ as $s\to +\infty$ and the following upper bound for the maximum fluctuation of the counting function $N(s)$ holds:
\begin{equation}\label{eq:upperbound}
\lim_{a\to\infty} \mathbb{P}\left(\sup_{s>a}\left |\frac{N(s)-\left(\frac{3\sqrt{3}}{4\pi}s^{2/3}-\frac{\sqrt{3}\rho}{4\pi}s^{1/3}\right)}{\ln s}\right | \leq \frac{2}{3\pi} +\epsilon \right)=1,\qquad \epsilon>0.
\end{equation}
\end{corollary}
\begin{proof}
It is well-known that the moment generating function of $N(s)$ is equal to a specified deformed hard edge Pearcey determinant:
\begin{equation}\label{eq: M-F}
\mathbb{E}\left(e^{-2\pi \nu N(s)}\right)=\sum_{k=0}^{\infty}\mathbb{P}(N(s)=k)e^{-2 \pi \nu k}= \det\left(I-(1-e^{-2\pi \nu}) \mathcal{K}_{s,\alpha}\right), \qquad \nu \geq 0;
\end{equation}
cf. \cite{Johansson06,Soshnikov00}. Since
\begin{equation}\label{eq:momentexpan}
\mathbb{E}\left(e^{-2\pi \nu N(s)}\right)=1-2\pi \mathbb{E}(N(s))\nu+2\pi^2 \mathbb{E}(N(s)^2 )\nu^2+\Boh(\nu^3), \quad \nu \to 0,
\end{equation}
one has
\begin{align}\label{eq:momentexpan2}
F(s;1-e^{-2\pi \nu},\rho) & = \log \mathbb{E}\left(e^{-2\pi \nu N(s)}\right)
\nonumber
\\
&=-2\pi \mathbb{E}(N(s))\nu+2\pi^2 {\rm Var}(N(s))\nu^2+ \Boh(\nu^3), \quad \nu \to 0,
\end{align}
where $F$ is defined in \eqref{def:Fnotation}. In view of \eqref{eq: large gap asy}, we have, as $s\to +\infty$,
\begin{multline}\label{eq:MomentGExpansion}
F(s;1-e^{-2\pi \nu},\rho)
=-2\pi\left(\mu(s)-\frac{\alpha}{3}\right) \nu + 2 \pi^2\left(\sigma(s)^2+\frac{\ln 3}{\pi^2}\right)\nu^2
\\+\ln \left(G(1+\nu i )G(1-\nu i)\right) +\Boh(s^{-\frac13}),
\end{multline}
where the functions $\mu(s)$ and $\sigma(s)^2$ are defined in \eqref{def:musigma}. Note that
\begin{equation}
G(1+z)=1+\frac{\ln(2\pi)-1}{2}z+\left(\frac{(\ln(2\pi)-1)^2}{8}-\frac{1+\gamma_{\textrm{E}}}{2}\right)z^2+\Boh(z^3),\qquad z\to 0,
\end{equation}
where $\gamma_{{\rm E}}$ is Euler's constant, we then obtain \eqref{eq:expecvar} and \eqref{eq:var} from \eqref{eq:momentexpan2} and \eqref{eq:MomentGExpansion}.

Finally, since $\sigma(s)^2\to +\infty$ for large positive $s$, it is easily seen from \eqref{eq: M-F} and \eqref{eq:MomentGExpansion} that
\begin{equation}
\mathbb{E}\left(e^{t\cdot \frac{N(s)-\mu(s)}{\sqrt{\sigma(s)^2}}}\right) \to e^{\frac{t^2}{2}}, \qquad s\to+\infty,
\end{equation}
which implies the convergence of $\frac{N(s)-\mu(s)}{\sqrt{\sigma(s)^2}}$ in distribution to the normal law $\mathcal{N}(0,1)$. The upper bound
\eqref{eq:upperbound} follows directly by combining \cite[Lemma 2.1 and Theorem 1.2]{CC20}, \eqref{eq:momentexpan2} and \eqref{eq:MomentGExpansion}.

This completes the proof of Corollary \ref{cor:application}.
\end{proof}

According to \cite{SDMS1,SDMS2}, it is believed that ${\rm Var}(N(s))$ takes a universal asymptotic form $c_0 \ln s + c_1 +o(1)$ for many determinantal point processes in random matrix theory, where the constants $c_0$ and $c_1$ vary for different processes. This has been confirmed for the classical sine, Airy, Bessel and Pearcey processes; cf. \cite{CharlierPlms,Charlier21,DXZ202,SosJSP}. By \eqref{eq:var}, the same is true for the hard edge Pearcey process. In addition, the upper bound \eqref{eq:upperbound} particularly means that, for large $s$, the counting function of hard edge Pearcey process lies in the interval $(\mu(s)-(2 /(3\pi) + \epsilon)\ln s,\mu(s)+ (2 /(3\pi) +\epsilon)\ln s)$ with high probability.

The rest of this paper is devoted to the proofs of our main results. The idea follows from a general framework established in \cite{BD02,DIK2008}, which gives a connection between the Fredholm determinants of integrable kernels and RH problems. The precise relationship in our case is
discussed in Section \ref{sec:RHP}, where we connects $\ud F(s;\gamma,\rho)/\ud s$ to a $3 \times 3$ RH problem for $\Phi$ with constant jumps. We then derive a Lax pair for $\Phi$ and several useful differential identities for the Hamiltonian in Section \ref{sec:Lax}. After performing a Deift-Zhou steepest descent analysis \cite{DZ93} on the RH problem for $\Phi$ as $s \to \infty$ and $s \to 0^+$ in Sections \ref{sec:AsyPhiinfty} and \ref{sec:AsyPhi0}, we present asymptotics of the relevant solution to the system of equations \eqref{def:sysdiff} and \eqref{eq:constraint}, and proofs of Theorems \ref{thm:integralRep} and \ref{thm:FAsy} in Sections \ref{sec:Asypq} and \ref{sec:proofmainresults}, respectively.

\section{An RH formulation} \label{sec:RHP}

\subsection{RH characterizations of $K_\alpha$ and its resolvent kernel}

As illustrated in \cite[RH problem 1.3 and Lemma 5.3]{Kui:AMF:Wie2011}, the hard edge Pearcey kernel $K_\alpha$ in \eqref{eq: pearcey kernel} is also characterized by the following $3 \times 3$ RH problem.

\begin{rhp} \label{rhp: Pearcey}
\hfill 
\begin{itemize}
\item[\rm (1)] $\Psi_\alpha(z)=\Psi_\alpha(z;\rho)$ is analytic in $\C \setminus
\{\cup_{j=0}^5\Sigma_j \cup \{ 0 \} \}$,
where
\begin{equation}\label{def:sigmai}
\begin{aligned}
&\Sigma_0=(0,+\infty), \qquad \Sigma_1=e^{\frac{\pi i}{4}}(0,+\infty), \quad &&\Sigma_2=e^{\frac{ 3 \pi i}{4}}(0,+\infty),
\\
&\Sigma_3=(-\infty, 0), \qquad \Sigma_4=e^{-\frac{3\pi i}{4}}(0,+\infty), \quad &&\Sigma_5=e^{-\frac{\pi i}{4}}(0,+\infty),
\end{aligned}
\end{equation}
with the orientations as shown in Figure \ref{fig:Pearcey}.

\begin{figure}[h]
\begin{center}
   \setlength{\unitlength}{1truemm}
   \begin{picture}(100,70)(-5,2)
       \put(40,40){\line(-1,-1){20}}
       \put(40,40){\line(-1,1){20}}
       \put(40,40){\line(-1,0){30}}
       \put(40,40){\line(1,0){30}}
       \put(40,40){\line(1,1){20}}
       \put(40,40){\line(1,-1){20}}

       \put(30,50){\thicklines\vector(1,-1){1}}
       \put(30,40){\thicklines\vector(1,0){1}}
       \put(30,30){\thicklines\vector(1,1){1}}
       \put(50,50){\thicklines\vector(1,1){1}}
       \put(50,40){\thicklines\vector(1,0){1}}
       \put(50,30){\thicklines\vector(1,-1){1}}

       \put(39,36.3){$0$}

       \put(27,22){$\Sigma_4$}
       \put(-7,18) {$\begin{pmatrix} 1&0&0 \\ 0&1&e^{- \alpha \pi i} \\ 0&0&1 \end{pmatrix}$}

       \put(27,55){$\Sigma_2$}
       \put(-7,66){$\begin{pmatrix} 1&0&0 \\ 0&1& e^{\alpha \pi i} \\ 0&0&1 \end{pmatrix}$}

       \put(14,42){$\Sigma_3$}
       \put(-28,40){$\begin{pmatrix} 1&0&0 \\ 0&0& - e^{-\alpha \pi i} \\ 0&e^{-\alpha \pi i}&0 \end{pmatrix}$}

       \put(48,22){$\Sigma_5$}
       \put(60,18){$\begin{pmatrix} 1&0&0 \\ 1&1&0 \\ 0&0&1 \end{pmatrix}$}

       \put(48,55){$\Sigma_1$}
       \put(60,66) {$\begin{pmatrix} 1&0&0 \\ 1&1&0 \\ 0&0&1 \end{pmatrix}$}

       \put(58,42){$\Sigma_0$}
       \put(72,40){ $\begin{pmatrix} 0&1&0 \\ -1&0&0 \\ 0&0&1 \end{pmatrix}$ }

       \put(40,40){\thicklines\circle*{1}}

 \end{picture}
   \vspace{-10mm}
   \caption{The jump contours $\Sigma_k$, $k=0,\ldots,5$, and the corresponding jump matrix $J_{\Psi_{\alpha}}$ in the RH problem for $\Psi_\alpha$.}
   \label{fig:Pearcey}
\end{center}
\end{figure}
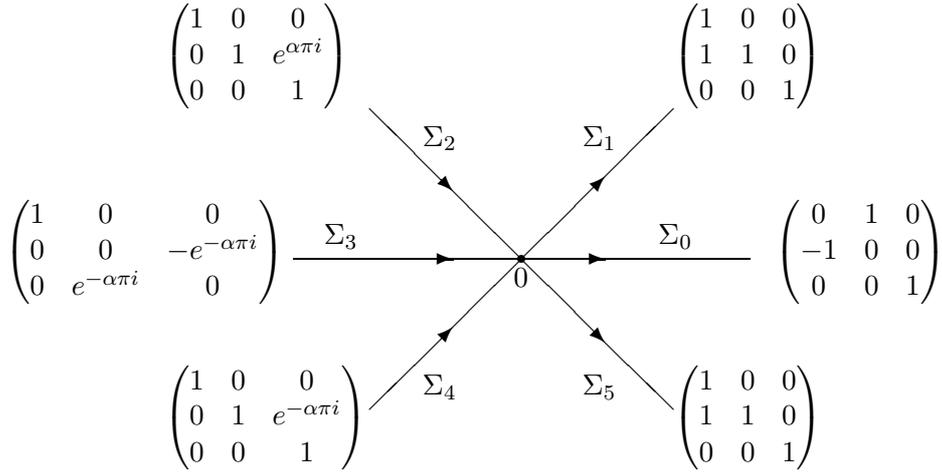

\item[\rm (2)] For $z\in \Sigma_k$, $k=0,1,\ldots,5$, the limiting values of $\Psi_\alpha$ exist and satisfy the jump condition
\begin{equation}\label{jumps:M}
\Psi_{\alpha,+}(z) = \Psi_{\alpha, -}(z)J_{\Psi_{\alpha}}(z),
\end{equation}
where $J_{\Psi_{\alpha}}(z)$ is depicted in Figure \ref{fig:Pearcey}.

\item[\rm (3)]
As $z \to \infty$ and $\pm\Im z>0$, we have
\begin{multline}\label{eq:asyPsi}
 \Psi_\alpha(z) = \frac{i z^{-\frac{\alpha}{3}}}{ \sqrt{3} }  C_\Psi \left(I+ \Boh(z^{-1} ) \right) \diag \left(z^{\frac13},1,z^{-\frac13} \right)L_{\pm} \diag \left( e^{\pm \frac{\alpha \pi i}{3}}, e^{\mp \frac{\alpha \pi i}{3}}, 1 \right) e^{\Theta(z)},
\end{multline}
where the constant matrix $C_\Psi$ is defined in \eqref{eq:cons-C-Psi} below, $L_{\pm}$ are given by
\begin{align}\label{def:Lpm}
L_{+}=
\begin{pmatrix}
\omega & \omega^2 & 1 \\ 1&1&1 \\ \omega^2 & \omega & 1
\end{pmatrix},
\qquad
L_{-}=
\begin{pmatrix}
\omega^2 & -\omega & 1 \\ 1&-1&1 \\ \omega & -\omega^2 & 1
\end{pmatrix},
\end{align}
with $\omega=e^{\frac{2\pi i}{3}}$, and
\begin{align}\label{def:Theta}
\Theta(z)=\Theta(z;\rho)&= \begin{cases}
\diag (\theta_1(z;\rho),\theta_2(z;\rho),\theta_3(z;\rho)), & \text{$\Im z >0$,} \\
\diag (\theta_2(z;\rho),\theta_1(z;\rho),\theta_3(z;\rho)), & \text{$\Im z <0$,} \\
\end{cases}
\end{align}
with
\begin{equation} \label{eq: theta-k-def}
\theta_k(z;\rho)=\frac32 \omega^{2k}z^{\frac23}+\rho\omega^kz^{\frac13}, \qquad k=1,2,3.
\end{equation}

\item[\rm (4)]
As $z\to 0$, we have
\begin{equation}\label{eq:Psizero}
\Psi_\alpha(z)=\Psi_\alpha^{(0)}(z)\left\{
                                     \begin{array}{ll}
                                      \begin{pmatrix}
1 & \frac{i}{2\pi}\ln z & 0
\\
0 & z^{-\alpha} & 0
\\
0 & \frac{(-1)^{\alpha+1}}{2\pi}i\ln z & 1
                                      \end{pmatrix}, & \hbox{if $\alpha \in \mathbb{N}\cup \{0\}$,}  \vspace{2mm}
\\
 \begin{pmatrix}
1 & \frac{e^{-\alpha \pi i}}{2\sin(\alpha \pi)}i & 0
\\
0 & z^{-\alpha} & 0
\\
0 & -\frac{i}{2\sin(\alpha \pi)} & 1
                                       \end{pmatrix}, & \hbox{if $\alpha \notin \mathbb{Z}$,}
\end{array}
                                   \right.
\end{equation}
for $\pi/4<\arg z<3\pi /4$, where $\Psi_\alpha^{(0)}(z)$ is analytic at $z=0$. The local behaviors of $\Psi_\alpha(z)$ near $z=0$ in other regions can be determined through the above formula and the jump condition \eqref{jumps:M}.

\end{itemize}
\end{rhp}

It is shown in \cite[Theorem 1.4 and Proposition 5.2]{Kui:AMF:Wie2011} that the above RH problem has a unique solution with
\begin{equation}\label{eq:detPsi}
\det \Psi_\alpha(z) = z^{-\alpha}, \qquad z \in \C \setminus
\{\cup_{j=0}^5\Sigma_j \cup \{ 0 \} \}.
\end{equation}
Moreover, the solution is explicitly given by
\begin{align}
\Psi_\alpha(z) &= \frac{e^{\rho^2/6}} {\sqrt{2 \pi}} \begin{pmatrix}
-p_4(z) & p_3(z) & p_1(z) \\
-p'_4(z) & p'_3(z) & p'_1(z)  \\
-p''_4(z) & p''_3(z) & p''_1(z)
\end{pmatrix}, \qquad 0<\arg z <\frac{\pi}{4}, \label{Psi-in-I} \\
\Psi_\alpha(z) &= \frac{e^{\rho^2/6}} {\sqrt{2 \pi}}
\begin{pmatrix}
p_2(z) & p_3(z) & p_1(z) \\
p'_2(z) & p'_3(z) & p'_1(z)  \\
p''_2(z) & p''_3(z) & p''_1(z)
\end{pmatrix}, \qquad \frac{\pi}{4}<\arg z <\frac{3\pi}{4}, \label{Psi-in-II} \\
\Psi_\alpha(z) &= \frac{e^{\rho^2/6}} {\sqrt{2 \pi}} \begin{pmatrix}
p_2(z) & p_3(z) & -e^{-\alpha \pi i} \, p_4(z) \\
p'_2(z) & p'_3(z) & -e^{-\alpha \pi i} \, p'_4(z)  \\
p''_2(z) & p''_3(z) & -e^{-\alpha \pi i} \, p''_4(z)
\end{pmatrix}, \qquad \frac{3\pi}{4}<\arg z <\pi,  \\
\Psi_\alpha(z) &= \frac{e^{\rho^2/6}} {\sqrt{2 \pi}} \begin{pmatrix}
p_2(z) & p_4(z) & e^{\alpha \pi i} \, p_3(z) \\
p'_2(z) & p'_4(z) & e^{\alpha \pi i} \, p'_3(z)  \\
p''_2(z) & p''_4(z) & e^{\alpha \pi i} \, p''_3(z)
\end{pmatrix}, \qquad -\pi<\arg z < -\frac{3\pi}{4}, \\
\Psi_\alpha(z) &= \frac{e^{\rho^2/6}} {\sqrt{2 \pi}} \begin{pmatrix}
p_2(z) & p_4(z) & p_1(z) \\
p'_2(z) & p'_4(z) & p'_1(z)  \\
p''_2(z) & p''_4(z) & p''_1(z)
\end{pmatrix}, \qquad -\frac{3\pi}{4}<\arg z <-\frac{\pi}{4}, \label{Psi-in-V}
\\
\Psi_\alpha(z) &= \frac{e^{\rho^2/6}} {\sqrt{2 \pi}} \begin{pmatrix}
p_3(z) & p_4(z) & p_1(z) \\
p'_3(z) & p'_4(z) & p'_1(z)  \\
p''_3(z) & p''_4(z) & p''_1(z)
\end{pmatrix}, \qquad  -\frac{\pi}{4}< \arg z < 0. \label{Psi-in-VI}
\end{align}
In the above expressions, $p_1(z)$ and $p_2(z)$ are entire functions in $z$, while $p_3(z)$ and $p_4(z)$ are analytic for $z \in \C \setminus i\R_-$ and $z \in \C \setminus i\R_+$, respectively. They are solutions to the third-order differential equation
\begin{equation}
x y'''(x) + \alpha y''(x) - \rho y'(x) - y(x) =0, \label{eq:Pearcey1}
\end{equation}
and any three of $p_k(z)$, $k=1,\ldots,4$, are linearly independent. In addition,  we have the following integral representations for $p_k$:
\begin{equation}\label{def:pkint}
p_k(z) =\left\{
          \begin{array}{ll}
            \int_{\Gamma_k} t^{-1} e^{zt+V(t)} \ud t, & \hbox{$k=1$,} \\
            e^{-\alpha \pi i}\int_{\Gamma_k} t^{-1} e^{zt+V(t)} \ud t, & \hbox{$k=2,3$,} \\
            e^{\alpha \pi i}\int_{\Gamma_k} t^{-1} e^{zt+V(t)} \ud t, & \hbox{$k=4$,}
          \end{array}
        \right.
\end{equation}
where
\begin{equation}\label{def:V}
V(t) = (\alpha - 2) \ln t + \frac{\rho}{t} + \frac{1}{2t^2}.
\end{equation}
Note that $t^{-1} e^{zt+V(t)}=t^{\alpha-3}e^{zt+\frac{\rho}{t}+\frac{1}{2t^2}}$, we choose $t^{\alpha}$ with $-\frac{\pi}{2}< \arg t < \frac{\pi}{2}$ in the definition of $p_1$, with $\frac{\pi}{2}< \arg t < \frac{3\pi}{2}$ in the definition of $p_2$, with $0< \arg t < \pi$ in the definition of $p_3$ and with $-\pi < \arg t < 0$ in the definition of $p_4$. The contours $\Gamma_k$, $k=1,2,3,4$, in \eqref{def:pkint} are defined as follows: $\Gamma_1$ is a simple closed contour lying in the right half plane, which is tangent to the imaginary axis at the origin, $\Gamma_2$ is the reflection of $\Gamma_1$ with respect to the imaginary axis, $\Gamma_3$ lies in the upper half plane which starts at infinity with an angle where $\Re (zt)<0$ as $t\to \infty$ and ends at the origin along the positive imaginary axis and $\Gamma_4$ lies in the lower half plane which starts at infinity with an angle where $\Re (zt)<0$ as $t\to \infty$ and ends at the origin along the negative imaginary axis; see Figure \ref{fig: Gamma-k} for an illustration.

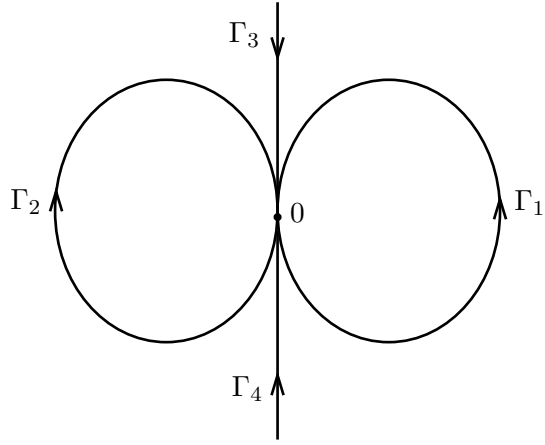
\begin{figure}[t]
\centering

\begin{tikzpicture}[x=0.75pt,y=0.75pt,yscale=-1,xscale=1]

\draw  [line width=1]  (215,147) .. controls (215,110.55) and (239.85,81) .. (270.5,81) .. controls (301.15,81) and (326,110.55) .. (326,147) .. controls (326,183.45) and (301.15,213) .. (270.5,213) .. controls (239.85,213) and (215,183.45) .. (215,147) -- cycle ;
\draw  [line width=1]  (326,147) .. controls (326,110.55) and (350.85,81) .. (381.5,81) .. controls (412.15,81) and (437,110.55) .. (437,147) .. controls (437,183.45) and (412.15,213) .. (381.5,213) .. controls (350.85,213) and (326,183.45) .. (326,147) -- cycle ;
\draw [line width=1]    (326,42) -- (326,147) ;
\draw [line width=1]    (326,157) -- (326,262) ;
\draw  [line width=1]  (329,59.67) -- (326,69.67) -- (323,59.67) ;
\draw  [line width=1]  (440,151.17) -- (437,141.17) -- (434,151.17) ;
\draw  [line width=1]  (329,239.67) -- (326,229.67) -- (323,239.67) ;
\draw  [line width=1]  (218.5,147.17) -- (215.5,137.17) -- (212.5,147.17) ;
\draw  [fill={rgb, 255:red, 0; green, 0; blue, 0 }  ,fill opacity=1 ][line width=1.5]  (327,150.08) .. controls (327,149.53) and (326.55,149.08) .. (326,149.08) .. controls (325.45,149.08) and (325,149.53) .. (325,150.08) .. controls (325,150.64) and (325.45,151.08) .. (326,151.08) .. controls (326.55,151.08) and (327,150.64) .. (327,150.08) -- cycle ;

\draw (191.5,135) node [anchor=north west][inner sep=0.75pt]   [align=left] {$\Gamma_{2}$};
\draw (302,229) node [anchor=north west][inner sep=0.75pt]   [align=left] {$\Gamma_{4}$};
\draw (300.5,51) node [anchor=north west][inner sep=0.75pt]   [align=left] {$\Gamma_{3}$};
\draw (443,135.5) node [anchor=north west][inner sep=0.75pt]   [align=left] {$\Gamma_{1}$};
\draw (331,142) node [anchor=north west][inner sep=0.75pt]   [align=left] {0};
\end{tikzpicture}
\caption{The contour of integration $\Gamma_k$ in the definition of $p_k(z)$, $k=1,2,3,4$.}
\label{fig: Gamma-k}
\end{figure}

\begin{remark}
Comparing the above RH problem with \cite[RH problem 1.3]{Kui:AMF:Wie2011}, the large-$z$ behavior in item (3) is slightly modified for brevity, while the local behavior near the origin in item (4) is expressed in details to facilitate our future discussions. The constant matrix $ C_\Psi$ in \eqref{eq:asyPsi} is given by
\begin{equation} \label{eq:cons-C-Psi}
C_\Psi =
\begin{pmatrix}
1 & \pi_3(\rho) & \pi_6(\rho)
\\
0 & 1 & \pi_3(\rho) + \frac{\rho}{3}
\\
0 & 0 & 1
\end{pmatrix},
\end{equation}
where
\begin{align}
\pi_3(\rho) &= \frac{\rho(\rho^2+9\alpha-18)}{27}, \label{pi3-def} \\
\pi_6(\rho)& = \frac{\rho^6 +(18 \alpha- 45) \rho^4 + (81 \alpha^2 -
    405 \alpha + 405) \rho^2  - 243 \alpha^2 + 729 \alpha - 405}{2 \cdot 3^6},
     \label{pi6-def}
\end{align}
are two polynomials in $\rho$. The derivation of $C_\Psi$ follows directly the explicit solution of $\Psi_\alpha$ and an asymptotic analysis of the relevant integrals.
\end{remark}

By \cite[Equation (8.20)]{Kui:AMF:Wie2011}, it follows that
\begin{equation}\label{eq:PearceyRH}
K_\alpha(x,y;\rho)=\frac{1}{2 \pi i(x-y)}
\begin{pmatrix}
0 & 1 & 0
\end{pmatrix}
 \widetilde\Psi_{\alpha}(y;\rho)^{-1} \widetilde\Psi_{\alpha}(x;\rho)
\begin{pmatrix}
1
\\
0
\\
0
\end{pmatrix}, \qquad x,y >0,
\end{equation}
where
\begin{equation} \label{eq: tilde-psi}
\widetilde\Psi_\alpha(z)=\widetilde\Psi_\alpha(z;\rho): =  \frac{e^{\frac{\rho^2}{6}}} {\sqrt{2 \pi}}
\begin{pmatrix}
p_2(z) & p_3(z) & p_1(z) \\
p'_2(z) & p'_3(z) & p'_1(z)  \\
p''_2(z) & p''_3(z) & p''_1(z)
\end{pmatrix}, \qquad z\in \mathbb{C}\setminus i\mathbb{R}_-,
\end{equation}
is the analytic extension of the restriction of $\Psi_\alpha$ on the region bounded by $\Sigma_1$ and $\Sigma_2$ in \eqref{Psi-in-II} to $\C \setminus i \R_-$. Then, one readily sees that
\begin{equation}\label{eq:tildeKdef}
\gamma K_\alpha (x,y;\rho) = \frac{\vec{f}(x)^t\vec{h}(y)}{x-y},
\end{equation}
where
\begin{equation}\label{def:fh}
\vec{f}(x)=\begin{pmatrix}
f_1
\\
f_2
\\
f_3
\end{pmatrix}:= \widetilde \Psi_{\alpha}(x)
\begin{pmatrix}
1
\\
0
\\
0
\end{pmatrix}, \qquad
\vec{h}(y)=\begin{pmatrix}
h_1
\\
h_2
\\
h_3
\end{pmatrix}
:=
\frac{\gamma}{2 \pi i}
\widetilde \Psi_{\alpha}(y)^{-t} \begin{pmatrix}
0
\\
1
\\
0
\end{pmatrix}.
\end{equation}
It is clear that the kernel $\gamma K_\alpha (x,y;\rho)$ is integrable for all $0< \gamma \leq 1$ in the sense of Its et al. \cite{IIKS90}, which is crucial for the subsequent analysis of $F$  defined in \eqref{def:Fnotation}.

First, it is known that the derivative of $F$ satisfies the following relation:
\begin{align}
\frac{\ud}{\ud s}F(s;\gamma,\rho)&=\frac{\ud}{\ud s} \ln \det(I-\gamma \mathcal{K}_{s,\alpha})
=-\textrm{tr}\left((I-\gamma \mathcal{K}_{s,\alpha} )^{-1} \gamma
\frac{\ud}{\ud s}\mathcal{K}_{s,\alpha} \right) \nonumber \\
&=-R(s,s), \label{eq:derivatives}
\end{align}
where $R(u,v)$ stands for the kernel of the resolvent operator, that is, the operator defined by
$$\left(I-\gamma \mathcal{K}_{s,\alpha} \right)^{-1}-I=\gamma \mathcal{K}_{s,\alpha} \left(I-\gamma \mathcal{K}_{s,\alpha} \right)^{-1}=\gamma \left(I-\gamma \mathcal{K}_{s,\alpha} \right)^{-1}\mathcal{K}_{s,\alpha} .$$
As $\gamma K_\alpha (x,y;\rho) $ is integrable, its resolvent kernel is integrable as well; cf. \cite{DIZ97,IIKS90}. More precisely, let $\vec{F}$ and $\vec{H}$ be defined as
\begin{equation}\label{def:FH}
\vec{F}(u)=
\begin{pmatrix}
F_1 \\
F_2 \\
F_3
\end{pmatrix}:=\left(I-\gamma \mathcal{K}_{s,\alpha} \right)^{-1}\vec{f}, \qquad \vec{H}(v)=\begin{pmatrix}
H_1 \\
H_2 \\
H_3
\end{pmatrix}
:=\left(I-\gamma \mathcal{K}_{s,\alpha} \right)^{-1}\vec{h},
\end{equation}
we then have
\begin{equation}\label{eq:resolventexpli}
R(u,v)=\frac{\vec{F}(u)^t\vec{H}(v)}{u-v}.
\end{equation}
Next, due to the integrable structure of the kernel $\gamma K_\alpha (x,y;\rho)$, the following fundamental RH problem relates $\vec{f},\vec{h}$ to $\vec{F},\vec{H}$, respectively; cf. \cite{DIZ97}.

\begin{rhp}  \label{rhp:Y}
\hfill
\begin{enumerate}
\item[\rm (1)] $Y(z)$ is defined and analytic in $\mathbb{C}\setminus [0,s]$, where the orientation is taken from the left to the right.

\item[\rm (2)] For $x\in(0,s)$, we have
\begin{equation}\label{eq:Y-jump}
 Y_+(x)=Y_-(x)(I-2\pi i \vec{f}(x)\vec{h}(x)^t),
 \end{equation}
with the functions $\vec{f}$ and $\vec{h}$  defined in \eqref{def:fh}.
\item[\rm (3)] As $z \to \infty$, we have
\begin{equation}\label{eq:Y-infty}
 Y(z)=I+\frac{\mathsf{Y}_1}{z}+\mathcal \Boh(z^{-2}).
 \end{equation}

\item[\rm (4)] As $z \to  0$, we have
\begin{equation} \label{eq:y-zero}
Y(z) = \left\{
           \begin{array}{ll}
  \Boh( \ln z), & \hbox{$\alpha \geq 0$,}  \\          

              \Boh(z^\alpha ), & \hbox{$-1<\alpha<0$.} 
           \end{array}
         \right.
\end{equation}

\item[\rm (5)]As $z \to  s$, we have $Y(z) = \mathcal \Boh(\ln(z - s))$.

\end{enumerate}
\end{rhp}
The unique solution to the above RH problem is given by
\begin{equation}\label{eq:Yexpli}
Y(z)=I-\int_{0}^s\frac{\vec{F}(w)\vec{h}(w)^t}{w-z}\ud w,
\end{equation}
and one has
\begin{equation}\label{def:FH2}
\vec{F}(z)=Y(z)\vec{f}(z), \qquad \vec{H}(z)=Y(z)^{-t}\vec{h}(z).
\end{equation}

\begin{remark}
Since $\det(I-\gamma \mathcal{K}_{s,\alpha} )$ stands for the gap probability for the thinned/unthinned hard edge Pearcey process, it must be strictly positive and invertible. This guarantees the solvability of the above RH problem for $Y$.
\end{remark}

\subsection{An RH problem related to $\frac{\ud }{\ud s}F$}
With the aid of RH problems \ref{rhp: Pearcey} and \ref{rhp:Y} for $\Psi_{\alpha}$ and $Y$, we will construct a new one with constant jumps and establish its connection with $\frac{\ud }{\ud s}F$. This is a standard procedure known as the undressing transformation in the study of Fredholm determinants associated with integrable kernels.
The idea is to incorporate the jump condition of $Y$ on $(0,s)$ in the RH problem \ref{rhp: Pearcey} for $\Psi_{\alpha}$, and replace the rays $\Sigma_1$ and $\Sigma_5$ by their parallel rays $\Sigma_1^{(s)}$ and $\Sigma_5^{(s)}$ emanating from the point $s$; see Figure \ref{fig:X}.


\begin{figure}[h]
\begin{center}
   \setlength{\unitlength}{1truemm}
   \begin{picture}(100,70)(-5,2)
       \put(25,40){\line(-1,0){30}}
       \put(55,40){\line(1,0){30}}

       \put(25,40){\line(1,0){30}}

       \put(25,40){\line(-1,-1){25}}
       \put(25,40){\line(-1,1){25}}

       \put(55,40){\line(1,1){25}}
       \put(55,40){\line(1,-1){25}}

       \put(15,40){\thicklines\vector(1,0){1}}
       \put(65,40){\thicklines\vector(1,0){1}}

       \put(10,55){\thicklines\vector(1,-1){1}}
       \put(10,25){\thicklines\vector(1,1){1}}
       \put(70,25){\thicklines\vector(1,-1){1}}
       \put(70,55){\thicklines\vector(1,1){1}}

       \put(-2,11){$\Sigma_4$}

       \put(-2,67){$\Sigma_2$}
       \put(3,42){$\Sigma_3$}
       \put(80,11){$\Sigma_5^{(s)}$}
       \put(80,67){$\Sigma_1^{(s)}$}
       \put(73,42){$\Sigma_0^{(s)}$}

       \put(10,46){$\texttt{III}$}
       \put(10,34){$\texttt{IV}$}
       \put(68,46){$\texttt{I}$}
       \put(68,34){$\texttt{VI}$}
       \put(38,55){$\texttt{II}$}
       \put(38,20){$\texttt{V}$}

       \put(25,40){\thicklines\circle*{1}}
       \put(55,40){\thicklines\circle*{1}}

       \put(24,36.3){$0$}
       \put(54,36.3){$s$}

   \end{picture}
   \caption{Regions $\texttt{I-VI}$ and the jump contours for the RH problem for $\Phi$.}
   \label{fig:X}
\end{center}
\end{figure}
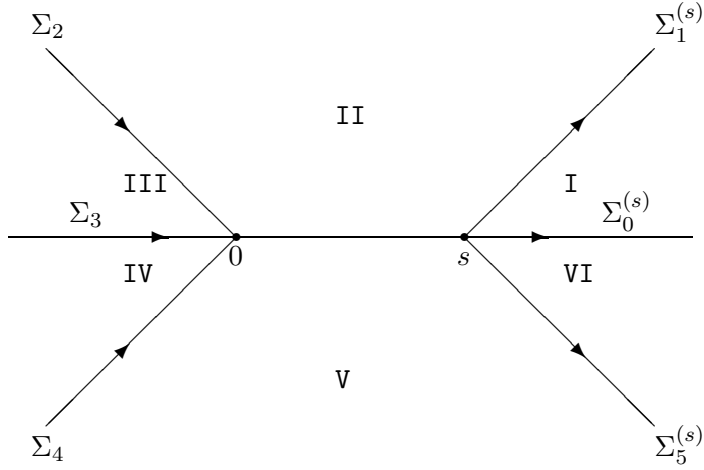

Let \texttt{I}--\texttt{VI} be six regions in the complex plane illustrated in Figure \ref{fig:X}, we now define $\Phi(z)=\Phi(z;s)$ as follows:
\begin{align}\label{eq:YtoX}
 \Phi(z) =  \begin{cases}
   Y(z)\Psi_{\alpha}(z), & \hbox{ $z \in \texttt{I}\cup \texttt{III}\cup \texttt{IV} \cup \texttt{VI}$,} \\
         Y(z)\widetilde \Psi_{\alpha}(z), & \hbox{ $z \in \texttt{II}$,} \\
Y(z)\widehat \Psi_{\alpha}(z), & \hbox{ $z \in \texttt{V}$,}
 \end{cases}
\end{align}
where $\widetilde \Psi_{\alpha}(z)$ is given in \eqref{eq: tilde-psi} and
\begin{equation} \label{eq: hat-psi}
\widehat\Psi_\alpha(z) :=  \frac{e^{\frac{\rho^2}{6}}} {\sqrt{2 \pi}} \begin{pmatrix}
p_2(z) & p_4(z) & p_1(z) \\
p'_2(z) & p'_4(z) & p'_1(z)  \\
p''_2(z) & p''_4(z) & p''_1(z)
\end{pmatrix},\qquad z\in \C \setminus i \R_+,
\end{equation}
i.e., the analytic extension of the restriction of $\Psi_\alpha$ on the region bounded by $\Sigma_4$ and $\Sigma_5$ in \eqref{Psi-in-V} to  $\C \setminus i \R_+$.
Then, $\Phi$ satisfies the following RH problem.
\begin{proposition}\label{rhp:X}
The function $\Phi(z)=\Phi(z;s)$ defined in \eqref{eq:YtoX} has the following properties:
\begin{enumerate}
\item[\rm (1)] $\Phi(z)$ is defined and analytic in $\mathbb{C}\setminus \Sigma_{\Phi}$, where
\begin{equation}\label{def:Sigmaphi}
\Sigma_{\Phi}:=\cup^4_{i=2}\Sigma_i\cup \Sigma_0^{(s)} \cup \Sigma_1^{(s)} \cup \Sigma_5^{(s)} \cup [0,s],
\end{equation}
with
\begin{equation}\label{def:sigmais}
\begin{aligned}
&\Sigma_0^{(s)}=(s,+\infty), \quad \Sigma_1^{(s)}=s+e^{\frac{\pi i}{4}}(0,+\infty), \quad \Sigma_5^{(s)}=s+e^{-\frac{\pi i}{4}}(0,+\infty),
\end{aligned}
\end{equation}
and $\Sigma_j$, $j=2,3,4$, given in \eqref{def:sigmai}; see Figure \ref{fig:X} for an illustration and the orientations.

\item[\rm (2)] For $z\in \Sigma_{\Phi}$, $\Phi$ satisfies the jump condition
\begin{equation}\label{eq:X-jump}
 \Phi_+(z)=\Phi_-(z)J_\Phi(z),
\end{equation}
where
\begin{equation}\label{def:JX}
J_\Phi(z):=\left\{
 \begin{array}{ll}
          \begin{pmatrix} 0 & 1 &0 \\ -1 & 0 &0 \\ 0&0&1 \end{pmatrix}, & \qquad \hbox{$z\in \Sigma_0^{(s)}$,} \\
          \begin{pmatrix} 1&0&0 \\ 1&1&0 \\ 0&0&1 \end{pmatrix},  & \qquad  \hbox{$z\in \Sigma_1^{(s)}$,} \\
           \begin{pmatrix} 1&0&0 \\ 1&1&0 \\ 0&0&1 \end{pmatrix}, & \qquad  \hbox{$z\in \Sigma_5^{(s)}$,} \\
          \begin{pmatrix} 1&1-\gamma&0\\ 0&1&0 \\ 0&0&1 \end{pmatrix}, & \qquad  \hbox{$z\in (0,s)$,}  \\
          J_{\Psi_{\alpha}}(z), & \qquad  \hbox{$z\in \cup^4_{i=2}\Sigma_i$,}
        \end{array}
      \right.
 \end{equation}
with $J_{\Psi_{\alpha}}(z)$ given in \eqref{jumps:M}.

\item[\rm (3)]As $z \to \infty$ and $\pm \Im z>0$, we have
\begin{equation}\label{eq:asyX}
\Phi(z)=
 \frac{i z^{-\frac{\alpha}{3}}}{ \sqrt{3} }   C_\Psi\left(I+ \Boh(z^{-1} ) \right) \diag \left(z^{\frac13},1,z^{-\frac13} \right)L_{\pm} \diag \left( e^{\pm \frac{\alpha \pi i}{3}}, e^{\mp \frac{\alpha \pi i}{3}}, 1 \right) e^{\Theta(z)},
\end{equation}
where $C_\Psi$, $L_{\pm}$ and $\Theta(z)$ are given in \eqref{eq:cons-C-Psi}, \eqref{def:Lpm} and \eqref{def:Theta}, respectively.

\item[\rm (4)] As $z \to 0$, we have, 
if $\alpha \in \mathbb{N} \cup \{0\}$,
\begin{align} \label{eq:X-near-0-2}
	\Phi(z) =& \ \Phi^{(L)}(z) \begin{pmatrix}
	1 & \frac{\gamma - 1}{2\pi i} \ln z & 0 \\
	0 & z^{-\alpha} & 0 \\
	0&  \frac{e^{\alpha \pi i}}{2\pi i} \ln z  & 1
	\end{pmatrix}
\begin{cases}
	I, & z \in \textrm{\texttt{II}}, \\
	\begin{pmatrix}
1 & \gamma-1 & 0
\\
0 & 1 & 0
\\
0 & 0 & 1
\end{pmatrix}, & z \in \textrm{\texttt{V}};
	\end{cases}
\end{align}
if $\alpha \notin \mathbb{Z}$,
\begin{equation} \label{eq:X-near-0-3}
	\Phi(z) = \Phi^{(L)}(z) \left\{
             \begin{array}{ll}
               \begin{pmatrix}
	1 & c_+ & 0 \\
	0 & z^{-\alpha} & 0 \\
	0& \frac{1}{2 \sin ( \alpha \pi ) i} & 1
	\end{pmatrix}, & \hbox{$z \in \textrm{\texttt{II}},$}
           \\
               \begin{pmatrix}
	1 & c_- & 0 \\
	0 & z^{-\alpha} & 0 \\
	0& \frac{1}{2 \sin ( \alpha \pi )i} & 1
	\end{pmatrix}, & \hbox{$z \in \textrm{\texttt{V}}$,}
             \end{array}
           \right.
 \end{equation}
where we take principal branch for $\ln z$ and $z^{\alpha}$,
\begin{equation}\label{def:cpm}
c_{\pm}=\frac{\gamma-1}{2i\sin(\alpha \pi)}e^{\mp\alpha \pi i},
\end{equation}and $\Phi^{(L)}(z)$ is analytic at $z=0$ satisfying the  expansion
\begin{equation}\label{eq: Phi-expand-0}
  \Phi^{(L)}(z) =\Phi_0^{(0)}(s)\left(I+\Phi_1^{(0)}(s) \, z+\Boh(z^2)\right ),\qquad z\to 0,
 \end{equation}
for some functions $\Phi_0^{(0)}(s)$ and $\Phi_1^{(0)}(s)$. The local behavior of $\Phi$ near $z=0$ in other regions can be determined  through \eqref{eq:X-near-0-2}, \eqref{eq:X-near-0-3} and the jump condition \eqref{eq:X-jump} and \eqref{def:JX}.

\item[\rm (5)] As $z \to s$, we have
\begin{align} \label{eq:X-near-s}
	\Phi(z) =& \ \Phi^{(R)}(z) \begin{pmatrix}
	1 & -\frac{\gamma}{2\pi i} \ln(z-s) & 0 \\
	0 & 1 & 0 \\
	0& 0 & 1
	\end{pmatrix}
\begin{cases}
	I, & z \in \textrm{\texttt{II}}, \\
	\begin{pmatrix}
1 & -1 & 0
\\
0 & 1 & 0
\\
0 & 0 & 1
\end{pmatrix}, & z \in \textrm{\texttt{V}},
	\end{cases}
\end{align}
where the principal branch is taken for $\ln(z-s)$, and $\Phi^{(R)}(z)$ is analytic at $z=s$ satisfying the expansion
\begin{equation}\label{eq: Phi-expand-s}
  \Phi^{(R)}(z) =\Phi_0^{(1)}(s)\left(I+\Phi_1^{(1)}(s)(z-s)+\Boh(z-s)^2\right ),\qquad z\to s,
 \end{equation}
for some functions $\Phi_0^{(1)}(s)$ and $\Phi_1^{(1)}(s)$. The local behavior of $\Phi$ near $z=s$ in other regions can be determined  through \eqref{eq:X-near-s} and the jump condition \eqref{eq:X-jump} and \eqref{def:JX}.

\end{enumerate}
\end{proposition}

\begin{proof}
It is clear that $\Phi(z)$ defined in \eqref{eq:YtoX} is analytic in $\mathbb{C}\setminus \Sigma_{\Phi}$. To show the jump condition \eqref{eq:X-jump} and \eqref{def:JX}, we only need to check the jump over $(0,s)$, while the other jumps follow directly from \eqref{eq:YtoX} and \eqref{jumps:M}. By \eqref{def:fh} and \eqref{eq:Y-jump}, we have, for $0<x<s$,
\begin{equation}
Y_-(x)^{-1}Y_+(x) = I-2\pi i  \vec{f}(x)\vec{h}(x)^t=\widetilde\Psi_\alpha(x)
\begin{pmatrix}
1 & -\gamma & 0
\\
0 & 1  & 0
\\
0 & 0 & 1
\end{pmatrix}\widetilde\Psi_\alpha(x)^{-1}.
\end{equation}
This, together with \eqref{eq:YtoX}, implies that for $x\in(0,s)$,
\begin{align}
J_\Phi(x) &= \Phi_-(x)^{-1} \Phi_+(x) =  \widehat \Psi_\alpha(x)^{-1} Y_-(x)^{-1} Y_+(x) \widetilde \Psi_\alpha(x)
\nonumber
\\
&=\widehat \Psi_\alpha(x)^{-1} \widetilde \Psi_\alpha(x)
\begin{pmatrix}
1 & -\gamma & 0
\\
0 & 1  & 0
\\
0 & 0 & 1
\end{pmatrix}.
\end{align}
Recall that both $\widetilde \Psi_\alpha(x)$ and $\widehat \Psi_\alpha(x)$ are analytic continuation of $\Psi_\alpha(z)$, it then follows from the jump condition of $\Psi_\alpha$ in \eqref{jumps:M} that
\begin{equation}
\widehat \Psi_\alpha(x)^{-1} \widetilde \Psi_\alpha(x)
=
\begin{pmatrix} 1&0&0 \\ 1&1&0 \\ 0&0&1 \end{pmatrix} \begin{pmatrix} 0&1&0 \\ -1&0&0 \\ 0&0&1 \end{pmatrix} \begin{pmatrix} 1&0&0 \\ 1&1&0 \\ 0&0&1 \end{pmatrix}
=
\begin{pmatrix}
1 & 1 & 0
\\
0 & 1  & 0
\\
0 & 0 & 1
\end{pmatrix}.
\end{equation}
A combination of the above two formulas gives us the jump of $\Phi(z)$ over $(0,s)$.

We next verify the local behavior of $\Phi(z)$ near $z=0$ given in \eqref{eq:X-near-0-2} and \eqref{eq:X-near-0-3}, and it suffices to show that they satisfy the jump condition of $\Phi$. If $\alpha \in \mathbb{N} \cup \{0\}$, we see from the jumps of $\Phi(z)$ on $\Sigma_2$ and $\Sigma_4$ that, as $z\to 0$,
\begin{align}\label{eq:PhizeroIII}
\Phi(z) = \Phi^{(L)}(z) \begin{pmatrix}
	1 & \frac{\gamma - 1}{2\pi i} \ln z & 0 \\
	0 & z^{-\alpha} & 0 \\
	0& \frac{e^{\alpha \pi i}}{2\pi i} \ln z & 1
	\end{pmatrix}
\left\{
  \begin{array}{ll}
    \begin{pmatrix} 1&0&0 \\ 0&1& - e^{\alpha \pi i} \\ 0&0&1 \end{pmatrix}, & \hbox{$z \in \textrm{\texttt{III}}$,} \\
    \begin{pmatrix}
1 & \gamma-1 & 0
\\
0 & 1 & 0
\\
0 & 0 & 1
\end{pmatrix} \begin{pmatrix} 1&0&0 \\ 0&1& e^{-\alpha \pi i} \\ 0&0&1 \end{pmatrix}, & \hbox{$z \in \textrm{\texttt{IV}}$.}
  \end{array}
\right.
\end{align}
Thus, for $x<0$, it follows from \eqref{eq:X-jump}, \eqref{def:JX} and \eqref{eq:PhizeroIII} that
\begin{align*}
\Phi_-(x)^{-1} \Phi_+(x) =& \begin{pmatrix} 1&0&0 \\ 0&1& -e^{-\alpha \pi i} \\ 0&0&1 \end{pmatrix} \begin{pmatrix}
1 & 1-\gamma & 0
\\
0 & 1 & 0
\\
0 & 0 & 1
\end{pmatrix}
 \begin{pmatrix}
	1 & -\frac{\gamma - 1}{2\pi i} (\ln |x| - \pi i)|x|^\alpha e^{-\alpha \pi i} & 0 \\
	0 & |x|^\alpha e^{-\alpha \pi i}  & 0 \\
	0& -\frac{e^{\alpha \pi i}}{2\pi i} (\ln |x| - \pi i)|x|^\alpha e^{-\alpha \pi i} & 1
	\end{pmatrix}
	\\
	& \times \Phi^{(L)}_-(x)^{-1}\Phi^{(L)}_+(x)\begin{pmatrix}
	1 & \frac{\gamma - 1}{2\pi i} (\ln |x| + \pi i) & 0 \\
	0 & |x|^{-\alpha} e^{-\alpha \pi i} & 0 \\
	0& \frac{e^{\alpha \pi i}}{2\pi i} (\ln |x| + \pi i) & 1
	\end{pmatrix} \begin{pmatrix} 1&0&0 \\ 0&1& - e^{\alpha \pi i} \\ 0&0&1 \end{pmatrix} \\
	= & \begin{pmatrix} 1&0&0 \\ 0&0& - e^{-\alpha \pi i} \\ 0&e^{-\alpha \pi i}&0 \end{pmatrix} ,
\end{align*}
as required.
Similarly, if $\alpha \notin \mathbb{Z}$, we have, as $z\to 0$,
\begin{align}\label{eq:PhizeroIV}
\Phi(z) =\Phi^{(L)}(z)\left\{
                           \begin{array}{ll}
                            \begin{pmatrix}
	                        1 & c_+ & 0 \\
	                        0 & z^{-\alpha} & 0 \\
	                        0& \frac{1}{2 i \sin(\alpha \pi)}  & 1
	                       \end{pmatrix}
                           \begin{pmatrix}
                           1&0&0 \\ 0&1& - e^{\alpha \pi i} \\ 0&0&1
                           \end{pmatrix}, & \hbox{$z \in \textrm{\texttt{III}}$,} \\
                            \begin{pmatrix}
	                        1 & c_- & 0 \\
	                        0 & z^{-\alpha} & 0 \\
	                        0& \frac{1}{2 i \sin(\alpha \pi)}  & 1
	                        \end{pmatrix}
                            \begin{pmatrix} 1&0&0 \\ 0&1& e^{-\alpha \pi i} \\ 0&0&1
                            \end{pmatrix}, & \hbox{$z \in \textrm{\texttt{IV}}$.}
                           \end{array}
                         \right.
\end{align}
Thus, it follows from
\eqref{eq:X-jump}, \eqref{def:JX} and \eqref{eq:PhizeroIV} that, for $x<0$,
\begin{align*}
\Phi_-(x)^{-1} \Phi_+(x) = &\begin{pmatrix} 1&0&0 \\ 0&1& -e^{-\alpha \pi i} \\ 0&0&1 \end{pmatrix}  \begin{pmatrix}
	1 &  -c_- |x|^\alpha e^{-\alpha \pi i}  & 0 \\
	0 & |x|^\alpha e^{-\alpha \pi i}  & 0 \\
	0& -\frac{1}{2  i \sin(\alpha \pi)} |x|^\alpha e^{-\alpha \pi i} & 1
	\end{pmatrix}
	\\
	& \times  \Phi^{(L)}_-(x)^{-1}\Phi^{(L)}_+(x) \begin{pmatrix}
	1 & c_+ & 0 \\
	0 & |x|^{-\alpha} e^{-\alpha \pi i} & 0 \\
	0& \frac{1}{2  i \sin(\alpha \pi)}   & 1
	\end{pmatrix} \begin{pmatrix} 1&0&0 \\ 0&1& - e^{\alpha \pi i} \\ 0&0&1 \end{pmatrix} \\
	 = & \begin{pmatrix} 1&0&0 \\ 0&0& - e^{-\alpha \pi i} \\ 0&e^{-\alpha \pi i}&0 \end{pmatrix},
\end{align*}
and from \eqref{eq:X-jump}, \eqref{def:JX} and \eqref{eq:X-near-0-3} that, for $x\in(0,s)$,
\begin{align*}
\Phi_-(x)^{-1} \Phi_+(x) = &
\begin{pmatrix}
1 & -c_-x^{\alpha} & 0 \\ 0 & x^{\alpha} & 0 \\ 0& -\frac{1}{2i\sin(\alpha \pi)}x^{\alpha} & 1 \end{pmatrix}  \Phi^{(L)}_-(x)^{-1}\Phi^{(L)}_+(x)
\begin{pmatrix}
	1 & c_+ & 0 \\
	0 & x^{-\alpha} & 0 \\
	0& \frac{1}{2  i \sin(\alpha \pi)}   & 1
\end{pmatrix}
	\\
	 = & \begin{pmatrix} 1 &  1-\gamma & 0 \\ 0& 1 & 0 \\ 0&0 &1 \end{pmatrix},
\end{align*}
which is consistent with the jump condition satisfied by $\Phi$.

Finally, the local behavior of $\Phi(z)$ at $z=s$ can be verified in a manner similar to that at $z=0$, we omit the details here.

This completes the proof of Proposition \ref{rhp:X}.
\end{proof}

The connection between the above RH problem and the derivative of $F(s;\gamma,\rho)$ is revealed in the following proposition.
\begin{proposition}\label{prop:derivativeandX}
With $F$ defined in \eqref{def:Fnotation}, we have
\begin{align}
\frac{\ud}{\ud s} F(s;\gamma,\rho)&=\frac{\ud}{\ud s} \ln \det\left(I- \gamma \mathcal{K}_{s,\alpha} \right) =-\frac{\gamma}{ 2 \pi i} \left(\Phi_1^{ (1) } (s) \right)_{21} ,
\label{eq:derivativeins}
\end{align}
where $\Phi_1^{(1)} (s)$ is given in \eqref{eq: Phi-expand-s} and $(M)_{ij}$ stands for the $(i,j)$-th entry of a matrix $M$.
\end{proposition}
\begin{proof}
For $z\in \texttt{II}$, we see from \eqref{def:fh}, \eqref{def:FH2} and \eqref{eq:YtoX} that
\begin{equation}\label{eq:FHinX}
\vec{F}(z)=Y(z)\vec{f}(z)=Y(z)\widetilde \Psi_\alpha(z)\begin{pmatrix}
1
\\
0
\\
0
\end{pmatrix}= \Phi(z)
\begin{pmatrix}
1
\\
0
\\
0
\end{pmatrix}
\end{equation}
and
\begin{align}
\vec{H}(z)&=Y(z)^{-t}\vec{h}(z)=  \Phi(z)^{-t}\widetilde \Psi_\alpha(z)^{t} \cdot \frac{\gamma }{2 \pi i} \widetilde \Psi_\alpha(z)^{-t} \begin{pmatrix}
0
\\
1
\\
0
\end{pmatrix} =\frac{\gamma }{2 \pi i}  \Phi(z)^{-t}
\begin{pmatrix}
0
\\
1
\\
0
\end{pmatrix}.
\end{align}
Combining the above formulas and \eqref{eq:resolventexpli}, we obtain
\begin{equation} \label{eq: R-Xij}
	R(z,z) = \frac{\gamma}{2 \pi i}  \left(\Phi(z)^{-1}\Phi'(z)\right)_{21}, \qquad  z\in \texttt{II}.
\end{equation}
This, together with \eqref{eq:derivatives}, implies that
\begin{equation}
 \frac{\ud}{\ud s} F(s;\gamma,\rho) = -\frac{\gamma}{2 \pi i} \lim_{z \to s} \left(\Phi(z)^{-1}\Phi'(z)\right)_{21}, \qquad  z\in \texttt{II}.
\end{equation}
With the aid of the local behavior of $\Phi(z)$ as $z \to s$ given in \eqref{eq:X-near-s} and \eqref{eq: Phi-expand-s}, we arrive at \eqref{eq:derivativeins}.

This completes the proof of Proposition \ref{prop:derivativeandX}.
\end{proof}

\section{The Lax pair and differential identities}\label{sec:Lax}

Based on the RH problem for $\Phi$ in Proposition \ref{rhp:X}, we will construct the Lax pair for $\Phi(z;s)$, which leads to the system of equations \eqref{def:sysdiff} and \eqref{eq:constraint}.  Several useful differential identities for the associated Hamiltonian \eqref{def:H} will be also derived for later use.

\subsection{Lax pair equations}
The Lax pair for $\Phi(z;s)$ is given as follows.
\begin{proposition}\label{pro:Lax pair}
  For the function $\Phi(z) = \Phi(z;s)$ defined in \eqref{eq:YtoX}, we have
  \begin{equation} \label{eq:Lax pair}
    \frac{\partial}{\partial z}\Phi(z;s) = L(z;s)\Phi(z;s), \qquad  \frac{\partial}{\partial s}\Phi(z;s) = U(z;s)\Phi(z;s),
  \end{equation}
  where
   \begin{equation} \label{eq:L}
L(z;s)=\begin{pmatrix}
      0 & 1 & 0 \\
      0& 0 & 1 \\
      0   & 0& 0
    \end{pmatrix}+\frac{A_0(s)}{z}+\frac{A_1(s)}{z-s}
  \end{equation}
  and
   \begin{equation} \label{eq:U}
U(z;s)=
     - \frac{A_1(s)}{z-s},
  \end{equation}
   with
  \begin{equation} \label{eq:A-k}
A_i(s)  = \begin{pmatrix}
      q_{i,1}(s) \\
      q_{i,2}(s) \\
      q_{i,3}(s)
    \end{pmatrix}
    \begin{pmatrix}
      p_{i,1}(s) &  p_{i,2}(s) & p_{i,3}(s)
    \end{pmatrix}, \quad i=0,1.\end{equation}
Moreover, the functions $p_{i,k}$ and $q_{i,k}$, $i=0, 1$, $k=1,2,3$, in \eqref{eq:A-k} satisfy \eqref{def:sysdiff}, \eqref{eq:constraint}
and
\begin{align}
p_{0,1}(s)q_{0,3}(s)+p_{1,1}(s)q_{1,3}(s)&=1, \label{eq: equationpqk-1} \\
p_{0,1}(s)q_{0,2}(s)+p_{0,2}(s)q_{0,3}(s)+p_{1,1}(s)q_{1,2}(s)+p_{1,2}(s)q_{1,3}(s) & = \rho. \label{eq: equationpqk-2}
\end{align}
\end{proposition}
\begin{proof}
Since all the jump matrices of $\Phi$ given in \eqref{def:JX} are independent of $z$ and $s$,  then the coefficient matrices
   \begin{equation}
   L(z;s):=\frac{\partial}{\partial z}\Phi(z;s)\cdot \Phi(z;s)^{-1}, \qquad U(z;s):=\frac{\partial}{\partial s}\Phi(z;s)\cdot \Phi(z;s)^{-1}
   \end{equation}
are analytic in the complex $z$-plane except for possible isolated singularities at $z= 0$, $z= s $ and $z = \infty$.

To evaluate $L$ and $U$, we need to know their local behaviors near these points. As $z \to \infty$,  we have from \eqref{eq:asyX} that
\begin{align}\label{eq:L-expand}
   &L(z;s)=  A + \Boh(z^{-1}),
  \end{align}
where
\begin{equation}\label{def:A}
A:=\begin{pmatrix}
      0 & 1 & 0 \\
      0 & 0 & 1 \\
      0& 0 & 0
    \end{pmatrix}.
\end{equation}
On the other hand, in view of the local behaviors of $\Phi$ near $z=z_0=0$ and $z=z_1=s$ as shown in \eqref{eq:X-near-0-2}--\eqref{eq: Phi-expand-s}, it is readily seen that
$$L(z;s)\sim \frac{A_i(s)}{z-z_i},\qquad z \to z_i, \quad i=0,1,$$
where
 \begin{align}\label{eq: L-expand-z-k}
 A_0(s) &=\left\{
            \begin{array}{ll}
              -\alpha \Phi_0^{(0)}(s) \begin{pmatrix} 0&0&0\\ 0&1&0 \\ 0&0&0 \end{pmatrix} \Phi_0^{(0)}(s)^{-1}, & \hbox{$\alpha \neq 0$,}
              \\
              \frac{1}{2 \pi i} \Phi_0^{(0)}(s) \begin{pmatrix} 0&\gamma - 1&0\\ 0&0&0 \\ 0&1&0 \end{pmatrix} \Phi_0^{(0)}(s)^{-1}, & \hbox{$\alpha=0$,}
            \end{array}
          \right.
 \\
 A_1(s) &= -\frac{\gamma}{2\pi i}\Phi_0^{(1)}(s) \begin{pmatrix} 0&1&0\\ 0&0&0 \\ 0&0&0 \end{pmatrix} \Phi_0^{(1)}(s)^{-1}, \label{def:A1}
 \end{align}
and where $\Phi_0^{(i)}(s)$, $i=0,1$, are given in \eqref{eq: Phi-expand-0} and \eqref{eq: Phi-expand-s}, respectively.
Thus, if we define
 \begin{equation}\label{def: q-0}
\begin{pmatrix}
      q_{0,1}(s) \\
      q_{0,2}(s) \\
    q_{0,3}(s)
    \end{pmatrix}=\left\{
                     \begin{array}{ll}
                       \Phi_0^{(0)}(s) \begin{pmatrix}
      0\\
    1 \\
   0
    \end{pmatrix}, & \hbox{$\alpha \neq 0$,}
\\
                      \frac{1}{2\pi i} \Phi_0^{(0)}(s) \begin{pmatrix}
      \gamma-1\\
    0 \\
    1
    \end{pmatrix}, & \hbox{$\alpha=0$,}
                     \end{array}
                   \right.
\end{equation}
\begin{equation}\label{def:p0i}
\begin{pmatrix}
      p_{0,1}(s) \\
      p_{0,2}(s) \\
    p_{0,3}(s)
    \end{pmatrix}=\left\{
                     \begin{array}{ll}
                       -\alpha \Phi_0^{(0)}(s)^{-t}\begin{pmatrix}
      0\\
    1 \\
   0
    \end{pmatrix}, & \hbox{$\alpha \neq 0$,}
\\
                       \Phi_0^{(0)}(s)^{-t}
\begin{pmatrix}
      0\\
    1 \\
   0
    \end{pmatrix}, & \hbox{$\alpha=0$,}
                     \end{array}
                   \right.
\end{equation}
and
\begin{equation}\label{def: q-1}
\begin{pmatrix}
      q_{1,1}(s) \\
      q_{1,2}(s) \\
    q_{1,3}(s)
    \end{pmatrix}=\Phi_0^{(1)}(s) \begin{pmatrix}
      1\\
    0 \\
   0
    \end{pmatrix},\qquad
    \begin{pmatrix}
      p_{1,1}(s) \\
      p_{1,2}(s) \\
    p_{1,3}(s)
    \end{pmatrix}=-\frac{\gamma}{2\pi i}\Phi_0^{(1)}(s)^{-t}\begin{pmatrix}
      0\\
    1 \\
   0
    \end{pmatrix},
\end{equation}
we see the expressions of $A_0(s)$ and $A_1(s)$ in \eqref{eq:A-k} from \eqref{eq: L-expand-z-k} and \eqref{def:A1}. Moreover, the equations \eqref{eq: L-expand-z-k} and \eqref{def:A1} also imply that
\begin{equation}\label{eq: const1}
\mathtt{Tr} A_0(s) = \sum_{k=1}^3q_{0,k}(s)p_{0,k}(s)=-\alpha, \quad \mathtt{Tr} A_1(s) = \sum_{k=1}^3q_{1,k}(s)p_{1,k}(s)=0,
\end{equation}
which are the relations \eqref{eq:constraint}. Similarly, by \eqref{eq:asyX}, it follows that
\begin{align}\label{eq:U-expand}
   U(z;s)= \Boh(z^{-1}), \qquad z\to \infty.
\end{align}
A further appeal to \eqref{eq:X-near-0-2}--\eqref{eq: Phi-expand-s} shows that
\begin{equation}\label{eq:U-local}
U(z;s)= \Boh(1), \quad z\to 0, \qquad U(z;s) \sim -\frac{A_1(s)}{z-s}, \quad z\to s,
\end{equation}
where $A_1(s)$ is given in \eqref{def:A1}. Combining \eqref{eq:U-expand} and \eqref{eq:U-local} gives us \eqref{eq:U}.

We next come to the differential equations satisfied by the functions $p_{i,k}$ and $q_{i,k}$, $i=0, 1$, $k=1,2,3$. From \eqref{eq:Lax pair}--\eqref{eq:U}, one has
\begin{equation}\label{eq: Ulimit}
	\lim_{z\to 0} \frac{\partial}{\partial s}\Phi(z;s) \Phi(z;s)^{-1} = \lim_{z\to 0}U(z;s)=\frac{A_1(s)}{s},
\end{equation}
and
 \begin{equation}\label{eq:LUlimit}
	\lim_{z\to s} \left(\frac{\partial }{\partial z}\Phi(z;s)+\frac{\partial}{\partial s}\Phi(z;s) \right)\Phi(z;s)^{-1}=\lim_{z\to s}(L(z;s)+U(z;s))= A + \frac{A_0(s)}{s}.
\end{equation}
Substituting \eqref{eq: Phi-expand-0} and  \eqref{eq: Phi-expand-s} into the left-hand side of \eqref{eq: Ulimit} and \eqref{eq:LUlimit}, respectively, it is easily seen that
\begin{equation}\label{eq:dphi00}
\frac{\ud}{\ud s}\Phi_0^{(0)}(s) =\frac{A_1(s)}{s}\Phi_0^{(0)}(s),
\end{equation}
and
\begin{equation}\label{eq:dphi01}
\frac{\ud}{\ud s}\Phi_0^{(1)}(s) = \left( A +\frac{A_0(s)}{s} \right) \Phi_0^{(1)}(s).
\end{equation}
Recall the definitions of $q_{i,k}(s)$, $i=0,1$, $k=1,2,3$, given in \eqref{def: q-0} and \eqref{def: q-1}, we then obtain from the above two formulas that
\begin{align}\label{eq:d-Phi-0}
	\begin{pmatrix}
  q_{0,1}'(s)  & q_{0,2}'(s) & q_{0,3}'(s)
\end{pmatrix}^{t}&=\frac{A_1(s)}{s} \begin{pmatrix}
  q_{0,1}(s)  & q_{0,2}(s) & q_{0,3}(s)
\end{pmatrix}^{t},
	\\
 \label{eq:d-Phi-1}
	\begin{pmatrix}
  q_{1,1}'(s)  & q_{1,2}'(s) & q_{1,3}'(s)
\end{pmatrix}^{t}&=\left( A +\frac{A_0(s)}{s} \right)  \begin{pmatrix}
  q_{1,1}(s)  & q_{1,2}(s) & q_{1,3}(s)
\end{pmatrix}^{t}.
	\end{align}
By \eqref{eq:A-k} and \eqref{def:A}, equations \eqref{eq:d-Phi-0} and \eqref{eq:d-Phi-1} are equivalent to the first two equations in \eqref{def:sysdiff}. To show the equations for $p_{0,k}'(s)$ and $p_{1,k}'(s)$ in \eqref{def:sysdiff}, we note that for an invertible matrix-valued function $\mathcal{M}(s)$,
$$
\frac{\ud}{\ud s}\mathcal{M}(s)^{-t}=-\mathcal{M}(s)^{-t}\cdot \frac{\ud}{\ud s} \mathcal{M}(s)^t \cdot \mathcal{M}(s)^{-t}.
$$
Thus, we see from \eqref{eq:dphi00}, \eqref{eq:dphi01} and the definitions of $p_{i,k}(s)$, $i=0,1$, $k=1,2,3$, given in \eqref{def:p0i} and \eqref{def: q-1} that
\begin{align}\label{eq:d-p-0}
	\begin{pmatrix}
  p_{0,1}'(s)  & p_{0,2}'(s) & p_{0,3}'(s)
\end{pmatrix}^{t}&=-\frac{A_1(s)^t}{s} \begin{pmatrix}
  p_{0,1}(s)  & p_{0,2}(s) & p_{0,3}(s)
\end{pmatrix}^{t},
	\\
 \label{eq:d-p-1}
	\begin{pmatrix}
  p_{1,1}'(s)  & p_{1,2}'(s) & p_{1,3}'(s)
\end{pmatrix}^{t}&=- \left( A^t +\frac{A_0(s)^t}{s} \right)  \begin{pmatrix}
  p_{1,1}(s)  & p_{1,2}(s) & p_{1,3}(s)
\end{pmatrix}^{t},
	\end{align}
which gives us the last two equations in \eqref{def:sysdiff}.

Finally, to show \eqref{eq: equationpqk-1} and \eqref{eq: equationpqk-2}, we note that, after a direct calculation, the $\Boh(z^{-1})$ term in \eqref{eq:L-expand} reads
\begin{equation}\label{eq:polecoeff}
\begin{pmatrix}
\ast & \ast & \ast
\\
\pi_3(\rho)+\frac{2\rho}{3} & \ast & \ast
\\
1 & -\pi_3(\rho)+\frac{\rho}{3} & \ast
\end{pmatrix} \frac{1}{z},
\end{equation}
where $\ast$ stands for some unimportant entry and $\pi_3(\rho)$ is defined in \eqref{pi3-def}. To that end, the following observation is helpful:
\begin{equation}
M(z) \Theta'(z) M(z)^{-1} =
\begin{pmatrix}
0 & 1 & \frac{\rho}{3} \vspace{3pt} \\
\frac{\rho}{3z} & 0 & 1 \vspace{3pt}  \\
\frac{1}{z} & \frac{\rho}{3z} & 0
\end{pmatrix},
\end{equation}
where $\Theta(z)$ is defined in \eqref{def:Theta} and $M(z) = \diag \left(z^{\frac13},1,z^{-\frac13} \right) L_{\pm} \diag \left( e^{\pm \frac{\alpha \pi i}{3}}, e^{\mp \frac{\alpha \pi i}{3}}, 1 \right)$ with $L_{\pm}$ given in \eqref{def:Lpm}. In view of \eqref{eq:L}, it is readily seen from \eqref{eq:L-expand} and \eqref{eq:polecoeff} that
\begin{equation}\label{eq: equationAk-1}
(A_0)_{31}+(A_1)_{31}=1, \quad (A_0)_{21}+(A_0)_{32}+(A_1)_{21}+(A_1)_{32}= \rho.
\end{equation}
Combining the above equations with the expressions of $A_0$ and $A_1$ in \eqref{eq:A-k} gives us \eqref{eq: equationpqk-1} and \eqref{eq: equationpqk-2}.

This completes the proof of the Proposition \ref{pro:Lax pair}.
\end{proof}

\begin{remark}
From the general theory of Jimbo-Miwa-Ueno \cite{JM}, it follows that
\begin{align}\label{eq:H-F}
H(s)= -\frac{\gamma}{2\pi i} \mathtt{Tr}
\left(\Phi_1^{(1)}(s)\begin{pmatrix}
      0 & 1 & 0 \\
      0 & 0 & 0 \\
      0 & 0 & 0
    \end{pmatrix}\right)
= -\frac{\gamma}{2\pi i}  \left(\Phi_1^{(1)} (s) \right)_{21},
\end{align}
where $\Phi_1^{(1)} (s)$ given in \eqref{eq: Phi-expand-s} appears in the expansion of $\Phi(z)$ near $z =s$. In view of the first equation in the Lax pair \eqref{eq:Lax pair}, by sending $z\to s$, we obtain from  \eqref{eq:L} and \eqref{eq:X-near-s} that the $\Boh(1)$ term in the expansion yields
\begin{equation}\label{eq:Phi-12}
\Phi_1^{(1)}(s)=\frac{\gamma}{2\pi i}\left[\Phi_1^{(1)}(s),\begin{pmatrix}
      0 & 1 & 0 \\
      0 & 0 & 0 \\
      0& 0 & 0
    \end{pmatrix} \right]+\Phi_0^{(1)}(s)^{-1}\left( A +\frac{A_0(s)}{s}\right) \Phi_0^{(1)}(s),
    \end{equation}
where $[A,B]:=AB-BA$ denotes the commutator of two matrices $A$ and $B$.
Since the $(2,1)$-entry of $\left[\Phi_1^{(1)}(s),\begin{pmatrix}
      0 & 1 & 0 \\
      0 & 0 & 0 \\
      0& 0 & 0
\end{pmatrix} \right]$ equals $0$, we obtain from \eqref{eq:H-F}, \eqref{eq:Phi-12}, the expression of $A_0(s)$ in \eqref{eq:A-k}, and the definitions of $p_{i,k}(s)$ and $q_{i,k}(s)$, $i=0,1$, $k=1,2,3$, in \eqref{def: q-0}--\eqref{def: q-1} that
\begin{align}\label{eq:Hpq}
sH(s) & = -\frac{s\gamma}{2\pi i}  \left(\Phi_1^{(1)} (s) \right)_{21}=-\frac{s\gamma}{2\pi i}  \begin{pmatrix}
0 & 1 & 0
\end{pmatrix}\Phi_1^{(1)} (s) \begin{pmatrix} 1
\\
0
\\
0
\end{pmatrix}
\nonumber
\\
&=-\frac{\gamma}{2\pi i}  \begin{pmatrix}
0 & 1 & 0
\end{pmatrix}\Phi_0^{(1)}(s)^{-1}\left(sA+A_0(s)\right) \Phi_0^{(1)}(s) \begin{pmatrix} 1
\\
0
\\
0
\end{pmatrix}
\nonumber
\\
&=\begin{pmatrix}
      p_{1,1}(s) &  p_{1,2}(s) & p_{1,3}(s)
    \end{pmatrix}
    \nonumber
    \\
    \nonumber
    &\quad \times \left(\begin{pmatrix}
      0 & s & 0 \\
      0 & 0 & s \\
      0& 0 & 0
    \end{pmatrix}+\begin{pmatrix}
      q_{0,1}(s) \\
      q_{0,2}(s) \\
      q_{0,3}(s)
    \end{pmatrix}
    \begin{pmatrix}
      p_{0,1}(s) &  p_{0,2}(s) & p_{0,3}(s)
    \end{pmatrix}\right) \begin{pmatrix}
      q_{1,1}(s) \\  q_{1,2}(s) \\ q_{1,3}(s)
    \end{pmatrix}
\nonumber
\\
&= s\left(p_{1,1}(s) q_{1,2}(s)+p_{1,2}(s) q_{1,3}(s)  \right)
 +\left(\sum_{k=1}^3p_{1,k}(s)q_{0,k}(s)\right)
\left (\sum_{k=1}^3p_{0,k}(s)q_{1,k}(s)\right),
\nonumber
\end{align}
which is equivalent to \eqref{def:H}.

\end{remark}

\subsection{Differential identities for the Hamiltonian}

Besides the Hamiltonian $H$ defined in \eqref{def:H}, the following differential identities are crucial in the derivation of the asymptotics for $F(s;\gamma,\rho)$ in Theorem \ref{thm:FAsy}, especially for the constant term.

\begin{proposition} \label{prop:H-diff}
With the Hamiltonian $H$ defined in \eqref{def:H}, we have
\begin{equation}\label{eq: dH-s}
 \frac{\ud}{\ud s}(sH(s)) =p_{1,1}(s)q_{1,2}(s)+p_{1,2}(s)q_{1,3}(s),
\end{equation}
and $H$ is related to the action differential  by
\begin{equation} \label{eq: action-diff}
\sum_{i=0}^1\sum_{k=1}^3p_{i,k}(s)q_{i,k}'(s)-H(s)
=H(s)-\frac{\ud}{\ud s}(sH(s)) .
\end{equation}
Moreover, we also have the following differential identities with respect to the parameters $\gamma$ and $\rho$:
\begin{align}
  \frac{\partial}{\partial \gamma}\left( \sum_{i=0}^1\sum_{k=1}^3p_{i,k}(s)q_{i,k}'(s)-H(s) \right)&=\frac{\ud}{\ud s}\left(
   \sum_{i=0}^1\sum_{k=1}^3p_{i,k}(s) \frac{\partial}{\partial \gamma}q_{i,k}(s)\right), \label{eq: dH-gamma} \\
   \frac{\partial}{\partial \rho}\left( \sum_{i=0}^1\sum_{k=1}^3p_{i,k}(s)q_{i,k}'(s)-H(s) \right)&=\frac{\ud}{\ud s}\left(
   \sum_{i=0}^1\sum_{k=1}^3p_{i,k}(s) \frac{\partial}{\partial \rho}q_{i,k}(s)\right).  \label{eq: dH-rho}
\end{align}
\end{proposition}
	
\begin{proof}
Applying the Hamiltonian equations \eqref{eq:Heq}, we have
\begin{equation}\label{eq:sHD}
  \frac{\ud}{\ud s} (sH )(s)= \sum_{i=0}^1\sum_{k=1}^3  \left(s\frac{\partial H}{\partial p_{i,k} }p_{i,k}'(s) +s\frac{\partial H}{\partial q_{i,k}} q_{i,k}'(s)\right)+\frac{\partial}{\partial s}(sH(s))=\frac{\partial }{\partial s}(sH(s)).
   \end{equation}
This, together with \eqref{def:H}, gives us \eqref{eq: dH-s}. A further combination of \eqref{def:H}, \eqref{eq: dH-s} and the differential equations \eqref{def:sysdiff} implies \eqref{eq: action-diff}. To obtain the differential identity with respect to the parameter $\gamma$, it is readily seen from \eqref{eq:Heq} that
\begin{multline}
  \frac{\partial}{\partial \gamma} H(s) = \sum_{i=0}^1\sum_{k=1}^3  \left(\frac{\partial H}{\partial p_{i,k} } \frac{\partial}{\partial \gamma} p_{i,k}(s) +\frac{\partial H}{\partial q_{i,k}}  \frac{\partial}{\partial \gamma} q_{i,k}(s)\right)
  \\
  =\sum_{i=0}^1\sum_{k=1}^3  \left(q_{i,k}'(s) \frac{\partial}{\partial \gamma} p_{i,k}(s) -p'_{i,k}(s)\frac{\partial}{\partial \gamma} q_{i,k}(s)\right),
\end{multline}
which gives us \eqref{eq: dH-gamma}. The differential identity with respect to $\rho$ in \eqref{eq: dH-rho} can be proved in a similar manner and we omit the details here.

This completes the proof of Proposition \ref{prop:H-diff}.
\end{proof}

In the next two sections, we will investigate asymptotics of the RH problem for $\Phi$ as $s \to +\infty$ and $s \to 0^+$ by using Deift-Zhou nonlinear steepest descent method \cite{DZ93}, respectively. The main idea is to covert the original RH problem into a small-norm one via a series of explicit and invertible transformations.


\section{Asymptotic analysis of the RH problem for $\Phi$ as $s \to +\infty$}\label{sec:AsyPhiinfty}

Throughout this section, it is assumed that $0<\gamma <1$.
	
\subsection{First transformation: $\Phi \to T$}
Note that the jump contour $\Sigma_{\Phi}$ in \eqref{def:Sigmaphi} is unbounded when $s \to +\infty$, our first transformation is to rescale the problem as follows:
\begin{equation}\label{def:PhiToT}
T(z)= -i \sqrt{3} \, s^{\frac{\alpha}{3}}  \diag \left( s^{-\frac13},1, s^{\frac13} \right) C_\Psi^{-1} \Phi(sz)e^{-\Theta(sz)},
\end{equation}
where we also take this opportunity to get rid of the exponential factor in the large-$z$ expansion in \eqref{eq:asyX}. In the above formula, the function $\Theta(z)$ and the constant matrix $C_\Psi$ are given in \eqref{def:Theta} and \eqref{eq:cons-C-Psi}, respectively. Then, $T(z)$ satisfies the following RH problem.

\begin{proposition}\label{rhp:T}
The function $T$ defined in \eqref{def:PhiToT} has the following properties:
\begin{enumerate}
\item[\rm (1)] $T(z)$ is defined and analytic in $\mathbb{C}\setminus \Sigma_T$,
where
$$
\Sigma_T:=\cup^4_{i=2}\Sigma_i\cup \Sigma_0^{(1)} \cup \Sigma_1^{(1)} \cup \Sigma_5^{(1)} \cup [0,1],
$$
and where the contours $\Sigma_j^{(1)}$, $j=0,1,5$, are defined in \eqref{def:sigmais} with $s=1$.

\item[\rm (2)] For $z\in \Sigma_T$, $T$ satisfies the jump condition
\begin{equation}\label{eq:T-jump}
 T_+(z)=T_-(z)J_T(z),
\end{equation}
where
\begin{equation}\label{def:JT}
J_T(z):=\left\{
 \begin{array}{ll}
          \begin{pmatrix} 0 & 1 &0 \\ -1 & 0 &0 \\ 0&0&1 \end{pmatrix}, & \qquad \hbox{$z\in \Sigma_0^{(1)}$,} \\
          \begin{pmatrix} 1&0&0 \\ e^{\theta_2(sz)-\theta_1(sz)}&1& 0 \\ 0&0&1 \end{pmatrix},  & \qquad  \hbox{$z\in \Sigma_1^{(1)}$,} \\
          \begin{pmatrix} 1&0&0 \\ 0&1&e^{\alpha \pi i} e^{\theta_2(sz)-\theta_3(sz)} \\ 0 & 0 &1 \end{pmatrix},  & \qquad \hbox{$z\in \Sigma_2$,} \\
          \begin{pmatrix} 1 &0&0 \\ 0&0&-e^{-\alpha \pi i} \\ 0&e^{-\alpha \pi i}&0 \end{pmatrix},  & \qquad  \hbox{$z\in \Sigma_3$,} \\
          \begin{pmatrix} 1&0&0 \\ 0&1& e^{- \alpha \pi i} e^{\theta_1(sz)-\theta_3(sz)} \\ 0 &  0 &1 \end{pmatrix}, & \qquad  \hbox{$z\in \Sigma_4$,} \\
          \begin{pmatrix} 1&0&0 \\ e^{\theta_1(sz)-\theta_2(sz)}&1& 0  \\ 0&0&1 \end{pmatrix}, & \qquad  \hbox{$z\in \Sigma_5^{(1)}$,} \\
          \begin{pmatrix} e^{\theta_2(sz)-\theta_1(sz)}&1-\gamma& 0\\ 0&e^{\theta_1(sz)-\theta_2(sz)}&0 \\ 0&0&1 \end{pmatrix}, & \qquad  \hbox{$z\in  (0,1)$,}
        \end{array}
      \right.
      \end{equation}
      with $\theta_k(z)=\theta_k(z;\rho)$, $k=1,2,3$, being defined in \eqref{eq: theta-k-def}.
\item[\rm (3)]As $z \to \infty$ and $\pm \Im z>0$, we have
\begin{equation}\label{eq:asyT}
T(z)= z^{-\frac{\alpha}{3}}
\left(I+ \Boh(z^{-1}) \right)\diag \left( z^{\frac13},1, z^{-\frac13} \right)L_{\pm} \diag \left( e^{\pm \frac{\alpha \pi i}{3}}, e^{\mp \frac{\alpha \pi i}{3}}, 1 \right),
\end{equation}
where $L_{\pm}$ is given in \eqref{def:Lpm}.

\item[\rm (4)]
As $z \to  1$, we have $T(z)=\Boh(\ln(z - 1))$.

\item[\rm (5)]
As $z \to  0$, we have
\begin{equation}
T(z)=\left\{
       \begin{array}{ll}
         \Boh( z^{-\alpha} ), & \hbox{$\alpha>0$,}
\\
\Boh(\ln z), & \hbox{$\alpha=0$,}
\\
\Boh(1), & \hbox{$-1<\alpha<0$.}
       \end{array}
     \right.
\end{equation}
\end{enumerate}
\end{proposition}
\begin{proof}
The above problem is a straightforward result from the transformation \eqref{def:PhiToT} and the RH problem for $\Phi$ given in Proposition \ref{rhp:X}, except for the jump on $\Sigma_3=(-\infty,0)$ where the function $\Theta(z)$ is not analytic.
For $x\in (-\infty,0)$, it is readily seen from  \eqref{eq:X-jump}, \eqref{def:JX} and \eqref{def:PhiToT} that
\begin{align}\label{eq:JT}
J_T(x)&=T_-(x)^{-1}T_+(x)
\nonumber \\
&=\diag\left(e^{\theta_{2,-}(sx)},e^{\theta_{1,-}(sx)},e^{\theta_{3,-}(sx)}\right)\Phi_-(sx)^{-1}\Phi_+(sx)
\nonumber \\
& \quad \times \diag\left(e^{-\theta_{1,+}(sx)},e^{-\theta_{2,+}(sx)},e^{-\theta_{3,+}(sx)}\right)
\nonumber \\
&=\begin{pmatrix} e^{\theta_{2,-}(sx)-\theta_{1,+}(sx)}& 0 & 0
\\ 0& 0 & - e^{-\alpha \pi i} e^{\theta_{1,-}(sx)-\theta_{3,+}(sx)}
\\ 0 &  e^{-\alpha \pi i} e^{\theta_{3,-}(sx)-\theta_{2,+}(sx)} & 0
\end{pmatrix}.
\end{align}	
In view of the definitions of $\theta_k$, $k=1,2,3$, in \eqref{eq: theta-k-def}, we have
\begin{align}\label{eq:thetarelations}
\theta_{1,+}(x)=\theta_{2,-}(x),\quad \theta_{2,+}(x)=\theta_{3,-}(x), \quad \theta_{3,+}(x)=\theta_{1,-}(x), \qquad x <0.
\end{align}
Combining the above two formulas, we obtain the jump $J_T$ on $\Sigma_3=(-\infty,0)$ as shown in \eqref{def:JT}.

This completes the proof of Proposition \ref{rhp:T}.
\end{proof}

\subsection{Second transformation: $T \to S$}
From the definitions of $\theta_k$, $k=1,2,3$, in \eqref{eq: theta-k-def}, one can see that all the entries containing $\theta_k(sz)$ in $J_T(z)$ \eqref{def:JT} are exponentially small as $s \to +\infty$, except $e^{\theta_2(sx)-\theta_1(sx)}$ when $x \in(0,1)$.
Indeed, we have
\begin{equation}
\theta_2(sx)-\theta_1(sx)=\sqrt{3}i s^{\frac23} \left( \frac{3}{2}x^{\frac23}-\frac{\rho}{s^{\frac{1}{3}}}x^{\frac13} \right), \qquad x\in(0,1),
\end{equation}
which implies that the $(1,1)$ and $(2,2)$ entries of $J_T(z)$ are highly oscillatory for $x\in(0,1)$ as $s \to +\infty$. Based on the factorizations
\begin{align}\label{eq:Deformation-1}
&
\begin{pmatrix}
e^{\theta_2(sx)-\theta_1(sx)}&1-\gamma& 0\\ 0&e^{\theta_1(sx)-\theta_2(sx)}&0 \\ 0&0&1
\end{pmatrix}
\nonumber \\
&=
\begin{pmatrix} 1&0&0\\ \frac{ e^{\theta_1(sx)-\theta_2(sx)}}{1-\gamma}&1& 0 \\ 0&0&1 \end{pmatrix}
\begin{pmatrix}
0 &1-\gamma&0
\\
\frac{1}{\gamma-1} & 0  & 0
\\ 0&0&1
\end{pmatrix}
\begin{pmatrix} 1&0&0\\ \frac{e^{\theta_2(sx)-\theta_1(sx)}}{1-\gamma} &1& 0 \\ 0&0&1 \end{pmatrix}, \qquad x\in (0,1),
\end{align}
a standard lens opening transformation can be introduced to remove the highly oscillatory entries of $J_T$ on $(0,1)$. Let $\Omega_\pm$ be the lens-shaped regions in the neighbourhood of $[0,1]$; see Figure \ref{fig:S}. The second transformation is set to be
\begin{equation}\label{def:TtoS}
S(z)=T(z) \left\{
\begin{array}{ll}
\begin{pmatrix} 1&0&0\\ \frac{ e^{\theta_1(sz)-\theta_2(sz)}}{1-\gamma}&1& 0 \\ 0&0&1 \end{pmatrix},  & \qquad z\in \Omega_{-},\\
  \begin{pmatrix} 1&0&0\\ \frac{e^{\theta_2(sz)-\theta_1(sz)}}{\gamma-1} &1&0 \\ 0&0&1 \end{pmatrix}, & \qquad z\in\Omega_{+},\\
I, & \qquad \mbox{elsewhere},
\end{array}
\right.
\end{equation}

\begin{figure}[h]
\begin{center}
   \setlength{\unitlength}{1truemm}
   \begin{picture}(100,70)(-5,2)
       \put(25,40){\line(-1,0){30}}
       \put(55,40){\line(1,0){30}}

       \put(25,40){\line(1,0){30}}

       \put(25,40){\line(-1,-1){25}}
       \put(25,40){\line(-1,1){25}}
       \put(55,40){\line(1,1){25}}
       \put(55,40){\line(1,-1){25}}

       \put(15,40){\thicklines\vector(1,0){1}}
       \put(65,40){\thicklines\vector(1,0){1}}

       \put(10,55){\thicklines\vector(1,-1){1}}
       \put(10,25){\thicklines\vector(1,1){1}}
       \put(70,25){\thicklines\vector(1,-1){1}}
       \put(70,55){\thicklines\vector(1,1){1}}

       \put(-2,11){$\Sigma_4$}

       \put(-2,67){$\Sigma_2$}
       \put(3,42){$\Sigma_3$}
       \put(80,11){$\Sigma_5^{(1)}$}
       \put(80,67){$\Sigma_1^{(1)}$}
       \put(73,42){$\Sigma_0^{(1)}$}


\put(39,42){$\Omega_{+}$}
\put(39,36){$\Omega_{-}$}


 \put(39,45){\thicklines\vector(1,0){1}}
  \put(39,35){\thicklines\vector(1,0){1}}
  \put(39,47){$\partial \Omega_{+}$}
  \put(39,31){$\partial \Omega_{-}$}


       \qbezier(25,40)(40,50)(55,40)
       \qbezier(25,40)(40,30)(55,40)

       \put(25,40){\thicklines\circle*{1}}
       \put(55,40){\thicklines\circle*{1}}

       \put(25,36.3){$0$}
       \put(54,36.3){$1$}
\end{picture}
   \caption{Regions $\Omega_\pm$ and the jump contours for the RH problem for $S$.}
   \label{fig:S}
\end{center}
\end{figure}
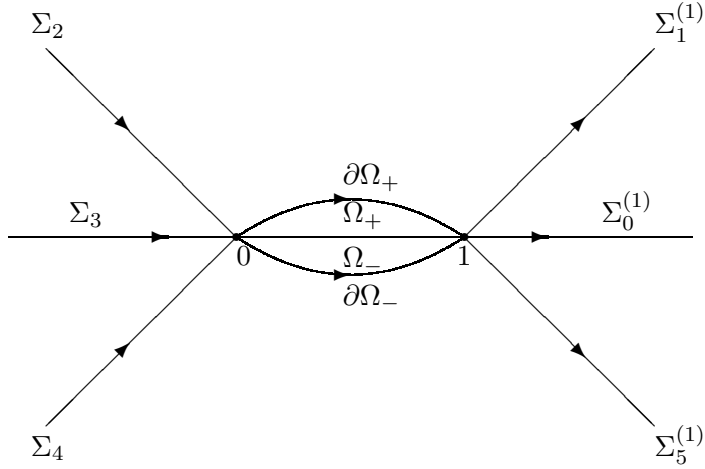

Then, it is straightforward to verify that $S(z)$ satisfies the following RH problem.
\begin{rhp}\label{rhp:S}
\hfill
\begin{enumerate}
\item[\rm (1)] $S(z)$ is defined and analytic in $\mathbb{C} \setminus \Sigma_S$,
where
\begin{equation}\label{def:SigmaS}
\Sigma_S:=\cup_{j=2}^4\Sigma_j\cup[0,1] \cup \partial \Omega_{\pm} \cup\Sigma_0^{(1)} \cup\Sigma_1^{(1)}  \cup\Sigma_5^{(1)};
\end{equation}
see Figure \ref{fig:S} for an illustration.

\item[\rm (2)] For $z\in \Sigma_S$, $S(z)$ satisfies the jump condition
\begin{equation}\label{eq:S-jump}
 S_+(z)=S_-(z)J_S(z),
\end{equation}
where
\begin{equation}\label{def:JS}
J_S(z):=\left\{
 \begin{array}{ll}
           J_T(z), & \qquad  \hbox{$z\in \cup_{j=2}^4\Sigma_j\cup\Sigma_0^{(1)} \cup\Sigma_1^{(1)}  \cup\Sigma_5^{(1)}$,} \\
    \begin{pmatrix} 1&0&0\\ \frac{ e^{\theta_1(sz)-\theta_2(sz)}}{1-\gamma}&1& 0 \\ 0&0&1 \end{pmatrix}, & \qquad  \hbox{$z\in \partial\Omega_{-}$,} \\
    \begin{pmatrix} 1&0&0\\ \frac{e^{\theta_2(sz)-\theta_1(sz)}}{1-\gamma} &1& 0 \\ 0&0&1 \end{pmatrix}, & \qquad  \hbox{$z\in \partial\Omega_{+}$,} \\
 \begin{pmatrix} 0&1-\gamma&0\\ \frac{1}{\gamma-1}&0&0 \\ 0&0&1 \end{pmatrix}, & \qquad  \hbox{$z\in  (0,1)$,}
        \end{array}
      \right.
 \end{equation}
and where $J_T$ is defined in \eqref{def:JT}.
\item[\rm (3)]As $z \to \infty$ and $\pm \Im z>0$, $S(z)$ has the same behavior as $T(z)$ given in \eqref{eq:asyT}.

\item[\rm (4)] As $z \to  1$ or $z\to 0$ outside the lens, $S(z)$ has the same local behaviors as $T(z)$.
\end{enumerate}
\end{rhp}

\subsection{Global parametrix}
As all the exponential entries in $J_S(z)$ tend to 0 exponentially fast as $s\to +\infty$, we are led to consider the following global parametrix.

\begin{rhp}\label{rhp:N}
\hfill
\begin{enumerate}
\item[\rm (1)] $N(z)$ is defined and analytic in $\mathbb{C}\setminus \mathbb{R}$.

\item[\rm (2)] For $x\in \mathbb{R}$, $N$ satisfies the jump condition
\begin{equation}\label{eq:N-jump}
 N_+(x)=N_-(x)J_N(x),
\end{equation}
where
\begin{equation}\label{def:JN}
J_N(x)=
\left\{
 \begin{array}{ll}
          \begin{pmatrix} 1 &0&0 \\ 0&0&-e^{-\alpha \pi i} \\ 0&e^{-\alpha \pi i}&0 \end{pmatrix}, & \qquad  \hbox{$x\in(-\infty,0)$,} \\
          \begin{pmatrix} 0&1-\gamma&0\\ \frac{1}{\gamma-1}&0&0 \\ 0&0&1 \end{pmatrix}, & \qquad  \hbox{$x\in (0,1)$,}\\
           \begin{pmatrix} 0 & 1 &0 \\ -1 & 0 &0 \\ 0&0&1 \end{pmatrix}, & \qquad \hbox{$x\in(1,+\infty)$.}
                  \end{array}
      \right.
 \end{equation}

\item[\rm (3)]As $z \to \infty$ and $\pm \Im z>0$, we have
\begin{equation}\label{eq:Ninfty}
N(z)= z^{-\frac{\alpha}{3}}
\left(I+ \Boh(z^{-1}) \right)  \diag \left( z^{\frac13},1, z^{-\frac13} \right)  L_{\pm} \diag \left( e^{\pm \frac{\alpha \pi i}{3}}, e^{\mp \frac{\alpha \pi i}{3}}, 1 \right).
\end{equation}

\end{enumerate}
\end{rhp}

The above RH problem can be solved explicitly in terms of the functions $d_k(z)$, $k=1,2,3,$ defined as follows. Recall that $\omega=e^{\frac{2}{3}\pi i}$, we set
\begin{equation}\label{eq:lambda}
\lambda(\xi)=\left( \frac{  \xi - \omega }{\xi-1}\right)^{\beta},\qquad \xi  \in \C \setminus \{ (  0, 1 ) \cup (0, 1)\omega \},
\end{equation}
where $\beta$ is given in \eqref{beta-def} and the branch is chosen such that $\lambda(\xi) \to 1$ as $\xi \to \infty$; see Figure \ref{fig:lambdaContour} for an illustration. The functions $d_k(z)$, $k=1,2,3,$ are then defined by
\begin{equation}\label{eq:dklambda}
\begin{array}{ll}
d_1(z)=\left\{
         \begin{array}{ll}
          \lambda(z^{\frac13}), & \quad \hbox{$\Im z>0$,} \\
           \lambda(\omega z^{\frac13}), & \quad \hbox{$\Im z<0$,}
         \end{array}
       \right.
 \\
d_{2}(z)= \left\{\begin{array}{ll}
          \lambda(\omega  z^{\frac13}), & \quad \hbox{$\Im z>0$,} \\
          \lambda(z^{\frac13}), & \quad \hbox{$\Im z<0$,}
         \end{array}
       \right.
\\
d_{3}(z)=\lambda(\omega^{-1} z^{\frac13}),
\end{array}
\end{equation}
where we choose  $\arg z\in(-\pi,\pi)$.

\begin{figure}[h]
\begin{center}
   \setlength{\unitlength}{1truemm}
   \vspace{-16mm}
   \begin{picture}(100,70)(-5,2)

       \put(40,40){\line(-1,2){6}}

       \put(40,40){\line(1,0){15}}

       \put(47.5,40){\thicklines\vector(1,0){1}}
       \put(37,46){\thicklines\vector(1,-2){1}}

  \put(40,40){\thicklines\circle*{1}}
%
       \put(55,40){\thicklines\circle*{1}}
       \put(34,52){\thicklines\circle*{1}}

       \put(32,54){$\omega$}
       \put(38,36.3){$0$}
       \put(54,36.3){$1$}
   \end{picture}
   \vspace{-33mm}
   \caption{The branch cut of the function $\lambda$.}
   \label{fig:lambdaContour}
\end{center}
\end{figure}
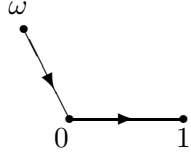

\begin{proposition}\label{prop:propdk}
The functions $d_k(z)$, $k=1,2,3,$ in \eqref{eq:dklambda} satisfy the following properties.
\begin{enumerate}
\item [\rm (i)] For $x \in \mathbb{R}$, we have
\begin{align}\label{eq:dJump1}
d_{1,+}(x)&=\frac{d_{2,-}(x)}{1-\gamma}, \qquad d_{2,+}(x)=d_{1,-}(x)(1-\gamma), \qquad &&x\in(0,1),
\\
d_{3,+}(x)&=d_{3,-}(x), \qquad  &&x\in(0,+\infty),\\
d_{1,+}(x)&=d_{2,-}(x), \qquad d_{2,+}(x)=d_{1,-}(x),  \qquad  && x\in(1,+\infty),
\\
d_{1,+}(x)&=d_{1,-}(x),  \qquad  && x\in(-\infty,0),   \\  \label{eq:dJump5}
d_{2,+}(x)&=d_{3,-}(x), \qquad d_{3,+}(x)=d_{2,-}(x), \qquad  && x\in(-\infty,0)  .
\end{align}

\item [\rm (ii)] As $z \to \infty$, we have
\begin{equation}\label{eq:dk-infty}
\begin{array}{ll}
d_1(z)=1+(1-\omega)\beta z^{-\frac{1}{3}}+(-\frac{3}{2}\omega \beta^2+\frac{1-\omega^2}{2}\beta) z^{-\frac{2}{3}}+\Boh(z^{-1}),\\
d_{2}(z)=1+(1-\omega)\omega^2\beta z^{-\frac{1}{3}} +(-\frac{3}{2}\omega \beta^2+\frac{1-\omega^2}{2}\beta)\omega z^{-\frac{2}{3}} +\Boh(z^{-1}),
\\
d_{3}(z)=1+(1-\omega)\omega \beta z^{-\frac{1}{3}}+(-\frac{3}{2}\omega \beta^2+\frac{1-\omega^2}{2}\beta)\omega^2 z^{-\frac{2}{3}}+\Boh(z^{-1})
.
\end{array}
\end{equation}

\item [\rm (iii)] As $z\to 1$, we have, for $\Im z>0$,
\begin{equation}\label{eq:dk-1}
\begin{array}{ll}
d_1(z)=(3\sqrt{3})^{\beta}e^{-\frac{\beta\pi i}{6}} (z-1)^{-\beta}\left(1 +\frac{\beta}{{3\sqrt{3}}} (\sqrt{3}+e^{\frac{\pi i}{6}})(z-1)+\Boh((z-1)^2)\right),\\
d_{2}(z)=(3\sqrt{3})^{-\beta}e^{-\frac{\beta\pi i}{6}} (z-1)^{\beta}\left(1  -\frac{\beta}{{3\sqrt{3}}}(\sqrt{3}+e^{-\frac{\pi i}{6}})(z-1)+\Boh((z-1)^2)\right),
\\
d_{3}(z)=e^{\frac{\beta\pi i}{3}} \left(1-\frac{ \beta i}{3\sqrt{3}}(z-1)+\Boh((z-1)^2)\right)
.
\end{array}
\end{equation}

\item [\rm (iv)] As $z\to 0$, we have, for $\Im z>0$,
\begin{equation}\label{eq:dk-0}
\begin{array}{ll}
d_1(z)=e^{-\frac{4\beta\pi i}{3}} \left(1+\sqrt{3}e^{\frac{\pi i}{6}} \beta z^{\frac13}+\frac{\sqrt{3}}{2}\beta(e^{-\frac{\pi i}{6}}-\sqrt{3}\beta \omega^2)z^{\frac23}+\Boh(z)\right),\\
d_{2}(z)= e^{\frac{2\beta\pi i}{3}}  \left(1-\sqrt{3}e^{-\frac{\pi i}{6}} \beta z^{\frac13}-\frac{\sqrt{3}}{2}\beta(e^{\frac{\pi i}{6}}+\sqrt{3}\beta \omega)z^{\frac23}+\Boh(z)\right),
\\
d_{3}(z)=e^{\frac{2\beta\pi i}{3}} \left(1-\sqrt{3}\beta i z^{\frac13}+\frac{\sqrt{3}}{2}\beta(i-\sqrt{3}\beta)z^{\frac23}+\Boh(z)\right).
\end{array}
\end{equation}

\end{enumerate}
\end{proposition}

\begin{proof}
It follows from the definition of $\lambda(\xi)$ in \eqref{eq:lambda} that
\begin{equation}\label{eq:lambdaJump}
\lambda_+(\xi)=e^{-2\beta \pi i}\lambda_{-}(\xi)=\frac{\lambda_{-}(\xi)}{1-\gamma}, \qquad \xi  \in ( 0,1 ) \cup (0, 1) \omega,
\end{equation}
where the orientation of the contour is shown in Figure \ref{fig:lambdaContour}, and as $\xi \to \infty$,
\begin{equation}
\lambda(\xi) =1+(1-\omega)\beta \frac{1}{\xi}+\left(-\frac{3}{2}\omega \beta^2+\frac{1-\omega^2}{2}\beta \right) \frac{1}{\xi^2}+\Boh(\xi^{-3}).
\end{equation}
The above formulas, together with the definitions of $d_k$, $k=1,2,3$, in \eqref{eq:dklambda}, give us the jump condition \eqref{eq:dJump1}--\eqref{eq:dJump5} and the large-$z$ behavior \eqref{eq:dk-infty}. Similarly, we obtain the asymptotics of $d_k$ near the points 0 and 1 via a straightforward computation.

This completes the proof of Proposition \ref{prop:propdk}.
\end{proof}

As a consequence of the above proposition, we have the following result.

\begin{lemma}\label{lem:NSolution}
The solution of RH problem \ref{rhp:N} is given by
\begin{equation}\label{eq:NSolution}
N(z)=z^{-\frac{\alpha}{3}}C_N \diag \left(  z^{\frac13},1,z^{-\frac13}  \right)L_{\pm}\diag\left( e^{\pm \frac{\alpha \pi i}{3}}d_1(z), e^{\mp \frac{\alpha \pi i}{3}}d_2(z),d_3(z)\right)
\end{equation}
for $\pm \Im z>0$,
where
\begin{equation}\label{eq:CN}
C_N =
\begin{pmatrix}
1 & -\sqrt{3}\beta i &-\frac{3}{2}\beta^2+\frac{\sqrt{3}}{2}\beta i \\0 &1 &-\sqrt{3} \beta i \\ 0&0&1
\end{pmatrix},
\end{equation}
the functions $d_k(z)$, $k=1,2,3$ and the constant matrices $L_{\pm}$ are given in \eqref{eq:dklambda} and \eqref{def:Lpm}, respectively.
\end{lemma}

\begin{proof}
By setting
\begin{equation}\label{def:Ninfty}
N_{\infty}(z)=z^{-\frac{\alpha}{3}} \diag \left( z^{\frac13},1,z^{-\frac13}  \right)L_{\pm}\diag \left( e^{\pm \frac{\alpha \pi i}{3}}, e^{\mp \frac{\alpha \pi i}{3}}, 1 \right),\quad \pm \Im z>0,
\end{equation}
it is readily seen that we could rewrite $N(z)$ in \eqref{eq:NSolution} as
\begin{equation}\label{eq:NNinfty}
N(z)= C_N N_\infty(z)\diag\left(d_1(z),d_2(z),d_3(z)\right).
\end{equation}
First, one could check directly that
\begin{equation}\label{eq:NinftyJump}
N_{\infty,+}(x)=N_{\infty,-}(x)\left\{
 \begin{array}{ll}
          \begin{pmatrix} 1 &0&0 \\ 0&0&-e^{-\alpha \pi i} \\ 0&e^{-\alpha \pi i}&0 \end{pmatrix}, & \qquad  \hbox{$x\in(-\infty,0)$,} \\
           \begin{pmatrix} 0 & 1 &0 \\ -1 & 0 &0 \\ 0&0&1 \end{pmatrix}, & \qquad \hbox{$x\in(0,+\infty)$.}
                  \end{array}
      \right.
\end{equation}
Next, with the help of the relations of $d_k(z)$ established in \eqref{eq:dJump1}--\eqref{eq:dJump5}, one can see $N(z)$ in \eqref{eq:NNinfty} indeed satisfies the desired  jump conditions \eqref{eq:N-jump} and \eqref{def:JN}. Finally, substituting the large-$z$ asymptotics of $d_k(z)$  \eqref{eq:dk-infty} into \eqref{eq:NSolution}, we obtain the asymptotic behavior \eqref{eq:Ninfty}.

This completes the proof of Lemma \ref{lem:NSolution}.
\end{proof}


Although $J_S(z) \to I$ as $s\to + \infty$ for $z\in \Sigma_S \setminus \mathbb{R}$, the convergence is not uniform near $z=0$ and $z=1$. The next two subsections are then devoted to the construction of local parametrices near these two points.

\subsection{Local parametrix near the origin}

Let $D(z_0,\delta)$ be a fixed open disc centered at $z_0$ with radius $\delta>0$ and denote by $\partial D(z_0,\delta)$ its boundary. We look for a function $P^{(0)}(z)$ satisfying the following RH problem in $D(0,\delta)$.

\begin{rhp}\label{rhp:P-0}
\hfill
\begin{enumerate}
\item[\rm (1)]  $P^{(0)}(z)$ is defined and analytic in $D(0,\delta)\setminus \Sigma_S$,
where $\delta$ is a small positive number and $\Sigma_S$ is defined in \eqref{def:SigmaS}.

\item[\rm (2)]  $P^{(0)}(z)$ satisfies the jump condition
\begin{equation}\label{eq:P0-jump}
P^{(0)}_+(z)=P^{(0)}_-(z)J_S(z), \qquad z\in D(0,\delta) \cap \Sigma_S,
\end{equation}
where $J_S(z)$ is defined in \eqref{def:JS}.

\item[\rm (3)]As $s\to +\infty$, we have the matching condition
\begin{equation}\label{eq:MatchingCond0}
P^{(0)}(z)=\left(I+\Boh(s^{-\frac13}) \right) N(z), \qquad  z\in \partial D(0,\delta),
\end{equation}
where $N(z)$ is given in \eqref{eq:NSolution}.
\end{enumerate}
\end{rhp}

We solve the above RH problem explicitly in terms of the solution $\Psi_{\alpha}(z)$ of the RH problem  \ref{rhp: Pearcey} in the following lemma.
\begin{lemma}\label{lem:P0Solution}
The solution to the RH problem \ref{rhp:P-0} is given by
\begin{equation}\label{eq:P0Solution}
P^{(0)}(z)= s^{\frac{\alpha}{3}}  E_0(z) \diag(s^{-\frac13}, 1,s^{\frac13}) \Psi_{\alpha}(sz)e^{-\Theta(sz)} \diag\left((1-\gamma)^{-\frac23},(1-\gamma)^{\frac13},(1-\gamma)^{\frac13}\right),
\end{equation}
where
\begin{equation}\label{eq:E0}
E_0(z)=  - \sqrt{3}\,i  N(z) \diag\left((1-\gamma)^{\frac23},(1-\gamma)^{-\frac13},(1-\gamma)^{-\frac13}\right)N_{\infty}(z)^{-1} ,
\end{equation}
with $\Theta(z)$, $N(z)$ and  $N_{\infty}(z)$ defined in \eqref{def:Theta}, \eqref{eq:NSolution} and \eqref{def:Ninfty}, respectively.
\end{lemma}

\begin{proof}
We see from \eqref{eq:NinftyJump} that  $N_{\infty}(z) \diag\left((1-\gamma)^{-\frac23},(1-\gamma)^{\frac13},(1-\gamma)^{\frac13}\right)$ satisfies the same jump relation \eqref{def:JN} as $N(z)$ for $z\in(-\delta,0)\cup (0,\delta)$. This means that  $E_0(z)$ given in \eqref{eq:E0} is analytic for $z \in D(0,\delta)\setminus\{0\}$. It follows from the definitions of $N(z)$ and  $N_{\infty}(z)$ in \eqref{eq:NSolution} and \eqref{def:Ninfty} that
\begin{align*}
& N(z) \diag\left((1-\gamma)^{\frac23},(1-\gamma)^{-\frac13},(1-\gamma)^{-\frac13}\right)N_{\infty}(z)^{-1}  = C_N \diag\left(z^{\frac{1}{3}}, 1 , z^{-\frac{1}{3}} \right) \\
& \qquad \qquad \times L_\pm \diag\left((1-\gamma)^{\frac23} d_1(z),(1-\gamma)^{-\frac13}d_2(z),(1-\gamma)^{-\frac13}d_3(z) \right) L_\pm^{-1} \diag\left(z^{-\frac{1}{3}}, 1 , z^{\frac{1}{3}} \right).
\end{align*}
On account of the local behaviors of $d_k(z)$, $k=1,2,3$, given in \eqref{eq:dk-0},  it is readily seen that $E_0(z)=\Boh(z^{-2/3})$ as $z\to 0$, which implies that $E_0(z)$ is actually analytic in $D(0,\delta)$. Recalling the jump of $\Psi_\alpha$ in \eqref{jumps:M}, it is then straightforward to see that $P^{(0)}(z)$ in \eqref{eq:P0Solution} satisfies the required jump condition \eqref{eq:P0-jump}. Finally, we observe that one may rewrite the large-$z$ asymptotics of $ \Psi_\alpha$ in \eqref{eq:asyPsi} as
\begin{multline}\label{eq:asyPsi-1/3}
 \Psi_\alpha(z) = \frac{i z^{-\frac{\alpha}{3}}}{ \sqrt{3} } \diag \left(z^{\frac13},1,z^{-\frac13} \right)
\\
\times \left(I+ \frac{ \Psi_1(\rho)}{z^{1/3}}+ \frac{\Psi_2(\rho)}{z^{2/3}}+ \Boh(z^{-1}) \right) L_{\pm}  \diag \left( e^{\pm \frac{\alpha \pi i}{3}}, e^{\mp \frac{\alpha \pi i}{3}}, 1 \right) e^{\Theta(z)},
\end{multline}
with
\begin{align}
\Psi_1(\rho) &=
\begin{pmatrix}
0 & \pi_3(\rho) & 0
\\
0 & 0 & \pi_3(\rho) + \frac{\rho}{3}
\\
\pi_3(\rho) + \frac{2\rho}{3} & 0 & 0
\end{pmatrix}, \label{def:psi1rho}
\\
\Psi_{2}(\rho)& = \begin{pmatrix} 0 & 0 & \pi_6(\rho) \\ \pi_6(\rho) + \frac{\rho}{3} \pi_3(\rho) + \frac{1 - \alpha}{3} & 0 & 0 \\ 0 & \pi_6(\rho) + \frac{2 \rho}{3} \pi_3(\rho)+ \frac{\rho^2}{9}+ \frac{1 - 2 \alpha}{3}  & 0  \end{pmatrix}, \label{def:psi2rho}
\end{align}
where $\pi_3(\rho)$ and $\pi_6(\rho)$ are polynomials defined in \eqref{pi3-def} and \eqref{pi6-def}. This implies that, as $z\to \infty$,
$$
 \Psi_{\alpha}(z)=\frac{i}{\sqrt{3}}N_\infty(z)\left(I+\Boh(z^{-\frac13})\right)e^{\Theta(z)}
$$
with $N_\infty(z)$ given in \eqref{def:Ninfty}. A combination of \eqref{eq:P0Solution}, \eqref{eq:E0} and the above formula gives us
\begin{align*}
P^{(0)}(z) N(z)^{-1} &= N(z) \diag\left((1-\gamma)^{\frac23},(1-\gamma)^{-\frac13},(1-\gamma)^{-\frac13}\right) \left(I+\Boh(s^{-\frac13}) \right) \\
& \quad \times \diag\left((1-\gamma)^{-\frac23},(1-\gamma)^{\frac13},(1-\gamma)^{\frac13}\right) N(z)^{-1}, \qquad s \to +\infty,
\end{align*}
for $z\in \partial D(0,\delta)$. As $N(z)$ is independent of $s$, the above formula is indeed equivalent to the matching condition \eqref{eq:MatchingCond0}.


This completes the proof of Lemma \ref{lem:P0Solution}.
\end{proof}

In view of \eqref{eq:dk-0}, the following proposition is immediate.
\begin{proposition} \label{prop:E0expansion}
With the function $E_0(z)$ defined in \eqref{eq:E0}, we have
\begin{equation}
E_0(z) = E_0(0) \Big(I + E_0(0)^{-1} E_0'(0) z + \Boh(z^2) \Big), \qquad z \to 0,
\end{equation}
where
\begin{equation}\label{eq:E0zero}
 E_0(0)= -\sqrt{3} \, i ~ C_N C_N^t
\end{equation}
with $C_N$ given in \eqref{eq:CN}, and
\begin{align}
 & E_0(0)^{-1} E_0'(0) \label{eq:E0zero-prime}\\
 & = \frac{\beta}{2} \begin{pmatrix}
 \beta(3 + \sqrt{3} \beta i) & \sqrt{3} \, i (1+\sqrt{3} \beta i) & -2\sqrt{3} \, i \\
- \frac{1}{4}(9 \beta^3 - 6\sqrt{3}\beta^2 i + 3 \beta +2\sqrt{3} \, i) & -2\sqrt{3} \beta^2 i & \sqrt{3}\,i (1-\sqrt{3} \beta i) \\
 \beta(\beta^2 + 1) - \frac{2\sqrt{3} \, i}{15}(6\beta^4 + 15\beta^2 -1) & \frac{1}{4}(9 \beta^3 + 6\sqrt{3}\beta^2 i + 3 \beta -2\sqrt{3} \, i) & -\beta (3-\sqrt{3} \beta i)
 \end{pmatrix}. \nonumber
\end{align}
\end{proposition}
To derive \eqref{eq:E0zero-prime}, we actually need more terms than those given in \eqref{eq:dk-0}. Since all the calculations are straightforward, we omit the details here.

\subsection{Local parametrix near $1$}
Near $z=1$, we intend to find a function $P^{(1)}(z)$ satisfying the following RH problem.
\begin{rhp}\label{rhp:P1}
\hfill
\begin{enumerate}
\item[\rm (1)]
$P^{(1)}(z)$ is defined and analytic in $D(1,\delta)\setminus \Sigma_S$, where $\Sigma_S$ is defined in \eqref{def:SigmaS}.

\item[\rm (2)]
For $z\in \Sigma_S \cap D(1,\delta)$, $P^{(1)}(z)$ satisfies the jump condition
\begin{equation}\label{eq:P1-jump}
P^{(1)}_+(z)=P^{(1)}_-(z)J_S(z),
\end{equation}
where $J_S(z)$ is given in \eqref{def:JS}.

\item[\rm (3)]
As $ s \to +\infty$, we have the matching condition
\begin{equation}\label{eq:MatchingCondP1}
P^{(1)}(z)=\left(I+\Boh(s^{-\frac23}) \right)N(z),\qquad z\in \partial D(1,\delta),
\end{equation}
where $N(z)$ is given in \eqref{eq:NSolution}.
\end{enumerate}
\end{rhp}

As in \cite{DXZ202}, we can construct $P^{(1)}(z)$ explicitly by using the confluent hypergeometric parametrix. To proceed, we introduce the following local conformal mapping near $z=1$:
\begin{align}\label{eq:f}
f(z) &:= -is^{-\frac23}[(\theta_2(sz)-\theta_1(sz))-(\theta_2(s)-\theta_1(s))]
\nonumber \\
& = \frac{3\sqrt{3}}{2} \left( z^{\frac23}-1\right)-\frac{\sqrt{3}\rho}{s^{1/3}}\left(z^{\frac13}-1 \right), \qquad z \in D(1,\delta).
\end{align}
It is easily seen that
\begin{equation}\label{eq:f'}
f(1)=0,\qquad f'(1)=\sqrt{3}-\frac{\sqrt{3}}{3s^{1/3}}\rho, \qquad f''(1)=-\frac{1}{\sqrt{3}}+\frac{2\rho}{3\sqrt{3}s^{1/3}},
\end{equation}
where $f'(1)$ is positive when $s$ is large enough. Let $\Phi^{(\CHF)}(z;\beta)$ be the confluent hypergeometric parametrix given in Appendix \ref{sec:CHF} with the parameter $\beta$ given in \eqref{beta-def}. We then define, for $z\in D(1,\delta)\setminus \Sigma_S$,
\begin{align}\label{eq:P1Solution}
P^{(1)}(z)&=E_1(z) \diag \left( \Phi^{(\CHF)}(s^{\frac23}f(z);\beta)e^{-\frac\beta2 \pi i\sigma_3},1 \right)
\nonumber\\
&~~~\times\left\{
\begin{array}{ll}
\diag\left(e^{\frac{\theta_2(sz)-\theta_1(sz)}{2}\sigma_3},1\right), & \quad  \Im z>0,\\
\diag\left(e^{\frac{\theta_1(sz)-\theta_2(sz)}{2}\sigma_3},1\right), & \quad  \Im z<0,
\end{array}
\right.
\end{align}
where $\sigma_3=
\begin{pmatrix}
1 & 0
\\
0 & -1
\end{pmatrix}$, $f(z)$ is given in \eqref{eq:f} and
\begin{equation}\label{def:E1}
E_1(z)=\left\{
         \begin{array}{ll}
           N(z)\diag\left(\left(e^{\frac{\theta_1(s)-\theta_2(s)+\beta \pi i}{2}}s^{\frac{2\beta}{3}}f(z)^{\beta}\right)^{\sigma_3},1\right), & \quad \hbox{$\Im z>0$}, \\
           N(z)\diag\left(
\begin{pmatrix}
0 & 1
\\
-1 & 0
\end{pmatrix}
\left(e^{\frac{\theta_1(s)-\theta_2(s)+\beta \pi i}{2}}s^{\frac{2\beta}{3}}f(z)^{\beta}\right)^{\sigma_3},1 \right), & \quad \hbox{$\Im z<0$},
         \end{array}
       \right.
\end{equation}
with $N(z)$ given in \eqref{eq:NSolution}. Here, we use the notation
$$\diag \left(A,1\right)=\begin{pmatrix}(A)_{11} & (A)_{12} & 0
\\
(A)_{21} & (A)_{22} & 0
\\
0 & 0 & 1
\end{pmatrix}$$
for a $2\times 2$ matrix $A$.

From the RH problem for $\Phi^{(\CHF)}(z;\beta)$ and our construction of $P^{(1)}(z)$ in \eqref{eq:P1Solution}, the following lemma is immediate.
We omit the proof but refer to \cite[Proposition 5.8]{DXZ202} for a similar situation.
\begin{lemma}\label{lem:P1}
The local parametrix $P^{(1)}(z)$  defined in \eqref{eq:P1Solution} solves the RH problem \ref{rhp:P1} for large positive $s$. Moreover, the prefactor $E_1(z)$ therein is analytic in $D(1,\delta)$.
\end{lemma}

We conclude this section with the local behavior of $E_1(z)$ near $z = 1$.
\begin{proposition}\label{prop:E1expansion}
With the function $E_1(z)$ defined in \eqref{def:E1}, we have
\begin{equation}\label{eq:E-expand}
E_1(z)=E_1(1)\left(I+E_1(1)^{-1}E_1'(1)(z-1)+\Boh((z-1)^2) \right), \qquad  z \to 1,
\end{equation}
where
\begin{equation}\label{eq:E-expand-coeff-1}
E_1(1)=C_N L_{+}\diag\Big(e^{\frac{\beta\pi i }{3}}c_1(s), e^{-\frac{2\beta\pi i }{3}}c_1(s)^{-1}, e^{\frac{\beta \pi i }{3}} \Big),
\end{equation}
and
\begin{align}
& E_1(1)^{-1}E_1'(1) \nonumber \\
& = \frac{1}{\sqrt{3}}\begin{pmatrix} \frac{\sqrt{3}\beta   f''(1)}{2f'(1)} - \frac{\beta }{3} e^{-\frac{1}{6}\pi i}  + \frac{2\beta}{\sqrt{3}}  &  \frac{i}{3} e^{-\beta \pi i} c_1(s)^{-2} &-  \frac{i}{3} c_1(s)^{-1}\\
-  \frac{i}{3}e^{\beta \pi i } c_1(s)^{2} &-\frac{\sqrt{3}\beta   f''(1)}{2f'(1)} + \frac{\beta }{3} e^{\frac{1}{6}\pi i}  - \frac{2\beta}{\sqrt{3}}  &  \frac{i}{3} e^{\beta \pi i}  c_1(s) \\
 \frac{i}{3}c_1(s)  & - \frac{i}{3}  e^{-\beta \pi i } c_1(s)^{-1}& -\frac{i }{3}\beta  \end{pmatrix}-\frac{\alpha}{3}I
\label{eq:E-expand-coeff-2}
\end{align}
with
\begin{equation}\label{eq:c-1}
c_1(s):=(3 \sqrt{3})^{\beta}f'(1)^{\beta}e^{\frac{\alpha \pi i}{3}}e^{\frac{\theta_1(s)-\theta_2(s)}{2}}s^{\frac{2\beta}{3}},
\end{equation}
$C_N$, $L_+$ and $f$ given in \eqref{eq:CN},  \eqref{def:Lpm} and \eqref{eq:f}, respectively.
\end{proposition}
\begin{proof}
The coefficients in $\eqref{eq:E-expand}$  follow directly from  \eqref{eq:dk-1}, \eqref{def:E1} and straightforward calculations. In particular, to derive \eqref{eq:E-expand-coeff-2}, one needs to use the fact that
$$L_+^{-1}\diag(1,0,-1)L_+=\frac{i}{\sqrt{3}}\begin{pmatrix} 0&1 & -1
\\
-1 &0& 1
\\
1 & -1 & 0
\end{pmatrix}.
$$
This completes the proof of Proposition \ref{prop:E1expansion}.
\end{proof}
\subsection{Final transformation and the small norm RH problem}
The final transformation is defined by
\begin{equation}\label{def:StoR}
 R(z)=\left\{
                \begin{array}{ll}
 S(z)N(z)^{-1}, &\qquad  \hbox{$z\in \mathbb{C}\setminus \{D(0,\delta)\cup D(1,\delta)\cup  \Sigma_S\}$,}
 \\
        S(z)P^{(0)}(z)^{-1},  &\qquad  \hbox{$z \in D(0,\delta)$,}\\
        S(z)P^{(1)}(z)^{-1},  &\qquad  \hbox{$z \in D(1,\delta)$.}
                \end{array}
              \right.
\end{equation}
The RH problems for $N$, $P^{(0)}$ and $P^{(1)}$ particularly show that $R$ is analytic for $z\in D(0,\delta)\cup D(1,\delta)\setminus \{0,1\}$. One can then check that both $0$ and $1$ are actually removable singularities of $R$. Then,  $R$ satisfies the following RH problem.
\begin{rhp}\label{rhp:R}
\hfill
\begin{enumerate}
\item[\rm (1)]  $R(z)$ is analytic in $\mathbb{C} \setminus \Sigma_{R}$,
where
\begin{equation}
\Sigma_R:=\Sigma_S \cup \partial D(0,\delta) \cup \partial D(1,\delta) \setminus \{\mathbb{R}
\cup D(0,\delta) \cup D(1,\delta) \}.
\end{equation}

\item[\rm (2)]  $R(z)$ satisfies the jump condition  $$ R_+(z)=R_-(z)J_R(z), \qquad z\in \Sigma_R,$$
where
  \begin{equation}\label{def:JR}
                     J_{R}(z)=\left\{
                                      \begin{array}{ll}
                                        P^{(0)}(z) N(z)^{-1}, & \hbox{ $z \in \partial D(0, \delta)$,} \\
                                        P^{(1)}(z) N(z)^{-1}, & \hbox{ $z \in \partial D(1, \delta)$,} \\
                                        N(z) J_S(z) N(z)^{-1}, & \hbox{ $ z \in \Sigma_{R} \setminus
                                                    \{D(0,\delta) \cup D(1,\delta) \}$,}
                                      \end{array}
\right.
  \end{equation}
  and where the orientations of $\partial D(0, \delta)$ and $\partial D(1, \delta)$ are taken in a clockwise manner.
\item[\rm (3)] As $z \to \infty$, we have
$R(z)=I+\Boh(z^{-1}).$
\end{enumerate}
\end{rhp}

For $z \in \Sigma_{R} \setminus \{\partial D(0,\delta) \cup \partial D(1,\delta) \}$, it is readily seen from \eqref{def:JS}, \eqref{eq:NSolution} and  \eqref{def:JR} that there exists a positive constant $c>0$ such that
\begin{equation}\label{eq:JRexpsmall}
 J_R(z) = I + \Boh \left( e^{-c s^{2/3}} \right), \qquad s \to \infty.
\end{equation}
For $z \in \partial D(0, \delta) \cup \partial D(1, \delta)$, it follows from \eqref{eq:MatchingCond0}, \eqref{eq:MatchingCondP1} and \eqref{def:JR} that
\begin{equation} \label{Jr-exp}
 J_R(z) = I +\frac{J_1(z)}{s^{1/3}} + \frac{J_2(z)}{s^{2/3}}+ \Boh (s^{-1}),  \qquad s \to \infty,
\end{equation}
for some matrix-valued functions $J_1(z)$ and $J_2(z)$. The above approximations show that $R$ satisfies a small norm RH problem when $s$ is large, which leads to the following approximation of $R$.

\begin{proposition} \label{prop-R-exp}
As $s\to \infty$, we have
\begin{equation} \label{R-large-s-exp}
 R(z) = I +\frac{R_1(z)}{s^{1/3}}+ \frac{R_2(z)}{s^{2/3}} + \Boh (s^{-1}),
\end{equation}
uniformly for $z \in \mathbb{C} \setminus \Sigma_R,$
where
\begin{equation} \label{R1-explicit}
R_1(z) =  \frac{1}{z}  \Res_{\zeta = 0} J_1^{(0)}(\zeta)  -\begin{cases}   J_1^{(0)}(z), & z\in D(0,\delta), \\
  0, &  z\in \mathbb{C}\setminus D(0,\delta),
\end{cases}
\end{equation}
and
\begin{multline} \label{R2-explicit}
R_2(z) =  \frac{1}{z} \Res_{\zeta = 0}\left[J_2^{(0)}(\zeta) + R_{1}(\zeta) J_1^{(0)} (\zeta) \right] + \frac{1}{z-1}  \Res_{\zeta = 1} J_2^{(1)}(\zeta)
\\ - \begin{cases} J_2^{(0)}(z) + R_{1}(z) J_1^{(0)} (z) ,& z\in D(0,\delta), \\ J_2^{(1)}(z) ,& z\in D(1,\delta), \\ 0,& \textrm{otherwise}.\end{cases}
\end{multline}
In \eqref{R1-explicit} and \eqref{R2-explicit}, the functions $J_k^{(0)}(z)$, $k=1,2$, and $J_1^{(1)}(z)$ are given by
\begin{align} \label{jk-formula-0}
J_k^{(0)}(z) &= z^{-\frac{k}{3}} E_0(z) \diag \left(z^{\frac13},1,z^{-\frac13} \right) \Psi_{k}(\rho) \nonumber
\\
& \quad \times \diag \left(z^{-\frac13},1,z^{\frac13} \right) E_0(z)^{-1}, \qquad z \in D(0, \delta)\setminus \{0\},
\end{align}
and
\begin{equation} \label{jk-formula-1}
J_2^{(1)}(z) = \frac{2}{3\sqrt{3} (z^{\frac{2}{3}} - 1)} E_1(z)\diag\Big(\Phi^{(\CHF)}_1, 0 \Big) E_1(z)^{-1} , \qquad z \in  D(1, \delta)\setminus\{1\},
\end{equation}
where $E_0$, $\Psi_{k}$, $f$, $E_1$ and $\Phi^{(\CHF)}_1$  are defined in  \eqref{eq:E0}, \eqref{def:psi1rho}, \eqref{def:psi2rho}, \eqref{eq:f}, \eqref{def:E1} and \eqref{Phi1-def}, respectively.
\end{proposition}

\begin{proof}
By a standard argument for the small norm RH problem (cf. \cite{DeiftBook,DZ93}), the expansion \eqref{R-large-s-exp} follows directly from asymptotics of the jump $J_R(z)$ given in \eqref{eq:JRexpsmall} and \eqref{Jr-exp}.

To obtain explicit expressions of $R_1(z)$ and $R_2(z)$, we need more information about the functions $J_1(z)$ and $J_2(z)$ in the expansion \eqref{Jr-exp}. For $z \in \partial D(0, \delta)$ and $\Im z > 0$, combining the large-$z$ expansion of $\Psi_\alpha$ in \eqref{eq:asyPsi-1/3} and the definitions of $P^{(0)} (z)$ and $E_0(z)$ in \eqref{eq:P0Solution} and \eqref{eq:E0},  we have, as $s\to \infty$
\begin{align*}
&P^{(0)}(z) N(z)^{-1}
\\
& =  \frac{i z^{-\alpha/3} } {\sqrt{3}} E_0(z)  \diag \left(z^{\frac13},1,z^{-\frac13} \right)  \Big( I + \frac{\Psi_1(\rho)}{(sz)^{1/3}} +  \frac{\Psi_2 (\rho)}{(sz)^{2/3}} +  \Boh (s^{-1}) \Big) L_+ \nonumber \\
&~~~~  \times \diag \left( e^{\frac{\alpha \pi i}{3}}, e^{-\frac{\alpha \pi i}{3}}, 1 \right) \diag\left((1-\gamma)^{-\frac23},(1-\gamma)^{\frac13},(1-\gamma)^{\frac13}\right) N(z)^{-1} \nonumber \\
&  = E_0(z)  \diag \left(z^{\frac13},1,z^{-\frac13} \right)  \Big( I + \frac{\Psi_1 (\rho)}{(sz)^{1/3}} +  \frac{\Psi_2 (\rho)}{(sz)^{2/3}} +  \Boh (s^{-1}) \Big) \diag \left(z^{-\frac13},1,z^{\frac13} \right) E_0(z)^{-1},
\end{align*}
where $\Psi_{1}(\rho)$ and $\Psi_{2}(\rho)$ are given in \eqref{def:psi1rho} and \eqref{def:psi2rho}, respectively. It is easy to verify that the above expansion also holds for $z \in \partial D(0, \delta)$ and $\Im z < 0$. As $E_0(z)$ is independent of $s$, the above formula gives us
\begin{equation} \label{eq:MatchingCond0-more}
P^{(0)}(z) N(z)^{-1} = I +\frac{J_1^{(0)}(z)}{s^{1/3}} + \frac{J_2^{(0)}(z)}{s^{2/3}}+ \Boh (s^{-1}), \qquad s\to \infty,
\end{equation}
for $z \in \partial D(0, \delta)$ with $J_k^{(0)}(z)$, $k=1,2$, given in \eqref{jk-formula-0}.

Similarly, for $z \in \partial D(1, \delta) $ and $\Im z > 0$, we have from \eqref{eq:P1Solution} and \eqref{H at infinity} that
\begin{multline}
P^{(1)}(z) N(z)^{-1}  = E_1(z) \left( I + \frac{\diag\left(\Phi^{(\CHF)}_1, 0 \right)}{s^{2/3} f(z)} +    \Boh (s^{-\frac{4}{3}}) \right)
\\ \times \diag\left(\left(e^{-\frac{\beta \pi i}{2}-\frac{s^{2/3}f(z)i}{2}}s^{-\frac{2\beta}{3}}f(z)^{-\beta}\right)^{\sigma_3},1\right) \diag\left(  e^{\frac{\theta_2(sz)-\theta_1(sz)}{2} \sigma_3}  ,1 \right)  N(z)^{-1}.  \nonumber
\end{multline}
Recalling the definitions of $f(z)$ and $E_1(z)$ in \eqref{eq:f} and \eqref{def:E1}, we have
\begin{equation*}
P^{(1)}(z) N(z)^{-1} = E_1(z) \left( I + \frac{2 \diag\left(\Phi_1^{(\CHF) }  , 0 \right)}{ 3\sqrt{3} (z^{2/3} - 1) s^{2/3} } +    \Boh (s^{-1}) \right) E_1(z)^{-1}.
\end{equation*}
Again, the above formula holds for $z \in \partial D(1, \delta) $ and $\Im z < 0$ as well. Although $E_1(z)$ depends on $s$, it is bounded when $s \to \infty $ since  both $\theta_1(s) - \theta_2(s) = (-\frac{3 \sqrt{3}}{2}  s^{2/3} + \sqrt{3}\rho  s^{1/3})i$ and $\beta = \frac{1}{2 \pi i} \ln(1-\gamma)$ are purely imaginary. Then, we have from the above formula that
\begin{equation} \label{eq:MatchingCond1-more}
P^{(1)}(z) N(z)^{-1} = I +\frac{J_2^{(1)}(z)}{s^{2/3}} + \Boh (s^{-1}), \qquad s\to \infty,
\end{equation}
for $z \in \partial D(1, \delta)$ with $J_2^{(1)}(z)$ given in \eqref{jk-formula-1}. Therefore, it follows from \eqref{def:JR}, \eqref{Jr-exp}, \eqref{eq:MatchingCond0-more} and \eqref{eq:MatchingCond1-more} that
\begin{equation}
J_1(z) = \begin{cases} J_1^{(0)} (z), &   z \in \partial D(0, \delta), \\
0, &   z \in \partial D(1, \delta), \end{cases} \qquad
J_2(z) = \begin{cases} J_2^{(0)} (z), &   z \in \partial D(0, \delta), \\
J_2^{(1)} (z), &   z \in \partial D(1, \delta). \end{cases}
\end{equation}

Now, we are ready to derive explicit expressions of $R_1(z)$ and $R_2(z)$ in \eqref{R-large-s-exp}. From the RH problem \ref{rhp:R} for $R$, it is easy to see from \eqref{eq:JRexpsmall}--\eqref{R-large-s-exp} that $R_1(z)$ satisfies an RH problem as follows:
\begin{rhp}\label{rhp:R1}
\hfill
\begin{enumerate}
\item[\rm (1)]  $R_1(z)$ is analytic in $\mathbb{C} \setminus \partial D(0, \delta)$.

\item[\rm (2)]  For $z \in \partial D(0, \delta)$, we have
  \begin{equation}
  R_{1,+}(z)  - R_{1,-}(z) = J_1(z) = J_1^{(0)}(z),
  \end{equation}
 where $J_1^{(0)}(z)$ is given in \eqref{jk-formula-0}.
\item[\rm (3)] As $z \to \infty$, we have $R_1(z)=\Boh(z^{-1}).$
\end{enumerate}
\end{rhp}
By Cauchy's theorem, we have
\begin{equation}\label{eq:R1integral}
R_1(z) = \frac{1}{2\pi i}\int_{\partial D(0,\delta)}\frac{J_1^{(0)}(\zeta)}{\zeta - z}\ud \zeta.
\end{equation}
To evaluate this integral, we observe from \eqref{def:psi1rho} that
\begin{multline}\label{eq:zpsi1rho}
z^{-\frac{1}{3}} \diag \left(z^{\frac13},1,z^{-\frac13} \right) \Psi_1(\rho) \diag \left(z^{-\frac13},1,z^{\frac13} \right)
\\
 =  \begin{pmatrix}
0 & \pi_3(\rho) & 0
\\
0 & 0 & \pi_3(\rho) + \frac{\rho}{3}
\\
\left(\pi_3(\rho)  + \frac{2\rho}{3} \right)/z & 0 & 0
\end{pmatrix},
\end{multline}
where $\pi_3(\rho)$ is given in \eqref{pi3-def}. This, together with \eqref{jk-formula-0} and the fact that $E_0(z)$ is analytic near $z =0$, implies that $J_1(z)$ admits a simple pole at $z=0$. We then obtain \eqref{R1-explicit} from \eqref{eq:R1integral} and the residue theorem.

The RH problem for $R_2$ reads as follows.
\begin{rhp}\label{rhp:R2}
\hfill
\begin{enumerate}
\item[\rm (1)]  $R_2(z)$ is analytic in $\mathbb{C} \setminus \{ \partial D(0, \delta) \cup \partial D(1, \delta) \}$.

\item[\rm (2)]  For $z \in \partial D(0, \delta) \cup \partial D(1, \delta)$, we have
  \begin{align}
  R_{2,+}(z)  - R_{2,-}(z) &=  J_2(z) + R_{1,-}(z) J_1(z)
  \nonumber
  \\
  &= \begin{cases} J_2^{(0)}(z) + R_{1,-}(z) J_1^{(0)}(z) , & z \in \partial D(0, \delta), \\ J_2^{(1)}(z), & z \in \partial D(1, \delta), \end{cases}
  \end{align}
 where the functions $J_k^{(0)}(z)$, $k=1,2$, and $J_2^{(1)}(z)$ are given in \eqref{jk-formula-0} and \eqref{jk-formula-1}.

\item[\rm (3)] As $z \to \infty$, we have $R_2(z)=\Boh(z^{-1}).$
\end{enumerate}
\end{rhp}
Similar to the calculation of $R_1$, we note that, due to \eqref{def:psi2rho}, the function
\begin{multline}\label{eq:psi2prod}
z^{-\frac{2}{3}} \diag \left(z^{\frac13},1,z^{-\frac13} \right) \Psi_2(\rho) \diag \left(z^{-\frac13},1,z^{\frac13} \right)
\\
=  \begin{pmatrix} 0 & 0 & \pi_6(\rho) \\ \left( \pi_6(\rho) + \frac{\rho}{3} \pi_3(\rho) + \frac{1 - \alpha}{3} \right)/z & 0 & 0 \\ 0 & \left(\pi_6(\rho) + \frac{2 \rho}{3} \pi_3(\rho)+ \frac{\rho^2}{9}+ \frac{1 - 2 \alpha}{3} \right)/z & 0
\end{pmatrix}
\end{multline}
is analytic in a neighbourhood of the origin except for a simple pole at $z = 0$ in the (2,1) and (3,2) entries. This implies  $J_2^{(0)}(z) + R_{1,-}(z) J_1^{(0)}(z)$ only has a simple pole at $z = 0$ in $D(0,\delta)$. As $z^{2/3} - 1$ is analytic at $z = 1$, we have $J_2^{(1)}(z)$ is analytic in $D(1,\delta)\setminus\{1\}$. Thus, we obtain \eqref{R2-explicit} from
\begin{equation}
R_2(z) = \frac{1}{2 \pi i} \oint_{\partial D(0, \delta)} \frac{J_2^{(0)}(\zeta) + R_{1,-}(\zeta) J_1^{(0)}(\zeta)}{\zeta - z } \ud \zeta +  \frac{1}{2 \pi i} \oint_{\partial D(1, \delta)} \frac{J_2^{(1)}(\zeta)}{\zeta - z } \ud \zeta
\end{equation}
and the residue theorem.

This finishes the proof of Proposition \ref{prop-R-exp}.
\end{proof}

To facilitate our future calculations, it is convenient to rewrite the large-$s$ expansion of $R(z)$ in \eqref{R-large-s-exp} as
\begin{equation} \label{eq:R-R-hat}
R(z) = C_N \widehat{R} (z) C_N^{-1} = C_N \left( I +\frac{\widehat{R}_1(z)}{s^{1/3}}+ \frac{\widehat{R}_2(z)}{s^{2/3}} + \Boh (s^{-1}) \right) C_N^{-1},
\end{equation}
where $C_N$ is given in \eqref{eq:CN} and
\begin{equation}\label{def:hatRi}
\widehat R_{i}(z)=C_N^{-1} R_i(z) C_N, \qquad i=1,2.
\end{equation}
We conclude this section by calculating $\widehat R_1(0)$ and $\widehat R_2(0)$ explicitly.

We start with the observation that if
$$ \mathcal{J}(z) = \mathcal{E}(z)\left(\mathcal{C}_0 + \frac{\mathcal{C}_1}{z}   \right) \mathcal{E}(z)^{-1},$$
where $\mathcal{C}_0$ and $\mathcal{C}_1$ are two constant matrices and $\mathcal{E}(z)$ is analytic near the origin, then
\begin{equation}\label{eq:residue}
\lim_{z\to 0} \left(\frac{1}{z} \Res_{\zeta=0} \mathcal{J}(\zeta) - \mathcal{J}(z) \right)=   \mathcal{E}(0)  \Big(  [\mathcal{C}_1,  \mathcal{E}(0)^{-1}  \mathcal{E}'(0)] - \mathcal{C}_0 \Big)  \mathcal{E}(0)^{-1}.
\end{equation}
Thus, by \eqref{R1-explicit}, \eqref{jk-formula-0}, \eqref{eq:zpsi1rho} and the above formula, it follows that
\begin{equation*}
R_1(0)=E_0(0)\left(\left[\begin{pmatrix}
0 & 0 & 0
\\
0 & 0 & 0
\\
\pi_3(\rho)+\frac{2\rho}{3} & 0 & 0
\end{pmatrix}, E_0(0)^{-1}E_0'(0)\right]-\begin{pmatrix}
0 & \pi_3(\rho) & 0
\\
0 & 0 & \pi_3(\rho)+\frac{\rho}{3}
\\
0 & 0 & 0
\end{pmatrix}\right)E_0(0)^{-1}.
\end{equation*}
Inserting \eqref{eq:E0zero} and \eqref{eq:E0zero-prime} into the above formula, we obtain after a straightforward calculation that
\begin{equation}\label{eq:R10expl}
\widehat{R}_1(0)= C_N^{-1} R_1(0)   C_N = \begin{pmatrix}
\frac{2\sqrt{3}}{3}\rho\beta i & -\pi_3(\rho) & 0
\\
\frac{\sqrt{3}}{2}\rho\beta\left(\sqrt{3}\beta-\frac{i}{3}\right) & -\frac{\sqrt{3}}{3}\rho \beta i & -\pi_3(\rho)-\frac{\rho}{3}
\\
-\frac{\sqrt{3}}{2}\rho \beta^2 \left(\sqrt{3}+\beta i\right) &  \frac{\sqrt{3}}{3}\rho\beta i & -\frac{\sqrt{3}}{3}\rho \beta i
\end{pmatrix} .
\end{equation}


To evaluate $\widehat R_2(0)$, it is readily seen from \eqref{R1-explicit}, \eqref{R2-explicit} and \eqref{def:hatRi} that
\begin{equation}\label{eq:R20pre}
\widehat R_2(0)=C_N^{-1}\left(-\Res_{\zeta = 1} J_2^{(1)}(\zeta)+\lim_{z\to 0} \left(\frac{1}{z} \Res_{\zeta=0} J_2^{(0)}(\zeta) - J_2^{(0)}(z) \right)+R_1(0)^2\right)C_N.
\end{equation}
For later use, it suffices to find the first row and the last column of $\widehat R_2(0)$. We next calculate three terms on the right hand side of the above formula separately. First, it follows from
\eqref{jk-formula-1}, \eqref{eq:E-expand-coeff-1} and \eqref{Phi1-def} that
\begin{align}
 &\Res_{\zeta = 1} C_N^{-1} J_2^{(1)}(\zeta) C_N = \frac{1}{\sqrt{3}} C_N^{-1} E_1(1)\diag\Big(\Phi^{(\CHF)}_1, 0 \Big) E_1(1)^{-1} C_N \label{eq:resJ2at1-old}
\\
&=\frac{1}{3\sqrt{3}}  \nonumber
\\
&\times
\begin{pmatrix}
2\Im \left(\frac{\Gamma(1-\beta) c_1(s)^2}{\Gamma(\beta)\omega }\right) & -\sqrt{3}\beta^2+2\Im \left(\frac{\Gamma(1-\beta)c_1(s)^2\omega}{\Gamma(\beta)}\right) &  \sqrt{3}\beta^2+2\Im \left(\frac{\Gamma(1-\beta) c_1(s)^2}{\Gamma(\beta)}\right)
\\
\sqrt{3}\beta^2+2\Im \left(\frac{\Gamma(1-\beta) c_1(s)^2\omega}{\Gamma(\beta)}\right) & 2\Im \left(\frac{\Gamma(1-\beta) c_1(s)^2}{\Gamma(\beta)}\right) &  -\sqrt{3}\beta^2+2\Im \left(\frac{\Gamma(1-\beta) c_1(s)^2}{\Gamma(\beta) \omega}\right)
\\
-\sqrt{3}\beta^2+2\Im \left(\frac{\Gamma(1-\beta) c_1(s)^2}{\Gamma(\beta)}\right) & \sqrt{3}\beta^2+2\Im \left(\frac{\Gamma(1-\beta) c_1(s)^2}{\Gamma(\beta)\omega}\right) &  2\Im \left(\frac{\Gamma(1-\beta) c_1(s)^2\omega}{\Gamma(\beta)}\right)
\end{pmatrix} , \nonumber
\end{align}
where we have also made use of the fact that $\beta$ is purely imaginary. From the definition of $c_1(s)$ in \eqref{eq:c-1}, it is easily seen that
\begin{equation}\label{eq:c1theta}
\Gamma(1-\beta) c_1(s) = |\Gamma(1-\beta)| e^{i \vartheta(s)},
\end{equation}
where
\begin{align}\label{def:vtheta}
\vartheta(s) & = \frac{\alpha { \pi}  }{3} +  \arg \Gamma(1-\beta) -\beta i \left(\frac{2}{3}\ln s+\frac{3}{2}\ln3+\ln f'(1)\right) -i \frac{\theta_1(s)-\theta_2(s)} {2}.
\end{align}
Moreover, from \cite[Formula 5.5.3]{DLMF}, we have
\begin{equation*}
\Gamma(1-\beta) \Gamma(\beta) = \frac{\pi}{\sin(\beta \pi)},
\end{equation*}
and thus,
\begin{equation} \label{eq:Gamma-beta-relation}
\Gamma(1+\beta) \Gamma(1-\beta) =|\Gamma(1-\beta)^2| = \beta \Gamma (\beta)\Gamma(1-\beta)=\frac{\beta \pi}{\sin(\beta \pi)}.
\end{equation}
This gives us
\begin{equation}
\Im \frac{\Gamma(1-\beta) c_1(s)^2}{\Gamma(\beta)} = \Im \frac{\Gamma(1-\beta)^2 c_1(s)^2}{\Gamma(1-\beta)\Gamma(\beta)}= \Im\left( -\beta i \, e^{2i \vartheta(s) + \frac{\pi i }{2}} \right) = -\beta i \cos(2 \vartheta(s)).
\end{equation}
Therefore, \eqref{eq:resJ2at1-old} can be rewritten as
\begin{align}
 &\Res_{\zeta = 1} C_N^{-1} J_2^{(1)}(\zeta) C_N \label{eq:resJ2at1}
\\
&{\small =\frac{1}{3\sqrt{3}}
\begin{pmatrix}
-2 \beta i \cos(2 \vartheta(s) -\frac{2\pi}{3})  & -\sqrt{3}\beta^2 -2 \beta i \cos(2 \vartheta(s) +\frac{2\pi}{3}) &  \sqrt{3}\beta^2-2 \beta i \cos(2 \vartheta(s))
\\
\sqrt{3}\beta^2 -2 \beta i \cos(2 \vartheta(s) +\frac{2\pi}{3})  & -2 \beta i \cos(2 \vartheta(s))  &  -\sqrt{3}\beta^2-2 \beta i \cos(2 \vartheta(s) -\frac{2\pi}{3})
\\
-\sqrt{3}\beta^2-2 \beta i \cos(2 \vartheta(s)) & \sqrt{3}\beta^2-2 \beta i \cos(2 \vartheta(s) -\frac{2\pi}{3}) & -2 \beta i \cos(2 \vartheta(s) +\frac{2\pi}{3})
\end{pmatrix} .}  \nonumber
\end{align}
Next, by \eqref{jk-formula-0}, \eqref{eq:psi2prod} and \eqref{eq:residue},
it is easily seen that
\begin{align*}
&\lim_{z\to 0} \left(\frac{1}{z} \Res_{\zeta=0} J_0^{(2)}(\zeta) - J_0^{(2)}(z) \right)
\nonumber
\\
&=E_0(0)\left(\left[\begin{pmatrix}
0 & 0 & 0
\\
\mathcal{A}+ \pi_6(\rho)  & 0 & 0
\\
0 & \mathcal{B}+ \pi_6(\rho)  & 0
\end{pmatrix}, E_0(0)^{-1}E_0'(0)\right]-\begin{pmatrix}
0 & 0 & \pi_6(\rho)
\\
0 & 0 & 0
\\
0 & 0 & 0
\end{pmatrix}\right)E_0(0)^{-1},
\end{align*}
where
\begin{equation}\label{def:AB}
\mathcal{A}= \frac{\rho}{3} \pi_3(\rho) + \frac{1 - \alpha}{3},\quad \mathcal{B}= \frac{2 \rho}{3} \pi_3(\rho)+ \frac{\rho^2}{9}+ \frac{1 - 2 \alpha}{3}.
\end{equation}
This, together with \eqref{eq:E0zero} and \eqref{eq:E0zero-prime}, implies that
\begin{align}
\lim_{z\to 0} \left( C_N^{-1} \Big( \frac{1}{z}  \Res_{\zeta=0} J_0^{(2)}(\zeta) - J_0^{(2)}(z) \Big)C_N \right)_{11}& =\frac{\beta}{2}\left( \sqrt{3}\mathcal{A}(-i+\sqrt{3}\beta) -6\mathcal{B}\beta\right), 
\\
\lim_{z\to 0} \left(C_N^{-1} \Big(\frac{1}{z}\Res_{\zeta=0} J_0^{(2)}(\zeta) - J_0^{(2)}(z) \Big)C_N \right)_{12}& =\sqrt{3}\beta\mathcal{B}i,
\\
\lim_{z\to 0} \left( C_N^{-1} \Big(\frac{1}{z}\Res_{\zeta=0} J_0^{(2)}(\zeta) - J_0^{(2)}(z) \Big)C_N \right)_{13}& =-\pi_6(\rho),
\\
\lim_{z\to 0} \left( C_N^{-1} \Big(\frac{1}{z} \Res_{\zeta=0} J_0^{(2)}(\zeta) - J_0^{(2)}(z) \Big)C_N \right)_{23}& =-\sqrt{3}\beta \mathcal{A}i, 
\\
\lim_{z\to 0} \left(C_N^{-1} \Big(\frac{1}{z}\Res_{\zeta=0} J_0^{(2)}(\zeta) - J_0^{(2)}(z) \Big)C_N \right)_{33}& =\frac{\beta}{2} \left(\sqrt{3}\mathcal{B}(i+\sqrt{3}\beta) -6\mathcal{A}\beta \right).
\end{align}
We also observe from \eqref{eq:R10expl} that
\begin{align}\label{eq:R10square}
&\widehat{R}_1(0)^2 = C_N^{-1} R_1(0)^2   C_N =
\nonumber \\
&=
\begin{pmatrix}
-\frac43\rho^2\beta^2-\frac{\sqrt{3}}{2}\rho\pi_3(\rho)\beta\left(\sqrt{3}\beta-\frac{i}{3}\right) & -\frac{\sqrt{3}}{3}\rho\pi_3(\rho)\beta i & \pi_3(\rho)\left(\pi_3(\rho)+\frac{\rho}{3}\right)
\\
\ast & \ast & \frac{2\sqrt{3}}{3}\rho\beta\left(\pi_3(\rho)+\frac{\rho}{3}\right)i
\\
\ast & \ast & -\frac{\rho^2 \beta^2}{3}-\frac{\sqrt{3}}{3}\rho\beta \left(\pi_3(\rho)+\frac{\rho}{3}\right) i
\end{pmatrix}.
\end{align}
Thus, inserting \eqref{eq:resJ2at1}--\eqref{eq:R10square} into \eqref{eq:R20pre}, we finally obtain that
\begin{align}
\Big( \widehat{R}_2(0) \Big)_{11}&=\frac{\beta}{2}\left( \sqrt{3}\mathcal{A}(-i+\sqrt{3}\beta)-6\mathcal{B}\beta \right) 
 -\frac43\rho^2\beta^2
\nonumber
\\
& \quad -\frac{\sqrt{3}}{2}\rho\pi_3(\rho)\beta\left(\sqrt{3}\beta-\frac{i}{3}\right)+\frac{ 2 \beta i  }{3\sqrt{3}} \cos(2 \vartheta(s) -\frac{2\pi}{3})\nonumber\\
&= -\beta^2\left(3\rho\pi_3(\rho)+\frac{5}{3}\rho^2-\frac{3\alpha-1}{2}\right)+\frac{  \beta i  }{2\sqrt{3}}(\alpha-1) +\frac{ 2 \beta i  }{3\sqrt{3}} \cos(2 \vartheta(s) -\frac{2\pi}{3}),
 \label{eq:R2011}
\\
\Big( \widehat{R}_2(0) \Big)_{12}&=\sqrt{3}\beta \mathcal{B}i-\frac{\sqrt{3}}{3}\rho\pi_3(\rho)\beta i
 + \frac{\beta^2}{3}  +\frac{2 \beta i}{3\sqrt{3}}  \cos(2 \vartheta(s) +\frac{2\pi}{3}) \nonumber\\
 &= \frac{\beta^2}{3}+\frac{ \beta i}{\sqrt{3}}\left(\rho\pi_3(\rho)+\frac{1}{3}\rho^2-2\alpha+1\right) +\frac{2 \beta i}{3\sqrt{3}}  \cos(2 \vartheta(s) +\frac{2\pi}{3}),
\\
\Big( \widehat{R}_2(0) \Big)_{13}&=-\pi_6(\rho)+\pi_3(\rho)\left(\pi_3(\rho)+\frac{\rho}{3}\right) - \frac{\beta^2}{3}  +\frac{2 \beta i}{3\sqrt{3}}  \cos(2 \vartheta(s)) , \label{eq:R2013}
\\
\Big( \widehat{R}_2(0) \Big)_{23}&=-\sqrt{3}i\beta\mathcal{A}+\frac{2i\beta}{\sqrt{3}}\rho\left(\pi_3(\rho)+\frac{\rho}{3}\right)
+ \frac{\beta^2}{3}  +\frac{2 \beta i}{3\sqrt{3}}  \cos(2 \vartheta(s)-\frac{2\pi}{3}) \nonumber\\
&= \frac{\beta^2}{3}  +\frac{ \beta i}{\sqrt{3}} \left(\rho\pi_3(\rho)+\frac{2}{3}\rho^2+\alpha-1\right) +\frac{2 \beta i}{3\sqrt{3}}  \cos(2 \vartheta(s)-\frac{2\pi}{3}),
\\
\Big( \widehat{R}_2(0) \Big)_{33}&=\frac{\beta}{2}\left(\sqrt{3}\mathcal{B}(i+\sqrt{3}\beta)-6\mathcal{A}\beta \right)
 -\frac{\rho^2 \beta^2}{3}
\nonumber
\\
& \quad -\frac{\sqrt{3}}{3}\rho\beta \left(\pi_3(\rho)+\frac{\rho}{3}\right) i+\frac{ 2 \beta i  }{3\sqrt{3}} \cos(2 \vartheta(s) +\frac{2\pi}{3})
\nonumber
\\
&= -\frac{\beta^2}{2}\Big(1+\frac{\rho^2}{3}\Big)-\frac{  \beta i  }{2\sqrt{3}}\Big(\frac{\rho^2}{3}+2\alpha-1\Big)+\frac{ 2 \beta i  }{3\sqrt{3}} \cos(2 \vartheta(s) +\frac{2\pi}{3}),
\label{eq:R2033}
\end{align}
where $\pi_3(\rho)$, $\pi_6(\rho)$, $\mathcal{A}$ and $\mathcal{B}$ are given in \eqref{pi3-def}, \eqref{pi6-def} and \eqref{def:AB}, respectively.
\section{Asymptotic analysis of the RH problem for $\Phi$ as $s \to 0^+$}\label{sec:AsyPhi0}

Throughout this section, it is assumed that $0<\gamma \leq 1$.

\subsection{Global parametrix}
As $s \to 0^+$, the interval $(0,s)$ vanishes, and a comparison of the RH problems $\Psi_{\alpha}$ and $\Phi$ invokes us to define the global parametrix $\widetilde N$ for $|z|>\delta>s$ by
\begin{equation}\label{def:tildeN}
\widetilde N(z):=\Psi_{\alpha}(z)\left\{
                   \begin{array}{ll}
                     J_{\Psi_{\alpha},1}, & \hbox{$\arg z<\frac{\pi}{4}$ and $\arg (z-s)>\frac{\pi}{4}$,} \\
                     J_{\Psi_{\alpha},5}^{-1}, & \hbox{$\arg z>-\frac{\pi}{4}$ and $\arg (z-s)<-\frac{\pi}{4}$,} \\
                     I, & \hbox{elsewhere,}
                   \end{array}
                 \right.
\end{equation}
where $J_{\Psi_{\alpha},k}$, $k=0,1,\ldots,5$, stands for the jump matrix $J_{\Psi_{\alpha}}$ in \eqref{jumps:M} restricted on the contour $\Sigma_k$.

\subsection{Local parametrix}
Near $z=0$, we intend to find a function $\widetilde P^{(0)}(z)$ satisfying the following RH problem.
\begin{rhp}\label{rhp:tildeP0}
\hfill
\begin{enumerate}
\item[\rm (1)]  $\widetilde P^{(0)}(z)$ is defined and analytic in $D(0,\delta)\setminus \Sigma_\Phi$, where $\Sigma_\Phi$ is defined in \eqref{def:Sigmaphi}.

\item[\rm (2)]  For $z\in \Sigma_\Phi \cap D(0,\delta)$, $\widetilde P^{(0)}(z)$ satisfies the jump condition
\begin{equation}\label{eq:tildeP0-jump}
\widetilde P^{(0)}_+(z)=\widetilde P^{(0)}_-(z)J_\Phi(z),
\end{equation}
where $J_\Phi(z)$ is given in \eqref{def:JX}.

\item[\rm (3)]
As $ s \to 0^+$, we have the matching condition
\begin{equation}\label{eq:MatchingCondtildeP0}
\widetilde P^{(0)}(z)=\left(I+\Boh(s^{\alpha+1}) \right)\widetilde N(z),\qquad z\in \partial D(0,\delta),
\end{equation}
where $\widetilde N(z)$ is given in \eqref{def:tildeN}.
\end{enumerate}
\end{rhp}

Recall the analytic function $\Psi_\alpha^{(0)}(z)$ in \eqref{eq:Psizero} which describes the local behavior of $\Psi_{\alpha}$ near the origin, we look for a solution of the above RH problem in the form
\begin{multline}\label{eq:tildeP0}
\widetilde P^{(0)}(z)= \Psi_\alpha^{(0)}(z)\begin{pmatrix}
1 & \eta(z/s) & 0
\\
0 & 1 & 0
\\
0 & 0 & 1
\end{pmatrix}
\left\{
                                     \begin{array}{ll}
                                      \begin{pmatrix}
1 & \frac{i}{2\pi}\ln z & 0
\\
0 & z^{-\alpha} & 0
\\
0 & \frac{(-1)^{\alpha+1}}{2\pi}i\ln z & 1
                                      \end{pmatrix}, & \hbox{if $\alpha \in \mathbb{N}\cup \{0\}$}
\\
 \begin{pmatrix}
1 & \frac{e^{-\alpha \pi i}}{2\sin(\alpha \pi)}i & 0
\\
0 & z^{-\alpha} & 0
\\
0 & -\frac{i}{2\sin(\alpha \pi)} & 1
                                       \end{pmatrix}, & \hbox{if $\alpha \notin \mathbb{Z}$}
\end{array}
                                   \right\}
\\
\times \left\{
          \begin{array}{ll}
            J_{\Phi,1}^{-1}, & \hbox{$z\in D(0,\delta) \cap \textrm{I}$,} \\
            I, & \hbox{$z\in D(0,\delta) \cap \textrm{II}$,} \\
           J_{\Phi,2}^{-1}, & \hbox{$z\in D(0,\delta) \cap \textrm{III}$,} \\
            J_{\Phi,1}^{-1}J_{\Phi,0}^{-1}J_{\Phi,5}^{-1}J_{\Phi,4}, & \hbox{$z\in D(0,\delta) \cap \textrm{IV}$,} \\
           J_{\Phi,1}^{-1}J_{\Phi,0}^{-1}J_{\Phi,5}^{-1}, & \hbox{$z\in D(0,\delta) \cap \textrm{V}$,} \\
            J_{\Phi,1}^{-1}J_{\Phi,0}^{-1}, & \hbox{$z\in D(0,\delta) \cap \textrm{VI}$,} \\
           \end{array}
        \right.
\end{multline}
where $\eta(z)$ is a function to be determined later, $J_{\Phi,k}$, $k=0,1,\ldots,5$, denote the jump matrix $J_{\Phi}$ in \eqref{def:JX} restricted on the contour $\Sigma_0^{(s)}, \Sigma_1^{(s)}, \Sigma_2,\Sigma_3,\Sigma_4,\Sigma_5^{(s)}$, respectively, and the regions I--VI are shown in Figure \ref{fig:X}. By choosing the principal branch cut for $z^{\alpha}$  and $\ln z$, it is straightforward to check that
if $z<0$,
\begin{align}
&\begin{pmatrix}
1 & \frac{i}{2\pi}\ln z & 0
\\
0 & z^{-\alpha} & 0
\\
0 & \frac{(-1)^{\alpha+1}}{2\pi}i\ln z & 1
\end{pmatrix}_+
J_{\Phi,2}^{-1}=\begin{pmatrix}
1 & \frac{i}{2\pi}\ln z & 0
\\
0 & z^{-\alpha} & 0
\\
0 & \frac{(-1)^{\alpha+1}}{2\pi}i\ln z & 1
\end{pmatrix}_-
\begin{pmatrix}
1 & -1 & 0
\\
0 & 1 & 0
\\
0 & (-1)^{\alpha} & 1
\end{pmatrix}
J_{\Phi,2}^{-1}
\nonumber
\\
&=\begin{pmatrix}
1 & \frac{i}{2\pi}\ln z & 0
\\
0 & z^{-\alpha} & 0
\\
0 & \frac{(-1)^{\alpha+1}}{2\pi}i\ln z & 1
\end{pmatrix}_-J_{\Phi,1}^{-1}J_{\Phi,0}^{-1}J_{\Phi,5}^{-1}J_{\Phi,4}J_{\Phi,3}
\end{align}
for $\alpha \in \mathbb{N}\cup \{0\}$, and
\begin{align}
& \begin{pmatrix}
1 & \frac{e^{-\alpha \pi i}}{2\sin(\alpha \pi)}i & 0
\\
0 & z^{-\alpha} & 0
\\
0 & -\frac{i}{2\sin(\alpha \pi)} & 1
\end{pmatrix}_+
J_{\Phi,2}^{-1}= \begin{pmatrix}
1 & \frac{e^{-\alpha \pi i}}{2\sin(\alpha \pi)}i & 0
\\
0 & z^{-\alpha} & 0
\\
0 & -\frac{i}{2\sin(\alpha \pi)} & 1
                                       \end{pmatrix}_-
\begin{pmatrix}
1 & -e^{-2\alpha \pi i} & 0
\\
0 & e^{-2\alpha \pi i} & 0
\\
0 & e^{-\alpha \pi i} & 1
\end{pmatrix}
J_{\Phi,2}^{-1}
\nonumber
\\
&= \begin{pmatrix}
1 & \frac{e^{-\alpha \pi i}}{2\sin(\alpha \pi)}i & 0
\\
0 & z^{-\alpha} & 0
\\
0 & -\frac{i}{2\sin(\alpha \pi)} & 1
\end{pmatrix}_-J_{\Phi,1}^{-1} J_{\Phi,0}^{-1}J_{\Phi,5}^{-1}J_{\Phi,4}J_{\Phi,3}
\end{align}
for $\alpha \notin \mathbb{Z}$. This, together with the fact that
\begin{equation}
 J_{\Phi,1}^{-1}J_{\Phi,0}^{-1}J_{\Phi,5}^{-1}=
\begin{pmatrix}
1 & -1 & 0
\\
0 & 1 & 0
\\
0 & 0 & 1
\end{pmatrix},
\end{equation}
implies that $\widetilde P^{(0)}(z)$ in \eqref{eq:tildeP0} satisfies the jump condition \eqref{eq:tildeP0-jump} if the function $\eta(z)$ therein solves the following scalar RH problem.

\begin{rhp}\label{rhp:eta}
\hfill
\begin{enumerate}
\item[\rm (1)]  $\eta(z)$ is defined and analytic in $\C \setminus [0,1]$.

\item[\rm (2)] For $x \in (0,1)$, $\eta$ satisfies the jump condition
\begin{equation}
\eta_+(x)=\eta_-(x)-\gamma (sx)^\alpha,
\end{equation}
where the orientation of $(0,1)$ is taken from the left to the right.

\item[\rm (3)]
As $ z \to \infty$, we have $\eta(z)=\Boh(z^{-1})$.
\end{enumerate}
\end{rhp}
By the Sokhotske-Plemelj formula, it is easily seen that
\begin{equation}\label{def:eta}
\eta(z)=-\frac{\gamma s^{\alpha}}{2 \pi i}\int_0^1\frac{x^\alpha}{x-z} \ud x.
\end{equation}
Since $\eta(z/s)=\Boh(s^{\alpha+1})$ as $s \to 0^+$ for $z\in \partial D(0,\delta)$, we then obtain the matching condition \eqref{eq:MatchingCondtildeP0} from \eqref{def:tildeN}, \eqref{eq:tildeP0} and \eqref{eq:Psizero}.

\subsection{Final transformation}
The final transformation is defined by
\begin{equation}\label{def:tildeR}
\widetilde R(z)=\left\{
                  \begin{array}{ll}
                    \Phi(z)\widetilde N(z)^{-1}, & z \in \mathbb{C} \setminus D(0,\delta), \\
                    \Phi(z)\widetilde P^{(0)}(z)^{-1}, & z \in D(0, \delta).
                  \end{array}
                \right.
\end{equation}
By the RH problems for $\Phi$, $\Psi_{\alpha}$ and $\widetilde P^{(0)}$, it follows that $\widetilde R$ is analytic for $z\in D(0,\delta)\setminus \{0,s\}$. To this end, we obtain from \eqref{def:eta} and \cite[Section 30]{Musk} that as $z \to 0$,
\begin{equation}
\eta(z/s)=\left\{
            \begin{array}{ll}
              \frac{\gamma}{2\sin(\alpha\pi) i}e^{\mp \alpha \pi i}z^{\alpha}+\Boh(1), & \hbox{if $-1<\alpha<0$ and $\pm \Im z>0$,} \\
              \Boh(\ln z), & \hbox{if $\alpha=0$,} \\
              \Boh(1), & \hbox{if $\alpha>0$,}
            \end{array}
          \right.
\end{equation}
and as $z\to s$,
\begin{equation}
\eta(z/s)=\Boh(\ln(z-s)).
\end{equation}
A combination of the above two estimates, \eqref{eq:tildeP0}, and items (4) and (5) of Proposition \ref{rhp:X} implies that both $0$ and $s$ are removable singularities of $\widetilde R$ defined in \eqref{def:tildeR}. Thus, we arrive at the following RH problem for $\widetilde R$.
\begin{rhp}\label{rhp:tildeR}
\hfill
\begin{enumerate}
\item[\rm (1)]  $\widetilde R (z)$ is defined and analytic in $\C \setminus D(0,\delta)$.

\item[\rm (2)]  For $z \in \partial D(0,\delta)$, $\widetilde R (z)$ satisfies the jump condition
\begin{equation}
\widetilde R_+ (z)=\widetilde R_- (z)J_{\widetilde R}(z),
\end{equation}
where
$$J_{\widetilde R}(z) = \widetilde P^{(0)}(z) \widetilde N(z)^{-1} $$
and the orientation of $\partial D(0,\delta)$ is taken in a clockwise manner.

\item[\rm (3)]
As $ z \to \infty$, we have $\widetilde R(z)= I+\Boh(z^{-1})$.
\end{enumerate}
\end{rhp}
In view of \eqref{eq:MatchingCondtildeP0}, it is readily seen that $J_{\widetilde R}(z)=I+\Boh(s^{\alpha+1})$, as $s\to 0^+$, uniformly for the parameters $\gamma$ and $\rho$ in any compact subsets of $(0,1]$ and $\R$. As a consequence, we have
\begin{equation}\label{eq:esttildeR}
\widetilde R(z)=I+\Boh(s^{\alpha+1}),\qquad  \frac{\ud}{\ud z}\widetilde R(z)=\Boh(s^{\alpha+1}),  \qquad s\to 0^+,
\end{equation}
uniformly for $z\in \C \setminus \partial D(0,\delta)$.

\section{Asymptotics of the functions $p_{i,k}$ and $q_{i,k}$} \label{sec:Asypq}

By Proposition \ref{pro:Lax pair}, it is readily seen that there exists at least one solution to the system of equations \eqref{def:sysdiff} and \eqref{eq:constraint}. As an outcome of asymptotic analysis of the RH problem for $\Phi$, we are able to obtain asymptotics of this special solution, which is also needed in the proof of Theorem \ref{thm:FAsy}. To that end, it is convenient to introduce the functions $\hat p_{i,k}$ and $\hat q_{i,k}$, $i=0,1$, $k=1,2,3$, by
\begin{equation}\label{eq:hatq}
\begin{pmatrix}
      q_{i,1}(s)
      \\
      q_{i,2}(s)
      \\
      q_{i,3}(s)
\end{pmatrix}
    =C_\Psi \diag\left(s^{\frac13},1, s^{-\frac13}\right)C_N \diag\left(s^{-\frac13},1, s^{\frac13}\right) \begin{pmatrix}
      \hat{q}_{i,1}(s) \\
      \hat{q}_{i,2}(s) \\
      \hat{q}_{i,3}(s)
    \end{pmatrix}
 \end{equation}
and
\begin{equation}\label{eq:hatp}
\begin{pmatrix}
      p_{i,1}(s) \\
      p_{i,2}(s) \\
      p_{i,3}(s)
    \end{pmatrix}= C_\Psi^{-t}  \diag\left(s^{-\frac13},1, s^{\frac13}\right)C_N^{-t} \diag\left(s^{\frac13},1, s^{-\frac13}\right) \begin{pmatrix}
      \hat{p}_{i,1}(s) \\
      \hat{p}_{i,2}(s) \\
      \hat{p}_{i,3}(s)
    \end{pmatrix},
 \end{equation}
where the constant matrices $C_\Psi$ and $C_N$ are given in \eqref{eq:cons-C-Psi} and \eqref{eq:CN}, respectively.
\begin{proposition}\label{prop:asypq}
There exists one solution $\{p_{i,k}(s),q_{i,k}\}_{i=0,1,k=1,2,3}$ to the system of equations \eqref{def:sysdiff} and \eqref{eq:constraint} such that the following asymptotic behaviors hold. As $s\to +\infty$, we have, with $\{\hat p_{i,k}(s),\hat q_{i,k}\}_{i=0,1,k=1,2,3}$ defined in \eqref{eq:hatq} and \eqref{eq:hatp} and $0<\gamma <1$,
\begin{align}
\hat q_{0,1}(s) &= \frac{ (1-\gamma)^{1/3} e^{\rho^2/6} c(\alpha)}{ \sqrt{2 \pi}} \nonumber
\\
& \quad \times \Big(   -\pi_6(\rho)+\pi_3(\rho)\left(\pi_3(\rho)+\frac{\rho}{3}\right) - \frac{\beta^2}{3}  +\frac{2 \beta i}{3\sqrt{3}}  \cos(2 \psi(s)) + \Boh(s^{-\frac{1}{3}}) \Big), \label{eq:q01-asy-alpha}\\
\hat q_{0,2}(s) &= \frac{ (1-\gamma)^{1/3} e^{\rho^2/6} c(\alpha)}{ \sqrt{2 \pi}} \Big( - \pi_3(\rho) - \frac{\rho}{3} \nonumber\\
& \quad  + \Big( \frac{\beta^2}{3}  +\frac{ \beta i}{\sqrt{3}} \Big(\rho\pi_3(\rho)+\frac{2}{3}\rho^2+\alpha-1\Big) +\frac{2 \beta i}{3\sqrt{3}}  \cos(2 \psi (s)-\frac{2\pi}{3})\Big)s^{-\frac{1}{3}}
+ \Boh(s^{-\frac{2}{3}})\Big), \label{eq:q02-asy-alpha}\\
\hat q_{0,3}(s) &= \frac{ (1-\gamma)^{1/3} e^{\rho^2/6} c(\alpha)}{ \sqrt{2 \pi}}  \Big( 1  - \frac{\sqrt{3}}{3}\rho \beta i  s^{-\frac{1}{3}} \nonumber\\ 
& \quad -\Big( \frac{\beta^2}{2}\Big(1+\frac{\rho^2}{3}\Big)+\frac{  \beta i  }{2\sqrt{3}}\Big(\frac{\rho^2}{3}+2\alpha-1\Big)-\frac{ 2 \beta i  }{3\sqrt{3}} \cos(2 \psi(s) +\frac{2\pi}{3}) \Big)
s^{-\frac{2}{3}}
\nonumber\\
& \quad
+ \Boh(s^{-1}) \Big), \label{eq:q03-asy-alpha}
\end{align}
\begin{align}
\hat p_{0,1}(s) &=\frac{ \sqrt{2\pi} e^{-\rho^2/6} }{(1-\gamma)^{1/3} c(\alpha) }
 \Big(  1 +   \frac{i\beta}{\sqrt{3}} \rho s^{-\frac{1}{3}}+   \Big(\frac{1}{2} \beta^2\Big(3 \rho \pi_3(\rho)  + \frac{5}{3} \rho^2 + 1\Big) \nonumber\\
&\quad +\frac{i\beta}{2\sqrt{3}}\Big( \rho \pi_3(\rho)+\rho^2+2\alpha+1- \frac{ 4}{3 } \cos(2 \psi(s) -\frac{2}{3}\pi) \Big) \Big) s^{-\frac{2}{3}}+\Boh(s^{-1}) \Big), \label{eq:p01-asy-alpha} \\
\hat p_{0,2}(s) &=\frac{ \sqrt{2\pi} e^{-\rho^2/6} }{(1-\gamma)^{1/3} c(\alpha) }  \Big(  \pi_3(\rho) + \rho
 \nonumber\\
& \quad +\Big(-\frac{1}{3} \beta^2+\frac{i\beta}{3\sqrt{3}}\Big(3 \rho \pi_3(\rho)+2\rho^2-3\alpha-3- 2 \cos(2 \psi(s) +\frac{2}{3}\pi) \Big) \Big) s^{-\frac{1}{3}}
\nonumber\\
& \quad
+\Boh(s^{-\frac{2}{3}})
 \Big), \label{eq:p02-asy-alpha} \\
\hat p_{0,3}(s) &= \frac{  \sqrt{2\pi} e^{-\rho^2/6} }{(1-\gamma)^{1/3} c(\alpha) }   \Big( - \alpha + \pi_6(\rho)  + \rho \pi_3(\rho)  + \frac{\rho^2}{3}  + \frac{\beta^2}{3} - \frac{ 2 \beta i}{3 \sqrt{3} } \cos(2 \psi(s) )
\nonumber\\
& \quad+ \Boh(s^{-\frac{1}{3}}) \Big), \label{eq:p03-asy-alpha}
\end{align}
and
\begin{align}
\hat q_{1,1}(s) &=   \frac{2e^ {-(\frac{1}{2}\theta_3(s)+\frac{2}{3} \beta\pi i)} i } {\sqrt{3}} s^{-\frac{\alpha-1}{3}}  |\Gamma(1-\beta)|  \left(\cos(\psi(s)+\frac{2\pi }{3}) +\Boh(s^{-\frac13}) \right), \label{eq:q11-asy}
\\
\hat q_{1,2}(s) &=   \frac{2e^ {-(\frac{1}{2}\theta_3(s)+\frac{2}{3} \beta\pi i)} i } {\sqrt{3}} s^{-\frac{\alpha}{3}} |\Gamma(1-\beta)|  \left(\cos(\psi(s))+\Boh(s^{-\frac13}) \right), \label{eq:q12-asy} \\
\hat q_{1,3}(s) &=   \frac{2e^ {-(\frac{1}{2}\theta_3(s)+\frac{2}{3} \beta\pi i)} i } {\sqrt{3}} s^{-\frac{\alpha+1}{3}}  |\Gamma(1-\beta)|  \left( \cos(\psi(s)-\frac{2\pi }{3}) +\Boh(s^{-\frac13}) \right), \label{eq:q13-asy}
\\
\hat p_{1,1}(s) &=   \frac{\gamma e^ {\frac{1}{2}\theta_3(s)-\frac{1}{3} \beta\pi i} } {\sqrt{3} \pi i } s^{\frac{\alpha-1}{3}}   |\Gamma(1-\beta)|  \left( \sin(\psi(s)-\frac{\pi }{3})  +\Boh(s^{-\frac13}) \right), \label{eq:p11-asy}
\\
\hat  p_{1,2}(s) &=  - \frac{\gamma e^ {\frac{1}{2}\theta_3(s)-\frac{1}{3} \beta\pi i} } {\sqrt{3}\pi i } s^{\frac{\alpha}{3}}  |\Gamma(1-\beta)|   \left( \sin(\psi(s)) +\Boh(s^{-\frac13}) \right), \label{eq:p12-asy}
 \\
\hat p_{1,3}(s) & =    \frac{\gamma e^ {\frac{1}{2}\theta_3(s)-\frac{1}{3} \beta\pi i} } {\sqrt{3}\pi i } s^{\frac{\alpha+1}{3}}  |\Gamma(1-\beta)|  \left(\sin(\psi(s)+\frac{\pi }{3})  +\Boh(s^{-\frac13}) \right), \label{eq:p13-asy}
\end{align}
where
\begin{equation} \label{def:c-alpha}
c(\alpha):= \begin{cases}
	 1, & \alpha = 0, \\ 
	-\Gamma(\alpha), & \alpha \neq 0,	
	\end{cases}
\end{equation}
$\beta$, $\psi(s)$, $\pi_3(\rho)$ and $\theta_3(s)$ are given in \eqref{beta-def}, \eqref{def:psi}, \eqref{pi3-def} and \eqref{eq: theta-k-def}, respectively. As $s\to 0^+$, we have
\begin{equation} \label{eq: pk-qk-s=0}
\quad q_{0,k}(s)=\Boh(1), \quad p_{0,k}(s)=\Boh(1), \quad q_{1,k}(s)=\Boh(1), \quad p_{1,k}(s)=\Boh(s^\alpha), \quad k=1,2,3
\end{equation}
for $0<\gamma \leq 1$.
\end{proposition}
The rest of this section is devoted to the proof of Proposition \ref{prop:asypq}.


\subsection*{Proofs of \eqref{eq:q01-asy-alpha}--\eqref{eq:p03-asy-alpha}}
In view of \eqref{def: q-0} and \eqref{def:p0i}, it is readily seen that the functions $p_{0,k}$ and $q_{0,k}$, $k=1,2,3$, are related to the function $\Phi_0^{(0)}(s) $. From \eqref{eq:X-near-0-2}, \eqref{eq:X-near-0-3} and \eqref{eq: Phi-expand-0}, it follows that
\begin{equation}\label{eq:phi00}
\Phi_0^{(0)}(s)  = \lim_{z \to 0, \atop z \in \mathtt{II}} \Phi(z) \begin{cases}
\begin{pmatrix}
	1 & -\frac{\gamma - 1}{2\pi i} z^\alpha \ln z  & 0 \\
	0 & z^{\alpha} & 0 \\
	0&  -\frac{e^{\alpha \pi i}}{2\pi i} z^\alpha \ln z   & 1
	\end{pmatrix}, & \textrm{if } \alpha \in \mathbb{N} \cup \{0\}, \\
	\begin{pmatrix}
	                        1 & -c_+ z^\alpha & 0 \\
	                        0 & z^{\alpha} & 0 \\
	                        0& -\frac{z^\alpha }{2 i \sin(\alpha \pi)}
	                          & 1
	                       \end{pmatrix}, & \textrm{if } \alpha \notin \mathbb{Z}.
\end{cases}
\end{equation}
By tracing back the transformations $\Phi \mapsto T \mapsto S \mapsto R $ in \eqref{def:PhiToT}, \eqref{def:TtoS}, \eqref{def:StoR}, we observe that
\begin{equation}\label{eq:phisz2}
\Phi(sz)=\frac{i}{\sqrt{3}}s^{-\frac{\alpha}{3}} C_\Psi \diag\left(s^{\frac13},1, s^{-\frac13}\right)R(z)P^{(0)}(z) e^{\Theta(sz)}, \qquad z\in D(0, \delta) \cap \mathtt{II}.
\end{equation}
This, together with \eqref{eq:phi00}, \eqref{def:cpm} and $P^{(0)}(z)$ in \eqref{eq:P0Solution}, implies that
\begin{align}
\Phi_0^{(0)}(s) & =  (1-\gamma)^{\frac13} \frac{i}{\sqrt{3}}  C_\Psi  \diag\left(s^{\frac13},1, s^{-\frac13}\right)R(0) E_0(0)  \diag(s^{-\frac{1}{3}}, 1,  s^{\frac{1}{3}}) \nonumber \\
& \quad \times \lim_{z \to 0, \atop \frac{\pi}{4}<\arg z<\frac{3\pi}{4}  } \Psi_{\alpha}( z)
\begin{cases}
\begin{pmatrix}
	 \frac{1}{1-\gamma} & \frac{  z^\alpha}{2\pi i}  \ln z  & 0 \\
	0 &   z^{\alpha} & 0 \\
	0&  -\frac{e^{\alpha \pi i}}{2\pi i}   z^\alpha \ln z  & 1
	\end{pmatrix}, & \textrm{if } \alpha \in \mathbb{N} \cup \{0\}, \medskip \\
\begin{pmatrix}
	 \frac{1}{1-\gamma}  & \frac{e^{-\alpha \pi i}   z^{\alpha}}{2 i\sin(\alpha \pi)} &0 \\
	0 &    z^{\alpha} & 0 \\
	0& -\frac{   z^{\alpha} }{2 i\sin(\alpha \pi)}  & 1
	\end{pmatrix} , & \textrm{if } \alpha \notin \mathbb{Z}.
\end{cases} \label{eq:Phi-0-0}
\end{align}
To study the above limit, we recall from \eqref{Psi-in-II} that, if $\pi/4 < \arg z < 3 \pi /4$, $\Psi_{\alpha}(z)$ admits an explicit representation in terms of $p_k(z)$, $k=1,2,3$, defined in \eqref{def:pkint}, and the functions $p_1(z)$ and $p_2(z)$ are entire functions.
We then introduce the following auxiliary functions:
\begin{equation}
\mathcal{W}(z):=\mathcal{W}(z;\rho)=\mathcal{W}[p_1, p_2] (z) = p_1(z) p_2'(z) -p_1'(z) p_2(z), \label{def:wronskian-def}
\end{equation}
and
\begin{equation}\label{def:q}
q(z):=p_3(z)-\frac{1}{2 i\sin (\alpha\pi)}p_1(z)+\frac{e^{-\pi i\alpha}}{2 i\sin (\alpha\pi)}p_2(z), \qquad \alpha \notin \mathbb{Z}, \qquad z\in \mathbb{C}\setminus i\mathbb{R}_-.
\end{equation}
The proposition below will play an important role for us to determine the limit in \eqref{eq:Phi-0-0}. Since the proof is lengthy, it is postponed to Appendix \ref{sec:appen-prop}.
\begin{proposition}\label{prop:pk}
Let $p_k(z)$, $k=1,2,3$, $\mathcal{W}(z)$ 
and $q(z)$ be functions defined in \eqref{def:pkint}, \eqref{def:wronskian-def} 
and \eqref{def:q}, respectively. We have
\begin{equation}\label{eq:wronskian0}
\mathcal{W}(z) = - \frac{(2\pi )^{3/2} }{e^{\rho^2/2}} \sum_{k=0}^\infty w_k z^k, \qquad \alpha > -1,
\end{equation}
for some constants $w_k$, $k=0,1,2,\ldots$,
with \begin{equation}\label{eq:w0w1}
w_0 = \frac{1}{\Gamma(\alpha)}, \qquad w_1 = \frac{\rho}{\Gamma(\alpha+1)};
\end{equation}
if $\alpha=0$,
\begin{equation} \label{eq:p-alpha-exp}
 p_1(z)-p_2(z) = \sum_{k=0}^{\infty}b_k z^{2+k},
\end{equation}
for some constants $b_k$, $k=0,1,2,\ldots$,
with \begin{equation}\label{eq:b0}
b_0 = \pi i;
\end{equation}
and
\begin{equation}\label{eq:qExpand}
q(z)=\sum_{k=0}^{\infty}c_k z^{2+k-\alpha}, \qquad \alpha \notin \mathbb{Z},
\end{equation}
for some constants $c_k$, $k=0,1,2,\ldots$,
with \begin{equation}\label{eq:c0}
c_0 =-\frac{\pi }{\sin(\alpha\pi)\Gamma(3-\alpha)}.
\end{equation}
Moreover, we have
\begin{equation}\label{eq:ddp3}
\lim_{z \to 0} z^{\alpha}p_3''(z)=-\Gamma(\alpha), \qquad \alpha\in \mathbb{N}.
\end{equation}

\end{proposition}

With the aid of the above proposition, we are now ready to derive the asymptotics of $\hat q_{0,k}$ and $\hat p_{0,k}$, $k=1,2,3$. To show  \eqref{eq:q01-asy-alpha}--\eqref{eq:q03-asy-alpha}, we first have from \eqref{Psi-in-II}, \eqref{eq:p-alpha-exp} and \eqref{eq:b0} that, if $\alpha = 0$,
\begin{equation}
\lim_{z \to 0, \atop \frac{\pi}{4}<\arg z<\frac{3\pi}{4}  } \frac{\sqrt{2 \pi}}{e^{\rho^2/6}}  \Psi_{0}( z)
\begin{pmatrix}
	 \frac{1}{1-\gamma} & \frac{ 1}{2\pi i}  \ln z  & 0 \\
	0 &   1 & 0 \\
	0&  -\frac{1}{2\pi i}    \ln z  & 1
	\end{pmatrix} \begin{pmatrix}
      \gamma-1\\
    0 \\
    1
    \end{pmatrix} = \begin{pmatrix}
	p_1(0)-p_2(0) \\
	p_1'(0) -p_2'(0) \\
	p_1''(0)-p_2''(0)
	\end{pmatrix}
	 =\begin{pmatrix}
	0 \\
	0  \\
	2 \pi i
	\end{pmatrix}.
\end{equation}
If $\alpha \in \mathbb{N}$, we note from the fact $p_1(z)$ and $p_2(z)$ are entire functions that
$$\lim_{z\to 0} (\ln z) z^{\alpha}  p_k(z)=0,\qquad k=1,2.$$
Thus, it is readily seen from 
\eqref{Psi-in-II} and \eqref{eq:ddp3}  that
\begin{multline}
 \lim_{z \to 0, \atop \frac{\pi}{4}<\arg z<\frac{3\pi}{4}  } \frac{\sqrt{2 \pi}}{e^{\rho^2/6}} \Psi_{\alpha}( z)
\begin{pmatrix}
	 \frac{1}{1-\gamma} & \frac{  z^\alpha}{2\pi i}  \ln z  & 0 \\
	0 &   z^{\alpha} & 0 \\
	0&  -\frac{e^{\alpha \pi i}}{2\pi i}   z^\alpha \ln z  & 1
	\end{pmatrix} \begin{pmatrix}
      0\\
    1 \\
    0
    \end{pmatrix}
    \\
    = \lim_{z \to 0, \atop \frac{\pi}{4}<\arg z<\frac{3\pi}{4}  }
    \begin{pmatrix}
    z^\alpha p_3(z) \smallskip \\
    z^\alpha p_3'(z)\smallskip \\
    z^\alpha p_3''(z)
    \end{pmatrix}  = \begin{pmatrix}
	0 \\
	0  \\
	- \Gamma(\alpha)
	\end{pmatrix} .
\end{multline}
If $\alpha \notin \mathbb{Z}$, we obtain from  \eqref{eq:qExpand} and \eqref{eq:c0}  that 
\begin{multline}
\lim_{z \to 0, \atop \frac{\pi}{4}<\arg z<\frac{3\pi}{4}  } \frac{\sqrt{2 \pi}}{e^{\rho^2/6}} \Psi_{\alpha}( z)
\begin{pmatrix}
	 \frac{1}{1-\gamma}  & \frac{e^{-\alpha \pi i}   z^{\alpha}}{2 i\sin(\alpha \pi)} &0 \\
	0 &    z^{\alpha} & 0 \\
	0& -\frac{   z^{\alpha} }{2 i\sin(\alpha \pi)}  & 1
	\end{pmatrix} \begin{pmatrix}
      0\\
    1 \\
    0
    \end{pmatrix}  \\
     = \lim_{z \to 0, \atop \frac{\pi}{4}<\arg z<\frac{3\pi}{4}  } \begin{pmatrix}
z^{\alpha} q(z)
\\
z^{\alpha} q'(z)
\\
z^{\alpha} q''(z)
\end{pmatrix}  = \begin{pmatrix}
	0 \\
	0  \\
	- \Gamma(\alpha)
	\end{pmatrix} .
\end{multline}
Therefore, a combination of the above three formulas, \eqref{def: q-0} and \eqref{eq:Phi-0-0} yields
\begin{equation}
\begin{pmatrix}
    q_{0,1}(s) \\
    q_{0,2}(s) \\
    q_{0,3}(s)
    \end{pmatrix}  =  (1-\gamma)^{\frac{1}{3}} \frac{i e^{\rho^2/6}}{\sqrt{6 \pi}}  C_\Psi  \diag\left(s^{\frac13},1, s^{-\frac13}\right)R(0) E_0(0)  \diag(s^{-\frac{1}{3}}, 1,  s^{\frac{1}{3}})  \begin{pmatrix}
	0 \\
	0  \\
	c(\alpha)
	\end{pmatrix},  \label{eq:q0-limits}
\end{equation}
where the function $c(\alpha)$ is defined in \eqref{def:c-alpha}.

We next substitute \eqref{eq:hatq} into \eqref{eq:q0-limits} and obtain
from \eqref{eq:R-R-hat} that
\begin{equation}\label{eq:hatq0exac}
\begin{pmatrix}
      \hat{q}_{0,1}(s) \\
      \hat{q}_{0,2}(s) \\
    \hat{q}_{0,3}(s)
    \end{pmatrix} =  (1-\gamma)^{\frac{1}{3}} \frac{i e^{\rho^2/6}}{\sqrt{6 \pi}}   \diag\left(s^{\frac13},1, s^{-\frac13}\right)\widehat{R}(0) C_N^{-1} E_0(0)  \diag(s^{-\frac{1}{3}}, 1,  s^{\frac{1}{3}})  \begin{pmatrix}
	0 \\
	0  \\
	c(\alpha)
	\end{pmatrix} .
\end{equation}
%
%
From $E_0(0)$ in \eqref{eq:E0zero}, it is easily seen that
\begin{align}
& \diag\left(s^{\frac13},1, s^{-\frac13}\right) \widehat{R}(0) C_N^{-1} E_0(0) \diag(s^{-\frac{1}{3}}, 1,  s^{\frac{1}{3}}) \begin{pmatrix}
	0 \\
	0  \\
	c(\alpha)
	\end{pmatrix}   \nonumber \\
& = -\sqrt{3}\,i  \diag\left(s^{\frac13},1, s^{-\frac13}\right) \Big( C_N^t + \widehat{R}_1(0) C_N^t s^{-\frac{1}{3}}  + \widehat{R}_2(0) C_N^t s^{-\frac{2}{3}} + \Boh(s^{-1}) \Big)
\nonumber
\\
& \quad \times
\diag\left(s^{-\frac13},1, s^{\frac13}\right) \begin{pmatrix}
	0 \\
	0  \\
	c(\alpha)
	\end{pmatrix}  \nonumber \\
& = -\sqrt{3}  \, i \, c(\alpha)  \begin{pmatrix}
\Big( \widehat{R}_1(0) \Big)_{13} s^{\frac{1}{3}}+ \Big( \widehat{R}_2(0) \Big)_{13} + \Boh(s^{-\frac{1}{3}})   \\
 \Big(  \widehat{R}_1(0) \Big)_{23} + \Big(  \widehat{R}_2(0) \Big)_{23} s^{-\frac{1}{3}}  + \Boh(s^{-\frac{2}{3}}) \\
 1 + \Big( \widehat{R}_1(0) \Big)_{33} s^{-\frac{1}{3}} + \Big( \widehat{R}_2(0) \Big)_{33} s^{-\frac{2}{3}}   + \Boh(s^{-1})
\end{pmatrix}.
\end{align}
Inserting the above formula into \eqref{eq:hatq0exac}, we finally arrive at \eqref{eq:q01-asy-alpha}--\eqref{eq:q03-asy-alpha} from $\widehat{R}_1(0)$ in \eqref{eq:R10expl}, the first row and the last column of  $\widehat{R}_2(0)$  in \eqref{eq:R2011}--\eqref{eq:R2033} and the fact that
\begin{align}\label{eq:estvtheta}
\vartheta(s) = \psi(s)+\Boh(s^{-\frac13}),
\end{align}
where  $\vartheta(s)$ and $\psi(s)$ are defined in \eqref{def:vtheta} and \eqref{def:psi}, respectively.


The asymptotics of $\hat p_{0,k}(s)$, $k=1,2,3$, in \eqref{eq:p01-asy-alpha}--\eqref{eq:p03-asy-alpha} can be proved in a similar manner. From the definitions of $p_{0,k}(s)$ in \eqref{def:p0i}, we need to study $\Phi_0^{(0)}(s)^{-t}$, which, by \eqref{eq:Phi-0-0}, is
\begin{align}
\Phi_0^{(0)}(s)^{-t} & = -\frac{\sqrt{3} \, i}{(1-\gamma)^{1/3} }  C_\Psi^{-t} \diag\left(s^{-\frac13},1, s^{\frac13}\right)R(0)^{-t} E_0(0)^{-t}  \diag(s^{\frac{1}{3}}, 1,  s^{-\frac{1}{3}}) \nonumber \\
& \quad \times \lim_{z \to 0, \atop \frac{\pi}{4}<\arg z<\frac{3\pi}{4} } \Psi_{\alpha}( z)^{-t}
\begin{cases}
\begin{pmatrix}
	1-\gamma & 0  & 0 \\
	-\frac{  1-\gamma}{2\pi i}  \ln z &   z^{-\alpha} & \frac{e^{\alpha \pi i}}{2\pi i}   \ln z  \\
	0&  0 & 1
	\end{pmatrix}, & \textrm{if } \alpha \in \mathbb{N} \cup \{0\}, \medskip \\
\begin{pmatrix}
	 1-\gamma  &0 &0 \\
	 -\frac{(1-\gamma) e^{-\alpha \pi i}   }{2 i\sin(\alpha \pi)} &    z^{-\alpha} & \frac{   1 }{2 i\sin(\alpha \pi)}  \\
	0& 0 & 1
	\end{pmatrix} , & \textrm{if } \alpha \notin \mathbb{Z}. \label{eq:Phi-0-0-t}
\end{cases}
\end{align}
Recall that $\det \Psi_\alpha(z) = z^{-\alpha}$, we then obtain from \eqref{Psi-in-II} and the definition of $\mathcal{W}(z)$ in \eqref{def:wronskian-def} that
\begin{align}
\Psi_{\alpha}( z)^{-t} \begin{pmatrix}
0 \\ z^{-\alpha} \\0
\end{pmatrix} & = \frac{e^{\rho^2/3}  }{2\pi } \begin{pmatrix}
 p_1'(z) p_2''(z) - p_1''(z) p_2'(z)  \\
 - \Big(p_1(z) p_2''(z) - p_1''(z) p_2(z) \Big)   \\
 p_1(z) p_2'(z) -p_1'(z) p_2(z)
\end{pmatrix} \nonumber \\
& = \frac{e^{\rho^2/3}  }{2\pi }  \begin{pmatrix}
 p_1'(z) p_2''(z) - p_1''(z) p_2'(z)   \\
 -\mathcal{W}'(z)   \\
 \mathcal{W}(z)
\end{pmatrix} .
\end{align}
Using \eqref{eq:wronskian0}, \eqref{eq:w0w1}, \eqref{eq:p1'p2''-1} and \eqref{eq:p1'p2''-2}, we further have
\begin{align}
\lim_{z \to 0} \Psi_{\alpha}( z)^{-t} \begin{pmatrix}
0 \\ z^{-\alpha} \\0
\end{pmatrix}  = \frac{\sqrt{2\pi}  e^{-\rho^2/6}}{\Gamma(\alpha + 1)} \begin{pmatrix}
1  \\
 \rho  \\
-\alpha
\end{pmatrix}.
\end{align}
A combination of \eqref{def:p0i}, \eqref{eq:Phi-0-0-t} and the above formula gives us that
\begin{multline}
 \begin{pmatrix}
      p_{0,1}(s) \\
      p_{0,2}(s) \\
    p_{0,3}(s)
    \end{pmatrix}  =  \frac{\sqrt{6 \pi}  e^{-\rho^2/6} }{ i(1-\gamma)^{1/3} c(\alpha)}   C_\Psi^{-t} \diag\left(s^{-\frac13},1, s^{\frac13}\right)R(0)^{-t} E_0(0)^{-t}
    \\
      \times \diag(s^{\frac{1}{3}}, 1,  s^{-\frac{1}{3}}) \begin{pmatrix}
1   \\
 \rho   \\
 -\alpha
\end{pmatrix},
\end{multline}
where $c(\alpha)$ is given in \eqref{def:c-alpha}. By inserting \eqref{eq:hatp} into the above formula, we see from \eqref{eq:R-R-hat} that
\begin{multline}\label{eq:hatp0exac}
\begin{pmatrix}
      \hat{p}_{0,1}(s) \\
      \hat{p}_{0,2}(s) \\
    \hat{p}_{0,3}(s)
    \end{pmatrix} =  \frac{\sqrt{6 \pi}  \, e^{-\rho^2/6} }{ i(1-\gamma)^{1/3} c(\alpha)} \diag\left(s^{-\frac13},1, s^{\frac13}\right) \widehat{R}(0)^{-t} C_N^t E_0(0)^{-t}
    \\
    \times \diag(s^{\frac{1}{3}}, 1,  s^{-\frac{1}{3}})
    \begin{pmatrix}
   1   \\
 \rho   \\
 -\alpha
\end{pmatrix}.
\end{multline}
Since
\begin{align*}
\widehat{R}(0)^{-t} & = \left( I+ \frac{\widehat{R}_1(0)}{s^{1/3}} + \frac{\widehat{R}_2(0)}{s^{2/3}} + \Boh(s^{-1})\right)^{-t}
\\
& = I - \frac{\widehat{R}_1(0)^{t}}{s^{1/3}} + \frac{( \widehat{R}_1(0)^2-\widehat{R}_2(0))^t}{s^{2/3}} + \Boh(s^{-1}),\quad s\to +\infty,
\end{align*}
it follows from the explicit formula of $E_0(0)$ in \eqref{eq:E0zero} that
\begin{align}\label{eq:asyhatpointer}
  & \diag \left(s^{-\frac13},1, s^{\frac13}\right) \widehat{R}(0)^{-t} C_N^t  E_0(0)^{-t}  \diag(s^{\frac{1}{3}}, 1,  s^{-\frac{1}{3}}) \nonumber \\
  & =  \frac{i}{\sqrt{3}}   \diag\left(s^{-\frac13},1, s^{\frac13}\right) \Big( C_N^{-1} -\frac{\widehat{R}_1(0)^t C_N^{-1}}{s^{1/3}}  + \frac{(\widehat{R}_1(0)^2 - \widehat{R}_2(0))^t C_N^{-1}}{s^{2/3}} + \Boh(s^{-1}) \Big) \nonumber \\
  & \quad \times \diag\left(s^{\frac13},1, s^{-\frac13}\right) = \frac{i}{\sqrt{3}} \mathcal{P}(s),
  \end{align}
where
\begin{align*}
\left(\mathcal{P}(s)\right)_{11}&=1 -\big( \widehat{R}_1(0) \big)_{11}s^{-\frac{1}{3}}+\big( \widehat{R}_1(0)^2 - \widehat{R}_2(0) \big)_{11} s^{-\frac{2}{3}}+\Boh(s^{-1}),
\\
(\mathcal{P}(s))_{12}&=\sqrt{3} \beta i s^{-\frac{1}{3}}  -  \big(\sqrt{3} \beta i \big( \widehat{R}_1(0) \big)_{11} +\big( \widehat{R}_1(0) \big)_{21}\big) s^{-\frac{2}{3}}+\Boh(s^{-1}),
\\
(\mathcal{P}(s))_{13}&=- \left( \frac{3}{2} \beta^2 +  \frac{\sqrt{3} \beta i }{2} \right) s^{-\frac{2}{3}} + \Boh(s^{-1} ),
\\
(\mathcal{P}(s))_{21}&=- \big( \widehat{R}_1(0) \big)_{12}  +   \big( \widehat{R}_1(0)^2 - \widehat{R}_2(0) \big)_{12} s^{-\frac{1}{3}}+\Boh(s^{-\frac{2}{3}}),
\\
(\mathcal{P}(s))_{22}&=1  -\big(\sqrt{3} \beta i \big( \widehat{R}_1(0) \big)_{12} +\big( \widehat{R}_1(0) \big)_{22}\big) s^{-\frac{1}{3}}+  \Boh(s^{-\frac{2}{3}}),
\\
(\mathcal{P}(s))_{23}&= \sqrt{3} \beta i s^{-\frac{1}{3}}  + \Boh(s^{-\frac{2}{3}}),
\\
(\mathcal{P}(s))_{31}&= - \big( \widehat{R}_1(0) \big)_{13} s^{\frac{1}{3}}+ \big( \widehat{R}_1(0)^2 - \widehat{R}_2(0) \big)_{13} + \Boh(s^{-\frac{1}{3}}),
\\
(\mathcal{P}(s))_{32}&=- \sqrt{3} \beta i \big( \widehat{R}_1(0) \big)_{13}  - \big( \widehat{R}_1(0) \big)_{23}  + \Boh(s^{-\frac{1}{3}}),
\\
(\mathcal{P}(s))_{33}&=1 +  \Boh(s^{-\frac{1}{3}}).
\end{align*}
By further evaluating each entry in $\mathcal P$ with the aid of explicit formulas of $\widehat{R}_1(0)$, $ \widehat{R}_1(0)^2 $ and $\big(\widehat{R}_2(0))_{13}$ in \eqref{eq:R10expl}, \eqref{eq:R10square} and \eqref{eq:R2013}, we then obtain \eqref{eq:p01-asy-alpha}--\eqref{eq:p03-asy-alpha} from \eqref{eq:hatp0exac} and \eqref{eq:asyhatpointer}.


\subsection*{Proofs of \eqref{eq:q11-asy}--\eqref{eq:p13-asy}}
From the definitions of $p_{1,k} $ and $q_{1,k}$, $k=1,2,3$, in \eqref{def: q-1}, one needs to study $\Phi_0^{(1)}(s)$ which appears in the local behavior of $\Phi(z)$ as $z \to s$. By \eqref{eq: Phi-expand-s}, it follows that
\begin{equation}\label{eq:phi01s}
\Phi_0^{(1)}(s) = \lim_{z \to s, \atop z\in \mathtt{II}} \Phi(z)  \begin{pmatrix} 	1 & \frac{\gamma}{2\pi i} \ln(z-s) &0 \\
	0 & 1 & 0 \\
	0& 0 & 1
	\end{pmatrix}.
\end{equation}
To evaluate the above limit, we trace back the transformations $\Phi \mapsto T \mapsto S \mapsto R $ in \eqref{def:PhiToT}, \eqref{def:TtoS} \eqref{def:StoR}, and obtain
\begin{equation}\label{eq:phisz1}
\Phi(sz)=\frac{i}{\sqrt{3}}s^{-\frac{\alpha}{3}} C_\Psi \diag\left(s^{\frac13},1, s^{-\frac13}\right)R(z)P^{(1)}(z) e^{\Theta(sz)}, \qquad z\in D(1, \delta) \cap \mathtt{II}.
\end{equation}
In view of $P^{(1)}(z)$ in \eqref{eq:P1Solution}, we see from \eqref{eq:phi01s}, \eqref{eq:phisz1} and the fact $\theta_1(z)+ \theta_2(z) + \theta_3(z) = 0$ that
\begin{align*}
\Phi_0^{(1)}(s)& = \frac{i s^{-\frac{\alpha}{3}}}{\sqrt{3}}  C_\Psi \diag\left(s^{\frac13},1, s^{-\frac13}\right)R(1) E_1(1) \nonumber \\
& \quad \times \lim_{z \to 1, \atop z\in \mathtt{II}} \left[ e^{-\frac{1}{2}\theta_3(sz)} \diag \left( \Phi^{(\CHF)}(s^{\frac23}f(z))e^{-\frac\beta2 \pi i\sigma_3},e^{\frac{3}{2}\theta_3(sz)} \right)\begin{pmatrix}
	1 & \frac{\gamma}{2\pi i} \ln(sz-s) &0 \\
	0 & 1 & 0 \\
	0& 0 & 1
	\end{pmatrix} \right],
\end{align*}
where we have made use of fact that $R(z)$ and $E_1(z)$ are analytic at $z = 1$. Since $f(z)$ is analytic near $z=1$ and $f(1) = 0$ (see \eqref{eq:f'}),
we have from the local behavior of $\Phi^{(\CHF)}(z)$ near $z = 0$  in \eqref{eq:H-expand-2} that
\begin{align}
\Phi_0^{(1)}(s) &= \frac{i s^{-\alpha/3}}{\sqrt{3}} e^{-\frac{1}{2}\theta_3(s)}  C_\Psi \diag\left(s^{\frac13},1, s^{-\frac13}\right)R(1) E_1(1) \diag \left( \Upsilon_0 ,1 \right)  \nonumber \\
& \quad \times \begin{pmatrix}
	1 & \frac{\gamma}{2\pi i} \left( \ln s  - \ln (e^{-\frac{\pi i}{2}} s^{\frac{2}{3}} f'(1)) \right) & 0\\
	0 & 1 & 0 \\
	0& 0 & e^{\frac{3}{2}\theta_3(s)}
	\end{pmatrix} , \label{eq:Phi-0-1}
\end{align}
where $\Upsilon_0$ is the constant matrix given in \eqref{eq:H-expand-coeff-0}. We are now ready to derive the asymptotics of $\hat q_{1,k}$ and $\hat p_{1,k}$, $k=1,2,3$.

To show \eqref{eq:q11-asy}--\eqref{eq:q13-asy}, we see from \eqref{eq:Phi-0-1}, the definitions of $q_{1,k}$ in \eqref{def: q-1}, $R(z)$ in \eqref{eq:R-R-hat} and $\Upsilon_0$  in \eqref{eq:H-expand-coeff-0} that
 \begin{align}\label{eq:q1-expand}
\begin{pmatrix}
      q_{1,1}(s) \\
      q_{1,2}(s) \\
    q_{1,3}(s)
    \end{pmatrix}& =\Phi_0^{(1)}(s) \begin{pmatrix}
      1\\
    0 \\
   0
    \end{pmatrix}
    \nonumber \\
  &=\frac{i}{\sqrt{3}} s^{-\frac{\alpha}{3}}e^{-\frac{1}{2}\theta_3(s)} C_\Psi \diag\left(s^{\frac13},1, s^{-\frac13}\right) C_N \widehat{R}(1) C_N^{-1}E_1(1)\begin{pmatrix}
   \Gamma(1-\beta)e^{-\beta\pi i}  \\
      \Gamma(1+\beta) \\
0
    \end{pmatrix} .
\end{align}
On account of $E_1(1)$ in \eqref{eq:E-expand-coeff-1}, we note that
\begin{equation}
E_1(1)\begin{pmatrix}
   \Gamma(1-\beta)e^{-\beta\pi i}  \\
      \Gamma(1+\beta) \\
0
    \end{pmatrix} = e^{-\frac{2\beta\pi i}{3}}  C_N \begin{pmatrix}
\omega \Gamma(1-\beta)c_1(s)+\omega^2 \Gamma(1+\beta)  c_1(s)^{-1} \vspace{4pt}  \\
 \Gamma(1-\beta)c_1(s)+\Gamma(1+\beta)  c_1(s)^{-1} \vspace{4pt}  \\
\omega^2 \Gamma(1-\beta)c_1(s)+\omega \Gamma(1+\beta)  c_1(s)^{-1} \vspace{4pt}  \\    \end{pmatrix},
\end{equation}
with $c_1(s)$  given in \eqref{eq:c-1}. As $\Re \beta = 0$ and $\Re( \theta_1(s) - \theta_2(s) )= 0$, each entry in the above column matrix is a sum of two complex conjugates. This observation gives us
 \begin{equation}
E_1(1)\begin{pmatrix}
   \Gamma(1-\beta)e^{-\beta\pi i}  \\
      \Gamma(1+\beta) \\
0
    \end{pmatrix} =
2e^{-\frac{2}{3} \beta\pi i} C_N\Re \begin{pmatrix}
\omega \Gamma(1-\beta)c_1(s)\vspace{4pt}  \\
 \Gamma(1-\beta)c_1(s) \vspace{4pt}  \\
\omega^2 \Gamma(1-\beta)c_1(s)\vspace{4pt}  \\    \end{pmatrix}.
\end{equation}
Inserting the above formula into \eqref{eq:q1-expand}, it is readily seen from \eqref{eq:hatq} that
\begin{equation}
\begin{pmatrix}
      \hat{q}_{1,1}(s) \\
      \hat{q}_{1,2}(s) \\
    \hat{q}_{1,3}(s)
    \end{pmatrix} = \frac{i2e^ {-(\frac{1}{2}\theta_3(s)+\frac{2}{3} \beta\pi i)} } {\sqrt{3}} s^{-\frac{\alpha}{3}}  \diag\left(s^{\frac13},1, s^{-\frac13}\right)\widehat{R}(1) \, \Re \begin{pmatrix}
\omega \Gamma(1-\beta)c_1(s)\vspace{4pt}  \\
 \Gamma(1-\beta)c_1(s) \vspace{4pt}  \\
\omega^2 \Gamma(1-\beta)c_1(s)\vspace{4pt}  \\    \end{pmatrix}.
\end{equation}
The asymptotic formulas \eqref{eq:q11-asy}--\eqref{eq:q13-asy} then follow from the above formula, \eqref{eq:c1theta}, \eqref{eq:estvtheta} and asymptotics of $\widehat{R}(z)$ in \eqref{eq:R-R-hat}.

The asymptotics of $\hat p_{1,k}(s)$, $k=1,2,3$, in \eqref{eq:p11-asy}--\eqref{eq:p13-asy} can be derived in a similar manner.
We have from \eqref{def: q-1}, \eqref{eq:R-R-hat} and \eqref{eq:Phi-0-1} that
\begin{align}\label{eq:pasym1}
& \begin{pmatrix}
  p_{1,1}(s) \\
      p_{1,2}(s) \\
    p_{1,3}(s)
    \end{pmatrix} =-\frac{\gamma}{2\pi i}\Phi_0^{(1)}(s)^{-t}\begin{pmatrix}
      0\\
    1 \\
   0
    \end{pmatrix} \nonumber \\
     & = \frac{\sqrt{3}\gamma  e^{\theta_3(s)/2}}{2\pi } s^{\frac{\alpha}{3}} C_\Psi^{-t}  \diag\left(s^{-\frac13},1, s^{\frac13}\right) C_N^{-t} \widehat{R}(1)^{-t} C_N^tE_1(1)^{-t} \diag \left( \Upsilon_0^{-t} ,1 \right) \begin{pmatrix}
 0 \\
     1 \\
   0
    \end{pmatrix}.
\end{align}
Again, using $E_1(1)$ and $\Upsilon_0$  in \eqref{eq:E-expand-coeff-1} and \eqref{eq:H-expand-coeff-0}, it follows that
\begin{multline}\label{eq:pasym2}
E_1(1)^{-t} \diag \left( \Upsilon_0^{-t} ,1 \right) \begin{pmatrix}
 0 \\
     1 \\
   0
    \end{pmatrix}
    \\
= \frac{1}{3}e^{-\frac{1}{3}\beta\pi i} C_N^{-t} \begin{pmatrix}
-e^{-\frac{\pi }{3} i} \Gamma(1-\beta)c_1(s)+e^{\frac{\pi }{3}i} \Gamma(1+\beta)  c_1(s)^{-1}   \\
 \Gamma(1-\beta)c_1(s)-\Gamma(1+\beta)  c_1(s)^{-1}   \\
-e^{\frac{\pi }{3} i}  \Gamma(1-\beta)c_1(s)+e^{-\frac{\pi }{3} i}  \Gamma(1+\beta)  c_1(s)^{-1}  \end{pmatrix}.
\end{multline}
Inserting \eqref{eq:pasym2} into \eqref{eq:pasym1}, we obtain from \eqref{eq:hatp} that
\begin{equation}
\begin{pmatrix}
      \hat{p}_{1,1}(s) \\
      \hat{p}_{1,2}(s) \\
    \hat{p}_{1,3}(s)
    \end{pmatrix} =  - \frac{\gamma e^ {\frac{1}{2}\theta_3(s)-\frac{1}{3} \beta\pi i} } {\sqrt{3}\pi i }  s^{\frac{\alpha}{3}}   \diag\left(s^{-\frac13},1, s^{\frac13}\right) \widehat{R}(1)^{-t}  \, \Im  \begin{pmatrix}
-e^{-\frac{\pi }{3} i} \Gamma(1-\beta)c_1(s) \vspace{4pt}  \\
 \Gamma(1-\beta)c_1(s) \vspace{4pt}  \\
-e^{\frac{\pi }{3} i}  \Gamma(1-\beta)c_1(s) \vspace{4pt}  \\    \end{pmatrix}.
\end{equation}
A combination of \eqref{eq:c1theta}, \eqref{eq:estvtheta}, \eqref{eq:R-R-hat} and the above formula gives us \eqref{eq:p11-asy}--\eqref{eq:p13-asy}.

\subsection*{Proof of \eqref{eq: pk-qk-s=0}}
By \eqref{eq:phi01s}, we obtain from \eqref{def: q-1}, \eqref{def:tildeR}, \eqref{eq:esttildeR} and \eqref{eq:tildeP0} that, as $s\to 0^+$,
 \begin{align}\label{eq:q11zero}
\begin{pmatrix}
q_{1,1}(s)
\\
q_{1,2}(s)
\\
q_{1,3}(s)
\end{pmatrix}& =\Phi_0^{(1)}(s) \begin{pmatrix}
      1\\
    0 \\
   0
    \end{pmatrix}=\lim_{z \to s, \atop z\in \mathtt{II}} \Phi(z)  \begin{pmatrix} 	1 & \frac{\gamma}{2\pi i} \ln(z-s) &0 \\
	0 & 1 & 0 \\
	0& 0 & 1
	\end{pmatrix}\begin{pmatrix}
      1\\
    0 \\
   0
    \end{pmatrix}
    \nonumber \\
  &=\lim_{z \to s, \atop z\in \mathtt{II}} \widetilde R(z) \widetilde P^{(0)}(z)
  \begin{pmatrix}
      1\\
    0 \\
   0
    \end{pmatrix}=(I+\Boh(s^{\alpha+1}))(\Psi_\alpha^{(0)}(0)+\Boh(s))\begin{pmatrix}
      1\\
    0 \\
   0
    \end{pmatrix}
    \nonumber
    \\
    &=\frac{e^{\rho^2/6}}{\sqrt{2\pi }}(I+\Boh(s^{\alpha+1}))\begin{pmatrix}
      p_2(0)+\Boh(s)\\
    p_2'(0)+\Boh(s) \\
   p_2''(0)+\Boh(s)
    \end{pmatrix},
\end{align}
where $p_2(z)$ defined in \eqref{def:pkint} is an entire function. Hence, it is readily seen that, as $s\to 0^+$,
\begin{equation}
q_{1,k}(s)=\Boh(1), \qquad k=1,2,3.
\end{equation}
The other estimates in \eqref{eq: pk-qk-s=0} can be proved in a similar manner, and we omit the details here.

This completes the proof of Proposition \ref{prop:asypq}. \qed


\section{Proofs of main results}\label{sec:proofmainresults}
\subsection{Proof of Theorem \ref{thm:integralRep}}\label{sec:proofintegral}
By \eqref{eq:derivativeins} and \eqref{eq:H-F}, it follows that
\begin{equation}\label{eq:dF-H}
\frac{\ud}{\ud s} F(s;\gamma,\rho)=H(s;\gamma,\rho), \qquad s>0.
\end{equation}
We then obtain the integral representation of $F$ in \eqref{eq: integralRep} after integrating \eqref{eq:dF-H}.

To show the asymptotics of $H(s)$, we make use of \eqref{eq:H-F}, which establishes a relation between $H(s)$ and the function $\Phi_1^{(1)}(s)$ in  \eqref{eq: Phi-expand-s}. It is then easily seen from \eqref{eq:X-near-s} that
\begin{align}
H(s) & = -\frac{\gamma}{2\pi i} \Big(\Phi_1^{(1)}(s) \Big)_{21} \nonumber \\
&= -\frac{\gamma}{2\pi s i} \begin{pmatrix}
      0 & 1 & 0
    \end{pmatrix} \Phi_0^{(1)}(s)^{-1}  \lim_{z \to 1, \atop z\in \mathtt{II}} \Biggl[ \Phi(sz)  \begin{pmatrix}
	1 & \frac{\gamma}{2\pi i} \ln(sz-s) &0\\
	0 & 1 & 0 \\
	0& 0 & 1
	\end{pmatrix} \Biggr]' \begin{pmatrix}
      1\\
    0 \\
   0
    \end{pmatrix}, \label{Phi-10-formula}
\end{align}
where $'$ denotes the derivative with respect to $z$.

The small-$s$ asymptotics of $H(s)$ in \eqref{eq:Hasy0} then follows directly  from  \eqref{eq:tildeP0}, \eqref{def:tildeR}, \eqref{eq:esttildeR} and the above formula.

The large-$s$ asymptotics of $H(s)$ is more involved. We start with the calculation of the product outside the limit in \eqref{Phi-10-formula}.  From \eqref{eq:Phi-0-1}, we have
\begin{align}
&\begin{pmatrix}
      0 & 1 & 0
    \end{pmatrix} \Phi_0^{(1)}(s)^{-1} \nonumber \\
    &  = -\sqrt{3}i e^ {\frac{1}{2}\theta_3(s)}  s^{\frac{\alpha}{3}}
    \begin{pmatrix}
      0 & 1 & 0
    \end{pmatrix} \diag \left( \Upsilon_0^{-1} ,1 \right) E_1(1)^{-1} R(1)^{-1}    \diag\left(s^{-\frac13},1, s^{\frac13}\right)  C_\Psi^{-1}.
\end{align}
For the limit in \eqref{Phi-10-formula}, we observe from \eqref{eq:P1Solution} and \eqref{eq:phisz1} that
\begin{align}
& -\sqrt{3}i   s^{\frac{\alpha}{3}}   C_\Psi^{-1} \diag\left(s^{-\frac13},1, s^{\frac13}\right)
 \lim_{z \to 1, \atop z\in \mathtt{II}} \Biggl[ \Phi(sz)  \begin{pmatrix}
	1 & \frac{\gamma}{2\pi i} \ln(sz-s) & 0\\
	0 & 1 & 0 \\
	0& 0 & 1
	\end{pmatrix} \Biggr]' \begin{pmatrix}
      1\\
    0 \\
   0
    \end{pmatrix}
    \nonumber  \\
& =
\lim_{z \to 1, \atop z\in \mathtt{II}} \Biggl[ R(z) E_1(z) e^{-\frac{1}{2}\theta_3(sz)} \diag \left( \Phi^{(\CHF)}(s^{\frac23}f(z))e^{-\frac\beta2 \pi i\sigma_3}, 1 \right)
\nonumber \\
& \quad \times
\begin{pmatrix}
	1 & \frac{\gamma}{2\pi i} \ln(sz-s) & 0\\
	0 & 1 & 0 \\
	0& 0 & e^{\frac{3}{2} \theta_3(s)}
	\end{pmatrix}  \begin{pmatrix}
      1\\
    0 \\
   0
    \end{pmatrix} \Biggr] '   \nonumber \\
    & =
\lim_{z \to 1, \atop z\in \mathtt{II}} \Biggl[ R(z) E_1(z) e^{-\frac{1}{2}\theta_3(sz)} \diag \left( \Phi^{(\CHF)}(s^{\frac23}f(z))e^{-\frac\beta2 \pi i\sigma_3}, 1  \right) \begin{pmatrix}
      1\\
    0 \\
   0
    \end{pmatrix} \Biggr] '  .
\end{align}
By \eqref{eq:f'}, we can calculate the above limit with the aid of \eqref{eq:H-expand-2} and the result reads
\begin{multline}
e^{-\frac{1}{2}\theta_3(s)} \big(   R'(1) E_1(1) \diag \left( \Upsilon_0 , 1 \right) + R(1) E_1'(1) \diag \left( \Upsilon_0 , 1 \right)
\\
+ s^{\frac23} f'(1) R(1) E_1(1) \diag \left(  \Upsilon_0 \Upsilon_1 , 0 \right)  -\frac{s \theta_3'(s)}{2}  R(1) E_1(1) \diag \left( \Upsilon_0 , 1 \right)  \big) \begin{pmatrix}
      1\\
    0 \\
   0
    \end{pmatrix}. \label{Phi-10-limit-final}
\end{multline}
A combination of \eqref{Phi-10-formula}--\eqref{Phi-10-limit-final} gives us
\begin{align}
H(s) &= -\frac{\gamma}{2 \pi s i} \Big( \diag \left( \Upsilon_0^{-1} , 1 \right) E_1(1)^{-1} R(1)^{-1} R'(1) E_1(1) \diag \left( \Upsilon_0 ,  1 \right)
\nonumber \\
 & \quad + \diag \left( \Upsilon_0^{-1} ,1 \right) E_1(1)^{-1} E_1'(1) \diag \left( \Upsilon_0 ,  1 \right) +   s^{\frac23} f'(1)  \diag \left( \Upsilon_1,  0\right) -\frac{s \theta_3'(s)}{2} I\Big)_{21}. \label{H-final-form}
\end{align}
Next, we estimate the quantities contained in the above round bracket term by term.  We note from \eqref{R-large-s-exp} that $R(1)=I+\Boh(s^{-1/3})$ and $R'(1)=\Boh(s^{-1/3})$. As $\Re \beta = 0$ (see \eqref{beta-def}), it then follows from \eqref{eq:E-expand-coeff-1} that $ E_1(1)=\Boh(1)$. Combining all these estimates, we obtain
\begin{equation}\label{eq:term1}
  \Big(\diag \left( \Upsilon_0^{-1} ,1 \right) E_1(1)^{-1} R(1)^{-1} R'(1) E_1(1) \diag \left( \Upsilon_0 ,1 \right) \Big)_{21} = \Boh(s^{-\frac13}).
\end{equation}
For the second term in the bracket, we have from the explicit expression of $\Upsilon_0$ in \eqref{eq:H-expand-coeff-0} that
\begin{align}\label{eq:term2}
&\Big( \diag \left( \Upsilon_0^{-1} ,1 \right) E_1(1)^{-1} E_1'(1) \diag \left( \Upsilon_0 ,1 \right) \Big)_{21} \nonumber \\ &=\Gamma(1+\beta)\Gamma(1-\beta)e^{-\beta \pi i} \left( \left(E_1(1)^{-1}E_1'(1)\right)_{22}-\left(E_1(1)^{-1}E_1'(1)\right)_{11} \right) \nonumber\\
& \quad  -\Gamma(1+\beta)^2 \left( E_1(1)^{-1}E_1'(1) \right)_{12} + \Gamma(1-\beta)^2e^{-2\beta \pi i} \left( E_1(1)^{-1}E_1'(1) \right)_{21}.
\end{align}
In view of \eqref{eq:f'}, \eqref{eq:c1theta}, \eqref{eq:Gamma-beta-relation} and  $ E_1(1)^{-1} E_1'(1)$ in \eqref{eq:E-expand-coeff-2}, it follows from straightforward calculations that
\begin{align}
   &\Gamma(1+\beta)\Gamma(1-\beta) e^{-\beta \pi i} \left(\left(E_1(1)^{-1}E_1'(1)\right)_{22}-\left(E_1(1)^{-1}E_1'(1)\right)_{11} \right) \nonumber \\
   &= -\frac{\beta^2 \pi}{\sin(\beta \pi)} e^{-\beta \pi i} \left(\frac{f''(1)}{f'(1)}  + 1  \right) = -\frac{\beta^2 \pi}{\sin(\beta \pi)} e^{-\beta \pi i} \left(\frac{2}{3} + \Boh(s^{-\frac13}) \right), \label{H:derivation-1}
\end{align}
and
\begin{align}
  & -\Gamma(1+\beta)^2 \left( E_1(1)^{-1}E_1'(1) \right)_{12} + \Gamma(1-\beta)^2e^{-2\beta \pi i} \left( E_1(1)^{-1}E_1'(1) \right)_{21} \nonumber \\
  &= -\frac{\sqrt{3}i } {  9 } e^{-\beta \pi i}  \Big( \Gamma(1+\beta)^2 c_1(s)^{-2} + \Gamma(1-\beta)^2c_1(s)^2  \Big) \nonumber \\
  &= -\frac{2\sqrt{3} i}{ 9  }\frac{ \beta \pi   }{\sin(\beta \pi )} e^{-\beta \pi i} \cos (2\vartheta(s) ) , \label{H:derivation-2}
\end{align}
where $\vartheta(s)$ is given in \eqref{def:vtheta}.
For the last two terms in the bracket, we see from \eqref{eq:f'}
and \eqref{eq:H-expand-coeff-1} that
\begin{equation} \label{H:derivation-3}
  \Bigl( s^{\frac23} f'(1)  \diag \left( \Upsilon_1, 0 \right) \Bigr)_{21} = \frac{ \beta \pi i \, e^{-\beta \pi i} }{\sin(\beta \pi )} \left(   \sqrt{3} s^{\frac23} -\frac{\sqrt{3}}{3}\rho s^{\frac13} \right),
\end{equation}
and obviously
\begin{equation}\label{eq:term4}
\left(\frac{s \theta_3'(s)}{2} I \right)_{21} = 0.
\end{equation}
On account of the fact that $\gamma = -2i \sin(\beta \pi) e^{\beta \pi i}$ (see \eqref{beta-def}), we finally obtain the desired large-$s$ asymptotics of $H(s)$ in \eqref{thm:H-asy-infty} by combining \eqref{H-final-form}--\eqref{eq:term4} with \eqref{eq:estvtheta}.

This completes the proof of Theorem \ref{thm:integralRep}. \qed

\subsection{Proof of Theorem \ref{thm:FAsy}}

The proof is based on the differential identities established in Proposition \ref{prop:H-diff}. First, we have from \eqref{eq: integralRep}, \eqref{eq: action-diff} and the small-$s$ behavior of $H(s)$ in \eqref{eq:Hasy0} that
\begin{equation} \label{eq: H-integral-1}
F(s;\gamma,\rho)=\int_0^s H(\tau) \ud \tau = \int_0^s \left( \sum_{i=0}^1 \sum_{k=1}^3p_{i,k}(\tau)q_{i,k}'(\tau)-H(\tau) \right) \ud \tau
+ s H(s) .
\end{equation}
By an easy calculation, one can see that the differential identity with respect to $\gamma$ in \eqref{eq: dH-gamma} still holds if we replace $\gamma$ by $\beta$. Integrating both sides of \eqref{eq: dH-gamma} with respect to $s$, we have
\begin{align} \label{eq: dH-gamma-int}
  \frac{\partial}{\partial \beta} \int_0^s \left( \sum_{i=0}^1\sum_{k=1}^3p_{i,k}(\tau)q_{i,k}'(\tau)-H(\tau) \right) \ud \tau
   =\sum_{i=0}^1\sum_{k=1}^3p_{i,k}(s) \frac{\partial}{\partial \beta}q_{i,k}(s).
\end{align}
Here, we have used the fact that  $$\sum\limits_{i=0}^1\sum\limits_{k=1}^3p_{i,k}(0) \frac{\partial}{\partial \beta}q_{i,k}(0)= 0,$$ which follows from the small-$s$ asymptotics of $p_{i,k}(s),q_{i,k}(s)$, $i=0,1$, $k=1,2,3$, given in \eqref{eq: pk-qk-s=0} and \eqref{eq:q11zero}.

We next rewrite the second integral in \eqref{eq: H-integral-1} by integrating both sides of \eqref{eq: dH-gamma-int} with respect to $\beta$. Since $\gamma = 1 - e^{2 \beta \pi i}$, it follows that $\beta = 0$ if $\gamma=0$. This, together with \eqref{def: q-1} and \eqref{eq:H-F}, implies that $p_{1,k}(s) = 0$ and $H(s) =0$ for $\beta=0$. In addition, from the first equation in \eqref{def:sysdiff}, we also have $q_{0,k}'(s) = 0$ as $p_{1,k}(s) = 0$. These identities particularly give us
\begin{equation}
\sum_{i=0}^1\sum_{k=1}^3p_{i,k}(s)q_{i,k}'(s)-H(s)  \equiv 0, \qquad \beta = 0.
\end{equation}
Integrating both sides of \eqref{eq: dH-gamma-int}, we then obtain from the above formula that
\begin{align} \label{eq: dH-gamma-int-2}
  \int_0^s \left( \sum_{i=0}^1\sum_{k=1}^3p_{i,k}(\tau)q_{i,k}'(\tau)-H(\tau) \right) \ud \tau
   =\int_0^\beta\sum_{i=0}^1\sum_{k=1}^3p_{i,k}(s) \frac{\partial}{\partial \beta'}q_{i,k}(s) \ud \beta'.
\end{align}
A combination of \eqref{eq: dH-gamma-int-2} and \eqref{eq: H-integral-1} yields
\begin{equation} \label{eq: H-integral-2}
F(s;\gamma,\rho) = \int_0^\beta\sum_{i=0}^1\sum_{k=1}^3p_{i,k}(s) \frac{\partial}{\partial \beta'}q_{i,k}(s) \ud \beta' + s H(s) .
\end{equation}
Since the large-$s$ asymptotics of $p_{i,k}(s)$, $q_{i,k}(s)$ and $H(s)$ established in Proposition \ref{prop:asypq} and \eqref{thm:H-asy-infty} are uniform in $\beta$, we are able to obtain the asymptotic expansion of $F(s)$ by estimating each term in the above integrand. To that end, we see from the definitions of $\hat p_{i,k}$ and $\hat q_{i,k}$ in \eqref{eq:hatq} and \eqref{eq:hatp} that
\begin{align}\label{eq:pqtohatpq}
\sum_{k=1}^3p_{i,k}(s) \frac{\partial}{\partial \beta}q_{i,k}(s) &  = \begin{pmatrix}
      \hat{p}_{i,1}(s) \\
      \hat{p}_{i,2}(s) \\
    \hat{p}_{i,3}(s)
    \end{pmatrix}^t \diag\left(s^{\frac13},1, s^{-\frac13}\right)C_N^{-1} \diag\left(s^{-\frac13},1, s^{\frac13}\right) C_\Psi^{-1} \nonumber \\
    & \quad \times \frac{\partial}{\partial \beta}  \left( C_\Psi \diag\left(s^{\frac13},1, s^{-\frac13}\right)C_N \diag\left(s^{-\frac13},1, s^{\frac13}\right) \begin{pmatrix}
      \hat{q}_{i,1}(s) \\
      \hat{q}_{i,2}(s) \\
    \hat{q}_{i,3}(s)
    \end{pmatrix} \right).
\end{align}
Note that only $C_N$ depends on $\beta$, but it is not the case for $C_\Psi$; see their definitions in  \eqref{eq:cons-C-Psi} and \eqref{eq:CN}. By \eqref{eq:pqtohatpq} and the fact that
\begin{equation}
C_N^{-1} \frac{\partial}{\partial \beta} C_N = \begin{pmatrix}
0 & -\sqrt{3} i & \frac{\sqrt{3} }{2} i \\
0 & 0 &  -\sqrt{3} i  \\
0 & 0 & 0
\end{pmatrix},
\end{equation}
it follows
\begin{multline*}
\sum_{k=1}^3p_{i,k}(s) \frac{\partial}{\partial \beta}q_{i,k}(s) = \begin{pmatrix}
      \hat{p}_{i,1}(s) \\
      \hat{p}_{i,2}(s) \\
    \hat{p}_{i,3}(s)
    \end{pmatrix}^t \frac{\partial}{\partial \beta}   \begin{pmatrix}
      \hat{q}_{i,1}(s) \\
      \hat{q}_{i,2}(s) \\
    \hat{q}_{i,3}(s)
    \end{pmatrix}
    \\
    + \begin{pmatrix}
      \hat{p}_{i,1}(s) \\
      \hat{p}_{i,2}(s) \\
    \hat{p}_{i,3}(s)
    \end{pmatrix}^t \begin{pmatrix}
0 & -\sqrt{3} i s^{\frac{1}{3}} & \frac{\sqrt{3} }{2} i s^{\frac{2}{3}} \\
0 & 0 &  -\sqrt{3} i s^{\frac{1}{3}}  \\
0 & 0 & 0
\end{pmatrix}    \begin{pmatrix}
      \hat{q}_{i,1}(s) \\
      \hat{q}_{i,2}(s) \\
    \hat{q}_{i,3}(s)
    \end{pmatrix}.
\end{multline*}
Thus, it is easily seen that
\begin{align}
 &\sum_{i=0}^1\sum_{k=1}^3p_{i,k}(s) \frac{\partial}{\partial \beta}q_{i,k}(s) \nonumber
 \\
 &= \sum_{i=0}^1\sum_{k=1}^3\hat{p}_{i,k}(s) \frac{\partial}{\partial \beta}\hat{q}_{i,k}(s) 
 -\sqrt{3}is^{\frac{1}{3}} \sum_{i=0}^1 \left( \hat{ p}_{i,1}(s)\hat{q}_{i,2}(s)+ \hat{p}_{i,2}(s)\hat{q}_{i,3}(s) \right)
 \nonumber
 \\
 & \quad +\frac{\sqrt{3}}{2}is^{\frac{2}{3}}\sum_{i=0}^1 \hat{ p}_{i,1}(s)\hat{q}_{i,3}(s). \label{eq:pdbetaq}
\end{align}
It comes out that we can calculate the last two terms in the above formula explicitly. Note that the matrix
$ C_\Psi \diag\left(s^{\frac13},1, s^{-\frac13}\right)C_N \diag\left(s^{-\frac13},1, s^{\frac13}\right)$ commutes with the matrix  $A^2$, while
\begin{align}
&\diag\left(s^{\frac13},1, s^{-\frac13}\right)C_N^{-1} \diag\left(s^{-\frac13},1, s^{\frac13}\right) C_\Psi^{-1}  \cdot A \cdot C_\Psi \diag\left(s^{\frac13},1, s^{-\frac13}\right)C_N \diag\left(s^{-\frac13},1, s^{\frac13}\right) \nonumber \\
& = A + \frac{\rho}{3} A^2,
\end{align}
where $A$ is defined in \eqref{def:A}. It is then easily seen from \eqref{eq:hatq}, \eqref{eq:hatp} and the relations \eqref{eq: equationpqk-1}, \eqref{eq: equationpqk-2} that
\begin{align}
& \sum_{i=0}^1 \hat{ p}_{i,1}(s)\hat{q}_{i,3}=\sum_{i=0}^1  p_{i,1}(s)q_{i,3}=1, \label{eq:hatpqk-1} \\
& \sum_{i=0}^1 \Big( \hat{ p}_{i,1}(s)\hat{q}_{i,2}+ \hat{p}_{i,2}(s)\hat{q}_{i,3} \Big)=\sum_{i=0}^1
\Big( p_{i,1}(s)q_{i,2}+ p_{i,2}(s)q_{i,3}\Big) - \frac{\rho}{3} \sum_{i=0}^1 \hat{ p}_{i,1}(s)\hat{q}_{i,3} = \frac{2}{3} \rho. \label{eq:hatpqk-2}
\end{align}
Substituting the above two formulas into \eqref{eq:pdbetaq}, we have
\begin{equation}
\sum_{i=0}^1\sum_{k=1}^3p_{i,k}(s) \frac{\partial}{\partial \beta}q_{i,k}(s) = \sum_{i=0}^1\sum_{k=1}^3\hat{p}_{i,k}(s) \frac{\partial}{\partial \beta}\hat{q}_{i,k}(s) - \frac{2 \sqrt{3}}{3} \rho i s^{\frac{1}{3}} + \frac{\sqrt{3}}{2} i s^{\frac{2}{3}}. \label{eq: pkqk-pkqkhat-relation}
\end{equation}

We now compute large-$s$ asymptotics of $\hat{p}_{i,k}(s) \partial\hat{q}_{i,k}(s)/\partial \beta$. By rewriting $\partial \hat q_{1,k}(s)/\partial \beta $ as $\hat q_{1,k}(s) \partial \ln \hat q_{1,k}(s)/\partial \beta$, it follows from \eqref{eq:q11-asy}--\eqref{eq:p13-asy} that
\begin{align}
\hat  p_{1,1}(s)\frac{\partial}{\partial \beta} \hat  q_{1,1}(s)  &= -\frac{\gamma e^{- \beta \pi i}}{3\pi} |\Gamma(1-\beta)^2 | \Big( \sin (2 \psi(s) - \frac{2\pi}{3} ) + \Boh(s^{-\frac{1}{3}})   \Big)
\nonumber\\
&\quad \times \Big(-\frac{2}{3} \pi i + \frac{\partial}{\partial \beta} |\Gamma(1-\beta)| - \tan(\psi(s) - \frac{\pi}{3}) \frac{\partial}{\partial \beta} \psi(s) +  \Boh(s^{-\frac{1}{3}} \ln s) \Big),\label{eq:p11-dq11}
\\
 \hat  p_{1,2}(s)\frac{\partial}{\partial \beta}\hat  q_{1,2}(s)  &= -\frac{\gamma e^{- \beta \pi i}}{3\pi} |\Gamma(1-\beta)^2 |  \Big( \sin (2 \psi(s) ) + \Boh(s^{-\frac{1}{3}})  \Big)
\nonumber
\\
&\quad \times \Big(-\frac{2}{3} \pi i + \frac{\partial}{\partial \beta} |\Gamma(1-\beta)| - \tan(\psi (s) ) \frac{\partial}{\partial \beta} \psi(s)  + \Boh(s^{-\frac{1}{3}} \ln s) \Big),\label{eq:p12-dq12}
\\
 \hat  p_{1,3}(s)\frac{\partial}{\partial \beta}\hat  q_{1,3}(s)  &= -\frac{\gamma e^{- \beta \pi i}}{3\pi} |\Gamma(1-\beta)^2 | \Big( \sin (2 \psi(s) + \frac{2\pi}{3} ) + \Boh(s^{-\frac{1}{3}}) \Big)
\nonumber \\
&\quad \times \Big(-\frac{2}{3} \pi i + \frac{\partial}{\partial \beta} |\Gamma(1-\beta)| - \tan(\psi(s) + \frac{\pi}{3}) \frac{\partial}{\partial \beta} \psi(s)  + \Boh(s^{-\frac{1}{3}} \ln s ) \Big). \label{eq:p13-dq13}
\end{align}
Note that, there is a factor $s^\beta$ in error estimates $\Boh(s^{-\frac{1}{3}} )$ in \eqref{eq:q11-asy}--\eqref{eq:p13-asy}. Then, an additional $\ln s$ factor appears in the error estimates after differentiating with respect to $\beta$ in the above formulas. Adding \eqref{eq:p11-dq11}--\eqref{eq:p13-dq13} together, we have from  \eqref{eq:Gamma-beta-relation} and the facts $1-\gamma = e^{2 \beta \pi i }$,
\begin{align*}
\sin(2 \psi) + \sin(2  \psi - \frac{2\pi}{3}) + \sin(2  \psi + \frac{2\pi}{3}) & =0,\\
\sin( \psi)^2 + \sin( \psi- \frac{\pi}{3})^2 + \sin( \psi + \frac{\pi}{3})^2 & = \frac{3}{2},
\end{align*}
that
\begin{align}
\sum_{k=1}^3 \hat  p_{1,k}(s)\frac{\partial}{\partial \beta}\hat  q_{1,k}(s) & = -2 \beta i \frac{\partial}{\partial \beta} \psi(s)  + \Boh(s^{-\frac{1}{3}} \ln s) \nonumber \\
& = -2\beta i \frac{\partial}{\partial \beta} \arg \Gamma(1-\beta) - 2\beta \left(\frac{2}{3} \ln s + 2 \ln 3 \right)  + \Boh(s^{-\frac{1}{3}} \ln s) . \label{eq:p1q1hat-sum-asy}
\end{align}
Similarly, since $\partial (1- \gamma)^{1/3}/ \partial \beta  =2\pi i(1- \gamma)^{1/3}/3$,
it is readily seen from \eqref{eq:q01-asy-alpha}--\eqref{eq:q03-asy-alpha} and \eqref{eq:p01-asy-alpha}--\eqref{eq:p03-asy-alpha} that
\begin{align}
\hat  p_{0,2}(s)\frac{\partial}{\partial \beta} \hat q_{0,2}(s) & = (\pi_3(\rho) + \rho  + \Boh(s^{-\frac{1}{3}})) \left(\frac{2 \pi i}{3} + \frac{\partial}{\partial \beta} \right) \Big( - \pi_3(\rho) - \frac{\rho}{3}+ \Boh(s^{-\frac{1}{3}})  \Big) \nonumber \\
& = -\frac{2 \pi i}{3} (\pi_3(\rho) + \rho ) (\pi_3(\rho) + \frac{\rho}{3})  + \Boh(s^{-\frac{1}{3}} \ln s) ,
\end{align}
and
\begin{align}
 & \hat p_{0,1}(s)\frac{\partial}{\partial \beta}\hat q_{0,1}(s)  +   \hat p_{0,3}(s)\frac{\partial}{\partial \beta}\hat q_{0,3}(s) \nonumber \\
  &=   (1 + \Boh(s^{-\frac{1}{3}})) \left(\frac{2 \pi i}{3} + \frac{\partial}{\partial \beta} \right) \Big(  -\pi_6(\rho)+\pi_3(\rho)\left(\pi_3(\rho)+\frac{\rho}{3}\right) - \frac{\beta^2}{3}  +\frac{2 \beta i}{3\sqrt{3}}  \cos(2 \psi(s)) + \Boh(s^{-\frac{1}{3}})  \Big)  \nonumber \\
  &  \quad + \Big(- \alpha  + \pi_6(\rho) + \rho \pi_3(\rho)  + \frac{\rho^2}{3}   + \frac{\beta^2}{3} - \frac{ 2 \beta i}{3 \sqrt{3} } \cos(2 \psi(s) )   + \Boh(s^{-\frac{1}{3}})  \Big) \left(\frac{2 \pi i}{3} + \Boh(s^{-\frac{1}{3}}) \right) \nonumber \\
  & =  \frac{2 \pi i}{3}\Big( -\alpha +( \pi_3(\rho) + \rho)  (\pi_3(\rho) + \frac{ \rho}{3}) \Big) + \frac{\partial}{\partial \beta} \Big( \frac{2 \beta i }{3 \sqrt{3}}   \cos(2 \psi(s) )  - \frac{\beta^2}{3}\Big)  +  \Boh(s^{-\frac{1}{3}} \ln s).
\end{align}
Adding the above two formulas together, we obtain
\begin{equation}
\sum_{k=1}^3 \hat  p_{0,k}(s)\frac{\partial}{\partial \beta}\hat  q_{0,k}(s) = -\frac{2 \alpha \pi i}{3} +  \frac{\partial}{\partial \beta} \Big( \frac{2 \beta i }{3 \sqrt{3}}   \cos(2 \psi(s) )  - \frac{\beta^2}{3}\Big)  +  \Boh(s^{-\frac{1}{3}} \ln s).  \label{eq:p0q0hat-sum-asy}
\end{equation}
Inserting \eqref{eq:p1q1hat-sum-asy} and \eqref{eq:p0q0hat-sum-asy} into \eqref{eq: pkqk-pkqkhat-relation}, we have, after integrating with respect to $\beta$ and a straightforward calculation, that
\begin{align}
& \int_0^\beta\sum_{i=0}^1\sum_{k=1}^3p_{i,k}(s) \frac{\partial}{\partial \beta'}q_{i,k}(s) \ud \beta'
\nonumber
\\
& = -2i \int_{0}^\beta \beta' \frac{\partial}{\partial \beta} \arg \Gamma(1-\beta') \ud \beta' - \left(\frac{2}{3} \ln s + 2 \ln 3 \right) \beta^2 \nonumber \\
& \quad  - \frac{2\alpha \beta \pi i}{3}  + \frac{2 \beta i }{3 \sqrt{3}}    \cos(2 \psi(s) ) -  \frac{\beta^2}{3}  - \frac{2 \sqrt{3}}{3} \rho \beta i s^{\frac{1}{3}} + \frac{\sqrt{3}}{2} \beta i s^{\frac{2}{3}} + \Boh(s^{-\frac{1}{3}} \ln s) \nonumber \\
& = \frac{\sqrt{3}}{2} \beta i s^{\frac{2}{3}} - \frac{2 \sqrt{3}}{3} \rho \beta i s^{\frac{1}{3}}   -\frac{2 \beta ^2}{3} \ln s  + \frac{2 \beta i }{3 \sqrt{3}}    \cos(2 \psi(s) ) \nonumber \\
& \quad + \ln \Big(G(1+\beta) G(1-\beta) \Big) + \Big(\frac{2}{3}-2\ln 3 \Big) \beta^2 - \frac{2\alpha \beta \pi i}{3}  + \Boh(s^{-\frac{1}{3}} \ln s),
\end{align}
where $G(z)$ is the Barnes G-function; see a similar integral computation in \cite[Equation (7.50)]{DXZ202}.

Finally, we obtain \eqref{eq: large gap asy} with  the error estimate $\Boh(s^{-\frac{1}{3}} \ln s)$ by substituting the above formula and \eqref{thm:H-asy-infty} into \eqref{eq: H-integral-2}. 
Integrating on both sides of \eqref{thm:H-asy-infty}, we see that the error estimate $\Boh(s^{-\frac{1}{3}} \ln s)$ can be replaced by $\Boh(s^{-\frac{1}{3}} )$ eventually.


This completes the proof of Theorem \ref{thm:FAsy}. \qed

\begin{appendices}

\section{Confluent hypergeometric parametrix}\label{sec:CHF}
The confluent hypergeometric parametrix $\Phi^{(\CHF)}(z)=\Phi^{(\CHF)}(z;\beta)$ with $\beta$ being a parameter is a solution of the following RH problem.

\subsection*{RH problem for $\Phi^{(\CHF)}$}
 \begin{description}
  \item(a)   $\Phi^{(\CHF)}(z)$ is analytic in $\mathbb{C}\setminus \{\cup^6_{j=1}\widehat\Sigma_j\cup\{0\}\}$, where the contours $\widehat\Sigma_j$, $j=1,\ldots,6,$ are indicated in Fig. \ref{fig:jumps-Phi-C}.

  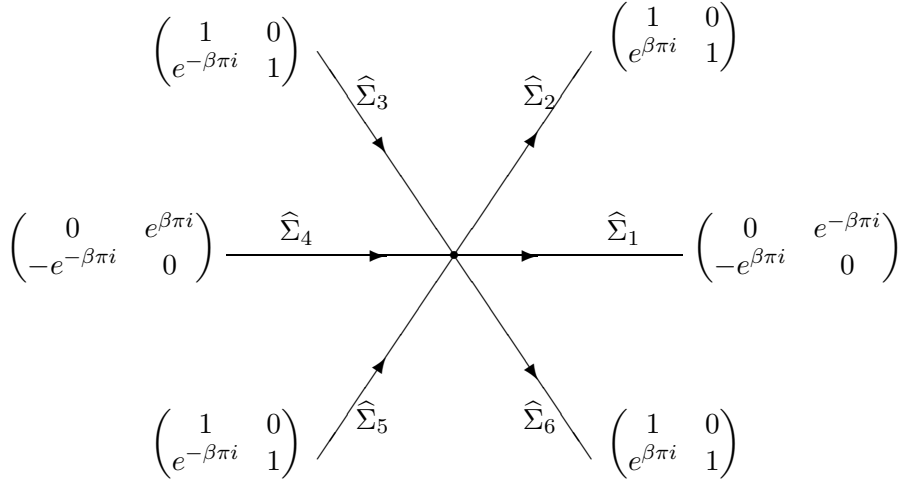
\begin{figure}[h]
\begin{center}
   \setlength{\unitlength}{1truemm}
   \begin{picture}(100,70)(-5,2)
       \put(40,40){\line(-2,-3){18}}
       \put(40,40){\line(-2,3){18}}
       \put(40,40){\line(-1,0){30}}
       \put(40,40){\line(1,0){30}}
  \put(40,40){\line(2,-3){18}}
    \put(40,40){\line(2,3){18}}

       \put(30,55){\thicklines\vector(2,-3){1}}
       \put(30,40){\thicklines\vector(1,0){1}}
       \put(50,40){\thicklines\vector(1,0){1}}
       \put(30,25){\thicklines\vector(2,3){1}}
      \put(50,25){\thicklines\vector(2,-3){1}}
       \put(50,55){\thicklines\vector(2,3){1}}



       \put(27,17){$\widehat\Sigma_5$}
       \put(0,14) {$\begin{pmatrix}
     1 & 0 \\
     e^{- \beta\pi i} & 1
     \end{pmatrix}$}

       \put(27,60){$\widehat\Sigma_3$}
       \put(0,66){$\begin{pmatrix}
    1 & 0 \\
    e^{ -\beta\pi i} & 1
    \end{pmatrix}$}

       \put(17,42){$\widehat\Sigma_4$}
       \put(-19,40){$\begin{pmatrix}
    0 &   e^{\beta\pi i} \\
     -  e^{-\beta\pi i} &  0
     \end{pmatrix}$}

       \put(49,17){$\widehat\Sigma_6$}
       \put(60,14){$\begin{pmatrix}
   1 & 0 \\
   e^{\beta\pi i} & 1
   \end{pmatrix}$}
       \put(49,60){$\widehat \Sigma_2$}
       \put(60,68) {$\begin{pmatrix}
    1 & 0 \\
    e^{ \beta \pi i } & 1
    \end{pmatrix}$}

       \put(60,42){$\widehat \Sigma_1$}
       \put(70,40){ $\begin{pmatrix}
    0 &   e^{-\beta \pi i} \\
    -  e^{\beta \pi i} &  0
    \end{pmatrix}$ }
%

       \put(40,40){\thicklines\circle*{1}}
\end{picture}
   \caption{The jump contours $\widehat \Gamma_i$ and the corresponding jump matrices $\widehat J_i$, $i=1,\ldots,6$,
    in the RH problem for $\Phi^{(\CHF)}$.}
   \label{fig:jumps-Phi-C}
\end{center}
\end{figure}

  \item(b) $\Phi^{(\CHF)}$ satisfies the following jump condition:
  \begin{equation}\label{HJumps}
  \Phi^{(\CHF)}_+(z)=\Phi^{(\CHF)}_-(z) \widehat J_i(z), \quad z \in \widehat\Sigma_i,\quad j=1,\ldots,6,
  \end{equation}
  where the constant matrices $J_i$ are shown in Figure \ref{fig:jumps-Phi-C}.

  \item(c) $\Phi^{(\CHF)}(z)$ satisfies the following asymptotic behavior at infinity:
  \begin{align}\label{H at infinity}
 \Phi^{(\CHF)}(z)=& ~ \left(I + \frac{\Phi^{(\CHF)}_1}{z}+\Boh(z^{-2})\right) z^{-\beta \sigma_3}e^{-\frac{iz}{2}\sigma_3}
  \nonumber
\\
& \times \left\{\begin{array}{ll}
                         I, & ~0< \arg z <\pi,
                         \\
                        \begin{pmatrix}
                                                             0 &   -e^{\beta\pi i} \\
                                                            e^{-\beta\pi i } &  0
                        \end{pmatrix}, &~ \pi< \arg z<\frac{3\pi}{2},
                        \\
                        \begin{pmatrix}
                        0 &   -e^{-\beta\pi i} \\
                        e^{\beta\pi i} &  0
                        \end{pmatrix}, & -\frac{\pi}{2}<\arg z<0,
 \end{array}\right.
\end{align}
where
\begin{equation} \label{Phi1-def}
\Phi^{(\CHF)}_1=
\begin{pmatrix}
\beta^2 i & -\frac{\Gamma(1-\beta)}{\Gamma(\beta)}e^{-\beta \pi i}i
\\
\frac{\Gamma(1+\beta)}{\Gamma(-\beta)}e^{\beta \pi i}i & -\beta^2 i
\end{pmatrix}.
\end{equation}
\item(d) As $z\to 0$, we have $\Phi^{(\CHF)}(z)=\Boh(\ln |z|)$.

\end{description}

From \cite{IK}, it follows that the above RH problem can be solved explicitly in terms of the confluent hypergeometric functions. Moreover, as $z\to 0$, we have (cf. \cite[Equation (A.10)]{DXZ202})
\begin{equation}\label{eq:H-expand-2}
\Phi^{(\CHF)}(z) e^{-\frac{\beta \pi i}{2} \sigma_3} = \Upsilon_0\left( I+ \Upsilon_1z+\Boh(z^2) \right) \begin{pmatrix} 1 & -\frac{\gamma}{2\pi i} \ln (e^{-\frac{\pi i}{2}}z) \\
0 & 1  \end{pmatrix},
\end{equation}
for $z$ belonging to the region bounded by the rays $\widehat \Sigma_2$ and $\widehat \Sigma_3$, where $\gamma=1-e^{2\beta \pi i}$,
\begin{align}\label{eq:H-expand-coeff-0}
\Upsilon_0
=\begin{pmatrix} \Gamma\left(1-\beta\right) e^{-\beta \pi i} &\frac{1}{\Gamma(\beta)} \left( \frac{\Gamma'\left(1-\beta\right)}{\Gamma\left(1-\beta\right)} +2\gamma_{\textrm{E}} \right) \vspace{5pt} \\
\Gamma\left(1+\beta\right) & -\frac{e^{\beta \pi i}}{\Gamma(-\beta)} \left( \frac{\Gamma'\left(-\beta\right)}{\Gamma\left(-\beta\right) } +2\gamma_{\textrm{E}}\right) \end{pmatrix}
\end{align}
with $\gamma_{\textrm{E}}$ being the Euler's constant,
and
 \begin{equation}\label{eq:H-expand-coeff-1}
(\Upsilon_1)_{21}=\frac{ \beta \pi i \, e^{-\beta \pi i} }{\sin(\beta \pi )}.
\end{equation}

\section{Proof of Proposition \ref{prop:pk}} \label{sec:appen-prop}
Since the functions $p_1(z)$ and $p_2(z)$ are entire, it is clear that their Wronskian $\mathcal{W}(z)$ admits an expansion as shown in \eqref{eq:wronskian0}. To establish the first two coefficients therein, we start with the special case that $\rho=0$. In this case, it is easily seen from \eqref{def:pkint} that
\begin{equation}
p_1(z)=\int_{\Gamma_1} t^{\alpha-3} e^{zt+\frac{1}{2t^2}} \ud t, \qquad p_2(-z)=e^{-\alpha \pi i}\int_{\Gamma_2} t^{\alpha-3} e^{-zt+\frac{1}{2t^2}} \ud t,
\end{equation}
where the contours $\Gamma_1$ and $\Gamma_2$ are illustrated in Figure \ref{fig: Gamma-k}. We then obtain the relation $p_1(z)=-p_2(-z)$ after a change of variable $t=e^{\pi i} \tau$ in the above integral representation of $p_2(-z)$. By inserting the expansion of $e^{tz}$ into the integral representation of $p_1(z)$, we find that
\begin{equation}\label{eq:p1expand}
p_1(z)=\sum_{k=0}^{\infty}\int_{\Gamma_1} t^{\alpha+k-3}e^{\frac{1}{2t^2}} \ud t \cdot \frac{z^k}{k!}.
\end{equation}
To evaluate the integral in \eqref{eq:p1expand}, we see after a change of variable $\tau=\frac{1}{2t^2}$ that
\begin{equation}\label{eq:p1Coef-1}
\int_{\Gamma_1} t^{\alpha+k-3}e^{\frac{1}{2t^2}} \ud t=2^{-\frac{\alpha+k}{2}} \int_{\mathcal{H}}  \tau^{-\frac{\alpha+k}{2}}e^{\tau} \ud \tau,
\end{equation}
where $\mathcal{H}$ is the Hankel's loop which starts at $-\infty$, circles the origin once in the positive direction and returns to $-\infty$.
Using the Hankel's loop integral for $1/\Gamma(z)$ (see \cite[Formula 5.9.2]{DLMF}), we have
\begin{equation}\label{eq:p1Coe2}
\int_{\Gamma_1} t^{\alpha+k-3}e^{\frac{1}{2t^2}}\ud t= 2^{-\frac{\alpha+k}{2}} \frac{ 2\pi i} {\Gamma((\alpha+k)/2)}.
 \end{equation}
This, together with \eqref{eq:p1expand}, gives us
\begin{equation}\label{eq:p1p2expand}
p_1(z)=-p_2(-z)=2\pi i\sum_{k=0}^{\infty}2^{-\frac{\alpha+k}{2}}\frac{z^k}{k!\Gamma((\alpha+k)/2))},
\end{equation}
for $\rho=0$. As a consequence, we get
\begin{align}
\mathcal{W}(0;0) &= 2p_1(0)p_1'(0)=-\frac{\pi^22^{-\alpha+5/2}}{\Gamma(\alpha/2)\Gamma((\alpha+1)/2)}
=-\frac{(2\pi )^{3/2}}{\Gamma(\alpha)}, \label{eq:w0} \\
\mathcal{W}'(0;0) &= 0, \label{eq:w0-d}
\end{align}
which agree with those given in \eqref{eq:wronskian0} and \eqref{eq:w0w1} for $\rho = 0$. The third equality in \eqref{eq:w0} follows from the following duplication formula for Gamma function:
\begin{equation*}
\Gamma(2z)=\frac{2^{2z-1}}{\sqrt{\pi}}\Gamma(z)\Gamma\left(z+\frac12\right),\qquad 2z \neq 0,-1,-2,\ldots;
\end{equation*}
cf. \cite[Formula 5.5.5]{DLMF}.

We next deal with the case for general $\rho$. By \eqref{def:pkint}, it is easily seen that
\begin{equation} \label{eq:dp10rho-1}
\frac{\ud}{\ud \rho}p_i^{(k+1)}(0)= p_i^{(k)}(0), \quad  i=1,2,\quad k =0,1,2,\ldots,
\end{equation}
where $p_i^{(k)}$ indicates the $k$-th derivative with respect to $z$. This, together with the definition of $\mathcal{W}$ in \eqref{def:wronskian-def}, implies that
\begin{equation}\label{eq:dwrho}
\frac{\ud}{\ud \rho}\mathcal{W}(0;\rho)=\frac{\ud}{\ud \rho}p_1(0)\cdot p_2'(0)-\frac{\ud}{\ud \rho}p_2(0)\cdot p_1'(0).
\end{equation}
Note that
\begin{align}\label{eq:dp10rho-2}
\frac{\ud}{\ud \rho}p_1(0)&= \int_{\Gamma_1}t^{\alpha-4}e^{\frac{\rho}{t}+\frac{1}{2t^2}}\ud t=- \int_{\Gamma_1}t^{\alpha-1}e^{\frac{\rho}{t}}\frac{\ud}{\ud t}\left(e^{\frac{1}{2t^2}}\right)
\nonumber
\\&=(\alpha-1)p_1'(0)-\rho  p_1(0),
\end{align}
and similarly,
\begin{equation}\label{eq:dp10rho-3}
\frac{\ud}{\ud \rho}p_2(0)=(\alpha-1)p_2'(0)-\rho p_2(0).
\end{equation}
We then obtain from \eqref{eq:dwrho}--\eqref{eq:dp10rho-3} that
\begin{equation}
\frac{\ud}{\ud \rho}\mathcal{W}(0;\rho)=-\rho \mathcal{W}(0;\rho).
\end{equation}
Solving this first order differential equation with the initial condition \eqref{eq:w0} gives us
\begin{equation} \label{eq:w0-value}
\mathcal{W}(0;\rho) = - \frac{(2\pi )^{3/2} }{\Gamma(\alpha) } e^{-\rho^2/2}, \qquad  \alpha > -1.
\end{equation}

The function $\mathcal{W}'(0;\rho)$  satisfies a similar first order differential equation as well. From \eqref{eq:dp10rho-1} and \eqref{eq:dp10rho-2}, we have
\begin{align}
\frac{\ud}{\ud \rho}\mathcal{W}'(0;\rho) &= \frac{\ud}{\ud \rho} p_1(0) \cdot p_2''(0) +  p_1(0) \cdot \frac{\ud}{\ud \rho} p_2''(0) - \frac{\ud}{\ud \rho} p_1''(0) \cdot p_2(0) -  p_1''(0) \cdot \frac{\ud}{\ud \rho} p_2(0) \nonumber \\
&=  - \rho \mathcal{W}'(0;\rho)+\mathcal{W}(0;\rho) + (\alpha - 1) \Big( p_1'(0) p_2''(0) - p_1''(0) p_2'(0)  \Big). \label{eq:dwrho'}
\end{align}
In view of \eqref{eq:p1p2expand} and \eqref{eq:w0-value}, we obtain the  following initial condition  after setting $\rho=0$ in \eqref{eq:dwrho'}
\begin{equation} \label{eq:dw0-value}
\frac{\ud}{\ud \rho}\mathcal{W}'(0;\rho)  \Big |_{\rho=0} = - \frac{(2\pi )^{3/2} }{\Gamma(\alpha+1) }.
\end{equation}
Recall that both $p_1$ and $p_2$ satisfy the equation \eqref{eq:Pearcey1}, by setting $x=0$ in \eqref{eq:Pearcey1},  we  have
\begin{equation}
\alpha p_i''(0) = \rho p_i'(0) + p_i(0), \qquad i = 1,2.
\end{equation}
The above formula gives us that, if $\alpha \neq 0$
\begin{align}
 p_1'(0) p_2''(0) - p_1''(0) p_2'(0) = \frac{1}{\alpha} \Big( p_1'(0) p_2(0) - p_1(0) p_2'(0) \Big) = -\frac{1}{\alpha}  \mathcal{W}(0; \rho), \label{eq:p1'p2''-1}
\end{align}
and if $\alpha=0$,
\begin{align}
p_1'(0) p_2''(0) - p_1''(0) p_2'(0) = - \frac{1}{\rho} \Big( p_1(0) p_2''(0) - p_1''(0) p_2(0) \Big) = -\frac{1}{\rho}  \mathcal{W}'(0; \rho). \label{eq:p1'p2''-2}
\end{align}
Substituting the above two formulas into \eqref{eq:dwrho'}, we obtain the following differential equations for $\mathcal{W}'(0;\rho)$:
\begin{align}
\frac{\ud}{\ud \rho}\mathcal{W}'(0;\rho)  = \begin{cases}
- \rho \mathcal{W}'(0;\rho)+\frac{1}{\alpha}\mathcal{W}(0;\rho) ,  & \textrm{if } \alpha \neq 0, \\
- (\rho - \frac{1}{\rho})  \mathcal{W}'(0;\rho) + \mathcal{W}(0;\rho) , & \textrm{if } \alpha = 0,
\end{cases}
\end{align}
Solving the above equations with $\mathcal{W}(0;\rho)$ in \eqref{eq:w0-value}, the initial conditions \eqref{eq:w0-d} and \eqref{eq:dw0-value}, we obtain
\begin{equation} \label{eq:w0'-value}
\mathcal{W}'(0;\rho) = - \frac{(2\pi )^{3/2} }{\Gamma(\alpha+1) } \rho e^{-\rho^2/2}, \qquad \alpha > -1.
\end{equation}
A combination of \eqref{eq:wronskian0}, \eqref{eq:w0-value} and \eqref{eq:w0'-value} yields \eqref{eq:w0w1}.

To prove \eqref{eq:p-alpha-exp} and \eqref{eq:b0}, we introduce
\begin{equation}
p(z;\alpha)=e^{\alpha \pi i} p_1(z) - p_2(z). \label{def:p-alpha}
\end{equation}
 First, we differentiate \eqref{eq:dp10rho-2} with respect to $\rho$ once again and obtain
\begin{align}
\frac{\ud^2}{\ud \rho^2}p_i(0) &= (\alpha - 1) \frac{\ud}{\ud \rho}p'_i(0)  - p_i(0) - \rho  \frac{\ud}{\ud \rho}p_i(0) \nonumber \\
& = (\alpha - 2) p_i(0)   - \rho  \frac{\ud}{\ud \rho}p_i(0) , \qquad i = 1,2,
\end{align}
where \eqref{eq:dp10rho-1} is used in the last step. This, together with the definition of $p(z;\alpha)$ in \eqref{def:p-alpha}, implies that $p(0;0)$ satisfies the following second order differential equation
\begin{equation} \label{eq:p-alpha-eqn-1}
\frac{\ud^2}{\ud \rho^2}p(0;0) = -2 p(0;0)   - \rho  \frac{\ud}{\ud \rho}p(0;0) .
\end{equation}
Similar to the derivation of \eqref{eq:dp10rho-2}, we also note that
\begin{align}
\frac{\ud}{\ud \rho}p_i'(0)  =  \alpha \, p_i''(0)-\rho \, p_i'(0), \qquad i=1,2. \label{eq:dp10rho-4}
\end{align}
We then obtain from \eqref{eq:dp10rho-1}, \eqref{eq:dp10rho-2} and the above differential relations the following differential equations for $p'(0;0)$ and $p''(0;0)$:
\begin{align}
\frac{\ud^2}{\ud \rho^2}p'(0;0) &= - p'(0;0)   - \rho  \frac{\ud}{\ud \rho}p'(0;0), \label{eq:p-alpha-eqn-2} \\
\frac{\ud^2}{\ud \rho^2}p''(0;0)&= - \rho  \frac{\ud}{\ud \rho}p''(0;0) .  \label{eq:p-alpha-eqn-3}
\end{align}
We next establish the initial conditions for these differential equations. By setting $\alpha = 0$ in \eqref{eq:p1p2expand}, we have, if $\rho=0$,
\begin{equation}
p(0; 0) = 2 p_1(0) =  0, \quad p'(0;0) = 0, \quad p''(0;0) = 2p_1''(0)= 2 \pi i.
\end{equation}
This, together with \eqref{eq:dp10rho-1} and \eqref{eq:dp10rho-2}, implies that, for $\rho=0$,
\begin{equation}
\frac{\ud}{\ud \rho} p(0; 0) =  0, \quad \frac{\ud}{\ud \rho} p'(0;0) = 0, \quad \frac{\ud}{\ud \rho} p''(0;0) =  0.
\end{equation}
Solving the differential equations \eqref{eq:p-alpha-eqn-1}, \eqref{eq:p-alpha-eqn-2} and \eqref{eq:p-alpha-eqn-3} with the above initial conditions, we have
\begin{equation}
p(0; 0) =  0, \quad p'(0;0) = 0, \quad p''(0;0) =  2 \pi i,
\end{equation}
for general $\rho$. Since $p(z;\alpha)$ is an entire function in $z$, the above formulas give us \eqref{eq:p-alpha-exp} and \eqref{eq:b0}.

To show \eqref{eq:qExpand} and \eqref{eq:c0}, we see from \eqref{def:q} that
\begin{equation}\label{eq:q}
q(z)=-\frac{1}{2\sin(\alpha\pi)i}\left(p_1(z)-e^{-\alpha\pi i}p_2(z)-(e^{\alpha \pi i}-e^{-\alpha\pi i})p_3(z)\right).
\end{equation}
We intend to rewrite right-hand side of the above formula as a single integral.  Note that, although the integrands of $p_k(z)$, $k=1,2,3$, in \eqref{def:pkint} appear to be the same, the branches of $t^\alpha$ are chosen differently. In \eqref{eq:q}, we maintain the branches of $t^\alpha$ in the integrals of $p_1(z)$ and  $e^{\alpha \pi i} p_3(z)$, while modify the branches of $t^\alpha$ in the integrals of $e^{-\alpha\pi i}p_2(z)$ and $e^{-\alpha\pi i}p_3(z)$ by making a change of variables $t \to t e^{2\pi i}$, that is,
\begin{align}
e^{-\alpha\pi i} p_2(z) &=  \int_{\Gamma_2} t^{\alpha-3}e^{zt + \frac{\rho}{t}+\frac{1}{2t^2}}\ud t, \quad \textrm{with } -\frac{3 \pi}{2} < \arg t <- \frac{\pi}{2}, \label{eq:p2-new-int} \\
e^{-\alpha\pi i} p_3(z) &= \int_{\Gamma_3} t^{\alpha-3}e^{zt + \frac{\rho}{t}+\frac{1}{2t^2}}\ud t, \quad \textrm{with } -2 \pi < \arg t <- \pi. \label{eq:p3-new-int}
\end{align}
Furthermore, we require that the integral contour $\Gamma_3$ of $e^{\alpha \pi i} p_3(z)$ and $e^{-\alpha\pi i} p_3(z) $ lies in the right/left half plane, respectively. Thus, the argument in \eqref{eq:p3-new-int} is eventually set to be $-3 \pi /2 < \arg t <- \pi$. By \eqref{def:pkint} and \eqref{eq:p2-new-int}--\eqref{eq:p3-new-int}, the contribution near $t=0$ cancels in \eqref{eq:q}. Therefore, we can deform the contour of  integral in \eqref{eq:q} to be a single Hankel loop $\mathcal{H} e^{-\frac{\pi}{2} i}$, i.e., for $\pi/4<\arg z< 3 \pi/4$, we have
\begin{equation}\label{eq:qInt-1}
q(z)=-\frac{1}{2 \sin(\alpha \pi)i}\int_{\mathcal{H} e^{-\frac{\pi}{2} i}} t^{\alpha-3}e^{zt + \frac{\rho}{t}+\frac{1}{2t^2}}\ud t, \quad  \textrm{with } -\frac{3 \pi}{2} < \arg t < \frac{\pi}{2}.
\end{equation}
A further change of variable $zt \to t$ yields
\begin{equation}\label{eq:qInt-2}
q(z)=-\frac{ z^{2-\alpha}}{2 \sin(\alpha \pi)i}\int_{\mathcal{H}} t^{\alpha-3}e^{t+\frac{z\rho}{t}+\frac{z^2}{2t^2}}\ud t.
\end{equation}
Expanding $\exp\left(\frac{z\rho}{t}+\frac{z^2}{2t^2}\right)$ in the above formula, we obtain \eqref{eq:qExpand}
with
\begin{equation}\label{eq:c00}
c_0=-\frac{1}{2 \sin(\alpha \pi)i}\int_{\mathcal{H}} t^{\alpha-3}e^t\ud t =-\frac{\pi }{\sin(\alpha \pi)\Gamma(3-\alpha)},
\end{equation}
as shown in \eqref{eq:c0}.

Finally, to show  \eqref{eq:ddp3}, we see from the integral representation for $p_3(z)$ in \eqref{def:pkint} that
\begin{equation}\label{eq:dp3}
p_3''(z)=e^{-\alpha \pi i}\int_{\Gamma_3} t^{\alpha-1}e^{zt+\frac{\rho}{t}+\frac{1}{2t^2}}\ud t,
\end{equation}
where the integral contour $\Gamma_3$ goes along the imaginary axis from $\infty e^{ \frac{\pi}{2} i }$ to the origin for  $\pi/4<\arg z< 3 \pi/4$. If $\alpha> 0$, we obtain after a change of variable $zt \to t$ that
\begin{equation}\label{eq:dp3Int}
p_3''(z)=e^{-\alpha \pi i}z^{-\alpha} \int_{-\infty}^0 t^{\alpha-1}e^{t+\frac{z \rho}{t}+\frac{z^2}{2t^2}}\ud t.
\end{equation}
As $z\to 0$, it follows that
\begin{equation}\label{eq:dp3expand}
\lim_{z\to0} z^{\alpha}p_3''(z)=e^{-\pi i\alpha} \int_{-\infty}^0 t^{\alpha-1}e^t\ud t=-\Gamma(\alpha),
\end{equation}
as required.

This completes the proof of Proposition \ref{prop:pk}. \hfill \qed

\end{appendices}

\section*{Acknowledgements}
Dan Dai was partially supported by grants from the City University of Hong Kong (Project No. 7005597 and 7005252), and grants from the Research Grants Council of the Hong Kong Special Administrative Region, China (Project No. CityU 11303016 and CityU 11300520). Shuai-Xia Xu was partially supported by National Natural Science Foundation of China under grant numbers 11971492, 11571376 and 11201493, and Natural Science Foundation for Distinguished Young Scholars of Guangdong Province of China  under grant number 2022B1515020063. Lun Zhang was partially supported by National Natural Science Foundation of China under grant numbers 11822104, ``Shuguang Program'' supported by Shanghai Education Development Foundation and Shanghai Municipal Education Commission.


\end{document}